\documentclass[nonatbib]{elsarticle}
\usepackage[utf8]{inputenc}
\usepackage[T1]{fontenc}
\usepackage{lmodern}
\usepackage{microtype}
\usepackage{amssymb}
\usepackage{mathdots}
\usepackage{ccicons}

\usepackage{amsthm}

\usepackage{graphicx}
\makeatletter
\providecommand{\leftsquigarrow}{%
  \mathrel{\mathpalette\reflect@squig\relax}%
}
\newcommand{\reflect@squig}[2]{%
  \reflectbox{$\m@th#1\rightsquigarrow$}%
}
\makeatother

\usepackage{url}
\usepackage{comment}
\usepackage{tikz}

\usepackage[plainpages=false,pdfpagelabels,colorlinks=true,allcolors=blue!80!black,hypertexnames=false,pdftex=true,unicode]{hyperref}

\makeatletter
\let\c@author\relax
\makeatother
\usepackage[style=authoryear, maxnames=50, sorting=nyt, isbn=false, uniquename=false]{biblatex}
\def\cite{\mynote{Citation problem! Use \textbackslash parencite/\textbackslash textcite only.}}

\newtheorem{thm}{Theorem}[section]
\newtheorem{prop}[thm]{Proposition}
\newtheorem{lem}[thm]{Lemma}
\newtheorem{cor}[thm]{Corollary}
\newtheorem{defi}[thm]{Definition}
\newtheorem{hyp}[thm]{Hypothesis}
\newtheorem{notation}[thm]{Notation}

\theoremstyle{definition}
\newtheorem{ex}[thm]{Example}

\theoremstyle{remark}
\newtheorem{rem}[thm]{Remark}

\usepackage{chngcntr}
\counterwithin{equation}{section}

\usepackage{multirow}
\usepackage{amsmath}
\makeatletter
\let\save@mathaccent\mathaccent
\newcommand*\if@single[3]{%
  \setbox0\hbox{${\mathaccent"0362{#1}}^H$}%
  \setbox2\hbox{${\mathaccent"0362{\kern0pt#1}}^H$}%
  \ifdim\ht0=\ht2 #3\else #2\fi
  }
\newcommand*\rel@kern[1]{\kern#1\dimexpr\macc@kerna}
\newcommand*\widebar[1]{\@ifnextchar^{{\wide@bar{#1}{0}}}{\wide@bar{#1}{1}}}
\newcommand*\wide@bar[2]{\if@single{#1}{\wide@bar@{#1}{#2}{1}}{\wide@bar@{#1}{#2}{2}}}
\newcommand*\wide@bar@[3]{%
  \begingroup
  \def\mathaccent##1##2{%
    \let\mathaccent\save@mathaccent
    \if#32 \let\macc@nucleus\first@char \fi
    \setbox\z@\hbox{$\macc@style{\macc@nucleus}_{}$}%
    \setbox\tw@\hbox{$\macc@style{\macc@nucleus}{}_{}$}%
    \dimen@\wd\tw@
    \advance\dimen@-\wd\z@
    \divide\dimen@ 3
    \@tempdima\wd\tw@
    \advance\@tempdima-\scriptspace
    \divide\@tempdima 10
    \advance\dimen@-\@tempdima
    \ifdim\dimen@>\z@ \dimen@0pt\fi
    \rel@kern{0.6}\kern-\dimen@
    \if#31
      \overline{\rel@kern{-0.6}\kern\dimen@\macc@nucleus\rel@kern{0.4}\kern\dimen@}%
      \advance\dimen@0.4\dimexpr\macc@kerna
      \let\final@kern#2%
      \ifdim\dimen@<\z@ \let\final@kern1\fi
      \if\final@kern1 \kern-\dimen@\fi
    \else
      \overline{\rel@kern{-0.6}\kern\dimen@#1}%
    \fi
  }%
  \macc@depth\@ne
  \let\math@bgroup\@empty \let\math@egroup\macc@set@skewchar
  \mathsurround\z@ \frozen@everymath{\mathgroup\macc@group\relax}%
  \macc@set@skewchar\relax
  \let\mathaccentV\macc@nested@a
  \if#31
    \macc@nested@a\relax111{#1}%
  \else
    \def\gobble@till@marker##1\endmarker{}%
    \futurelet\first@char\gobble@till@marker#1\endmarker
    \ifcat\noexpand\first@char A\else
      \def\first@char{}%
    \fi
    \macc@nested@a\relax111{\first@char}%
  \fi
  \endgroup
}
\makeatother

\usepackage{calc}
\usepackage[shortlabels]{enumitem}
\usepackage{moreenum}
\newlist{deepenum}{enumerate}{6}
% ...
\newlist{algospec}{itemize}{1}
\setlist[algospec]{font=\normalfont\itshape,itemsep=0ex,partopsep=0ex}
\usepackage{float}
\restylefloat{table}
\newfloat{algofloat}{tp}{lop}
\floatname{algofloat}{Algorithm}
\newcommand{\inputs}[1]{%
  \begin{algospec}[nosep,align=right,labelwidth=\widthof{Output:}]
    \item[Input:]
    #1%
  }
\newcommand{\outputs}[1]{%
    \item[Output:] #1
  \end{algospec}%
  \rule[.5\baselineskip]{\textwidth}{.05em}%
  \vskip-5pt%
}

\newenvironment{algo}
{
  \setcounter{totalnumber}{1}
  \setcounter{topnumber}{1}
  \begin{algofloat}\footnotesize
  \begin{center}\begin{minipage}{.9\linewidth}%
  \rule{\textwidth}{.08em}
}
{
  \vskip-5pt%
  \rule{\textwidth}{.08em}%
  \end{minipage}\end{center}
  \end{algofloat}
}

\newenvironment{myequation*}{\begin{minipage}{\linewidth}\begin{equation*}}{\end{equation*}\end{minipage}}

\newcommand{\bC}{\mathbb C}

\newcommand{\bK}{\mathbb K}
\newcommand{\bKbar}{\protect\widebar\bK}
\newcommand{\bL}{\mathbb L}
\newcommand{\bN}{\mathbb N}
\newcommand{\bP}{\mathbb P}
\newcommand{\bQ}{\mathbb Q}
\newcommand{\bQbar}{\protect\widebar\bQ}
\newcommand{\bR}{\mathbb R}
\newcommand{\bZ}{\mathbb Z}
\usepackage{mathrsfs}
\def\mathcal{\mathscr}
\newcommand{\cA}{\mathcal A}
\newcommand{\cB}{\mathcal B}
\newcommand{\cC}{\mathcal C}
\newcommand{\cD}{\mathcal D}
\newcommand{\cF}{\mathcal F}
\newcommand{\cN}{\mathcal N}

\newcommand{\cR}{\mathcal R}
\newcommand{\cS}{\mathcal S}
\newcommand{\cT}{\mathcal T}
\newcommand{\cU}{\mathcal U}
\newcommand{\cV}{\mathcal V}

\newcommand{\fD}{\mathfrak D}

\newcommand{\fG}{\mathfrak G}
\newcommand{\fH}{\mathfrak H}
\newcommand{\fI}{\mathfrak I}

\newcommand{\fR}{\mathfrak R}

\newcommand{\fp}{\mathfrak p}
\newcommand{\fq}{\mathfrak q}

\newcommand{\vecty}{Y}

\newcommand{\divides}{\mid}

\newcommand{\pol}[1]{#1[x]}
\newcommand{\polm}[1]{#1[a]}  
\newcommand{\rat}[1]{#1(x)}
\newcommand{\fps}[1]{#1[[x]]}
\newcommand{\fls}[1]{#1((x))}

\newcommand{\rrat}[2]{#1(x^{1/#2})}
\newcommand{\rfps}[2]{#1[[x^{1/#2}]]}
\newcommand{\rfls}[2]{#1((x^{1/#2}))}
\newcommand{\ramrat}[1]{\rrat{#1}{*}}
\newcommand{\puiseux}[1]{\rfls{#1}{*}}

\newcommand{\rflsalt}[2]{#1(((t^{-1})^{1/#2}))}
\newcommand{\puiseuxalt}[1]{\rflsalt{#1}{*}}
\newcommand{\omdr}{\fD} 
\newcommand{\ricsol}{\fR}
\newcommand{\dring}{D} 
\newcommand{\leftkern}{\Pi} 

\newcommand{\projzn}[1]{\bP\bigl(#1\bigr)}
\newcommand{\mult}[2]{\val_{#2}{#1}} 
\DeclareMathOperator*{\irred}{irred}
\DeclareMathOperator*{\forbiddenby}{\cF}
\newcommand{\rfact}[2]{#2^{\widebar{#1}}}
\newcommand{\val}{\operatorname{val}}
\DeclareMathOperator*{\lcm}{lcm}
\newcommand{\rk}{\operatorname{rk}}
\newcommand{\ideal}[1]{\langle #1 \rangle}
\newcommand{\intinv}[2]{[#1,#2]}
\newcommand{\nonz}[1]{{#1}_{\neq0}}

\newcommand{\Bnum}{B_{\mathrm{num}}}
\newcommand{\Bden}{B_{\mathrm{den}}}

\def\trsigma{^{[\sigma]}}
\def\trinfty{^{[\infty]}}

\newcommand{\mylogg}[1]{(\ln x)^{#1}}
\newcommand{\mylog}[1]{\mylogg{\log_b#1}}
\newcommand{\myloggalt}[1]{(\ln(t^{-1}))^{#1}}
\newcommand{\mylogalt}[1]{\myloggalt{\log_b#1}}

\newcommand{\zdim}{{N}}

\usepackage{xspace}
\newcommand{\mif}{monic irreducible\xspace}

\def\timelimitexceeded{\text{>12\,hr}}
\def\memorylimitexceeded{\text{>60\,GB}}
\def\myhash{\text{\tiny\#}}

\makeatletter
\renewcommand\part{%
   \if@noskipsec \leavevmode \fi
   \par
   \addvspace{4ex}%
   \@afterindentfalse
   \secdef\@part\@spart}

\def\thepart{\Roman{part}}

\def\@part[#1]#2{%
    \ifnum \c@secnumdepth >\m@ne
      \refstepcounter{part}%
      \addcontentsline{toc}{part}{\thepart\hspace{1em}#1}%
    \else
      \addcontentsline{toc}{part}{#1}%
    \fi
    {\parindent \z@ \raggedright
     \interlinepenalty \@M
     \normalfont
     \ifnum \c@secnumdepth >\m@ne
       \Large\bfseries \partname\nobreakspace\thepart:
     \fi
     \Large \bfseries #2%
     \par}%
    \nobreak
    \vskip 3ex
    \@afterheading}
\def\@spart#1{%
    {\parindent \z@ \raggedright
     \interlinepenalty \@M
     \normalfont
     \huge \bfseries #1\par}%
     \nobreak
     \vskip 3ex
     \@afterheading}

\usepackage{marginnote}

\makeatletter
\newcommand{\DeclareMathActive}[2]{%
  \expandafter\edef\csname keep@#1@code\endcsname{\mathchar\the\mathcode`#1 }
  \begingroup\lccode`~=`#1\relax
  \lowercase{\endgroup\def~}{#2}%
  \AtBeginDocument{\mathcode`#1="8000 }%
}

\newcommand{\std}[1]{\csname keep@#1@code\endcsname}
\patchcmd{\newmcodes@}{\mathcode`\-\relax}{\std@minuscode\relax}{}{\ddt}
\AtBeginDocument{\edef\std@minuscode{\the\mathcode`-}}
\makeatother

\begin{document}

\title{First-order factors of linear Mahler operators\tnoteref{dedicatory,cc-by,open-access,dererumnatura}}

\author[inria]{Frédéric Chyzak}
\ead{frederic.chyzak@inria.fr}

\author[imb]{Thomas Dreyfus}
\ead{thomas.dreyfus@cnrs.fr}

\author[inria]{Philippe Dumas}
\ead{philippe.dumas@inria.fr}

\author[lix]{Marc Mezzarobba}
\ead{marc@mezzarobba.net}

\affiliation[inria]{organization={Inria}, country={France}}
\affiliation[imb]{organization={Université Bourgogne Europe, CNRS, IMB UMR 5584}, postcode={F-21000}, city={Dijon}, country={France}}
\affiliation[lix]{organization={LIX, CNRS, École polytechnique, Institut polytechnique de Paris}, postcode={91120}, city={Palaiseau}, country={France}}

\tnotetext[dedicatory]{Dedicated to the memory of Marko Petkovšek.}

\tnotetext[cc-by]{%
\ccby\:
This work is licensed under a Creative Commons Attribution 4.0 International License
(\url{http://creativecommons.org/licenses/by/4.0/}).}

\tnotetext[open-access]{For the purpose of Open Access, a CC-BY public copyright licence
has been applied by the authors to the present document and will
be applied to all subsequent versions up to the Author Accepted
Manuscript arising from this submission.}

\tnotetext[dererumnatura]{Supported in part by the French ANR grant \href{https://mathexp.eu/DeRerumNatura/}{\emph{De rerum natura}} (ANR-19-CE40-0018)
and by the French-Austrian ANR-FWF grant EAGLES (ANR-22-CE91-0007 \& FWF-I-6130-N). This work has been achieved in the frame of the EIPHI Graduate school (contract ANR-17-EURE-0002).
}

\begin{abstract}
We develop and compare two algorithms for computing first-order right-hand factors in the
ring of linear Mahler operators
$\ell_r M^r + \dots + \ell_1 M + \ell_0$
where $\ell_0, \dots, \ell_r$ are polynomials in~$x$ and $Mx = x^b M$ for some
integer~$b \geq 2$.
In other words, we give algorithms for finding all formal infinite product
solutions of linear functional equations
$\ell_r(x) f(x^{b^r}) + \dots + \ell_1(x) f(x^b) + \ell_0(x) f(x) = 0$.

The first of our algorithms is adapted from Petkovšek's classical algorithm for
the analogous problem in the case of linear recurrences.
The second one proceeds
by computing a basis of generalized power series solutions of the functional equation
and by using Hermite--Padé approximants to detect those linear combinations of the solutions
that correspond to first-order factors.

We present implementations of both algorithms and discuss their use
in combination with criteria from the literature
to prove the differential transcendence of power series solutions of Mahler
equations.
\end{abstract}


\maketitle

\begin{footnotesize}
\setcounter{tocdepth}{2}
\tableofcontents
\end{footnotesize}

\section{Introduction}
\label{sec:intro}

\subsection{Mahler equations}

Mahler equations are a type of functional equations that,
for some fixed integer~$b\geq2$,
relate the iterates $y(x)$, $y(x^b)$, $y(x^{b^2})$, \dots,
of an unknown function~$y$ under the substitution of~$x^b$ for~$x$.
Mahler originally considered nonlinear multivariate equations of this type
in his work in transcendence theory in the 1920s.
The focus later shifted to linear univariate equations,
in relation to the study of automatic sequences.
More recently,
questions in difference Galois theory
related to the existence of nonlinear differential equations
satisfied by solutions of difference equations of various types
led to a revival of the topic.

\begin{ex}\label{ex:classical}
\begin{enumerate}
\item
Let $t_0 t_1 \dots = 01101001\dotsm$ be the Thue--Morse sequence \parencite[A010060]{oeis},
defined as the fixed point with $t_0=0$ of the substitution
$\{ 0 \mapsto 01, 1 \mapsto 10 \}$.
It is classical that the series
$y(x) = \sum_{n=0}^{\infty} (-1)^{t_n} x^n$
can be written as
\[ y(x) = \prod_{j=0}^{\infty} (1 - x^{2^j}). \]
This infinite product obviously satisfies the linear first-order Mahler equation
\[y(x^2) = (1-x) y(x).\]
\item\label{item:baum-sweet}
The Baum--Sweet sequence $(a_n)_{n\in \bN}$ is the automatic sequence defined by $a_n=1$
if the binary representation of $n$ contains no blocks of consecutive~$0$ of odd length,
and $a_n = 0$ otherwise \parencite[A086747]{oeis}.
The generating function
$y(x)=\sum_{n \in \bN}a_{n}x^n$
satisfies the Mahler equation
\begin{equation*}\label{eq:infprod1}
y(x^4)+x y(x^2)-y(x) =0.
\end{equation*}
\item\label{item:stern-brocot}
The Stern--Brocot sequence $(a_n)_{n\in\bN}$ \parencite[A002487]{oeis}
was introduced by Stern to enumerate the nonnegative rational numbers bijectively
by the numbers~$a_n/a_{n+1}$.
\textcite{AlloucheShallit-1992-RKR} showed
that the sequence is completely defined by $a_0 = 0$, $a_1 = 1$,
and for all~$n \in \bN$,
\begin{equation*}
a_{2n} = a_n, \quad a_{4n+1} = a_n + a_{2n+1}, \quad a_{4n+3} = 2 a_{2n+1} - a_n.
\end{equation*}
Starting from this property, the algorithms developed in the present article can be used to rediscover the well-known expression
\begin{equation*}\label{eq:infprod2}\
  y(x) = \sum_{n \geq 0} a_n x^n =  x \prod_{k\geq 0} (1 + x^{2^k} + x^{2^{k+1}})
\end{equation*}
as a consequence of the equation
\[ x y(x) - (1 + x + 2x^2) y(x^2) + (1 + x^2 + x^4) y(x^4) = 0 . \]

\item\label{item:adamczewski-faverjon}
\textcite[Example~8.2]{AdamczewskiFaverjon-2017-MMR} introduce the four generating series
$y_i(x) = \sum_{n\geq0}a_{i,n}x^n$ where $a_{i,n}\in\{0,1\}$ encodes
for each~$i$
a different condition on the parities of the number of occurrences of $1$~and~$2$ in the ternary expansion of~$n$.
After deriving the linear Mahler system
\begin{equation*}
\vecty(x) =
\begin{pmatrix}
1 & x & 0 & x^2 \\
x & 1 & x^2 & 0 \\
0 & x^2 & 1 & x \\
x^2 & 0 & x & 1
\end{pmatrix}
\vecty(x^3)
\qquad\text{for}\qquad
\vecty(x) =
\begin{pmatrix} y_1(x) \\ y_2(x) \\ y_3(x) \\ y_4(x) \end{pmatrix},
\end{equation*}
they prove that the~$y_i(x)$ are linearly independent over~$\rat\bQbar$.
\end{enumerate}
We return to these examples and discuss some additional ones in Section~\ref{sec:examples}.
Example~\ref{ex:classical}(\ref{item:adamczewski-faverjon}) will also be our running example in~\S\ref{sec:syzygies};
see Example~\ref{ex:adamczewski-faverjon}.
\end{ex}%

\subsection{First-order factors and hypergeometric solutions}
\label{sec:intro-first-order}

The present article continues a line of work initiated in
our earlier publication \parencite{ChyzakDreyfusDumasMezzarobba-2018-CSL},
to which we refer for more context.
This work started when, back in December 2015,
the second author asked
about making effective a differential transcendence criterion
introduced in \parencite{DreyfusHardouinRoques-2018-HSM}.
This criterion boils down to determining the rational function solutions
of a nonlinear Mahler equation
analogous to the Riccati equation associated
with a linear differential equation.
Equivalently, this can be viewed as computing
the first-order right-hand factors with rational function coefficients
of the linear Mahler operator underlying the Mahler equation.
We return to this  motivation in~\S\ref{sec:transcendence}
and postpone to that section a more detailed discussion of
differential transcendence and differential algebraic independence
criteria based on Mahler equations.

Another motivation for studying first-order factors is that,
in the differential case,
Beke's method \parencite{Markoff1891,Bendixson1892,Beke-94-IHL}%
\footnote{See \parencite{BostanRivoalSalvy-2024-MDE} for more on the history of this method.}
reduces
the problem of factoring a differential equation into irreducible factors
to that of finding the first-order right-hand factors
of some auxiliary equations.
One may expect that the method adapts to the Mahler case,
so the present work paves the way to future factorization algorithms.

Our algorithms require
solving linear Mahler equations in various domains.
To this end, we build on results from \parencite{ChyzakDreyfusDumasMezzarobba-2018-CSL}.
There, we worked with polynomials, rational functions, and series whose coefficients were in a computable subfield~$\bK$ of~$\bC$,
so we continue with this assumption here.
The following definitions and notation will be used throughout.

\begin{defi}\label{def:L}
For some fixed integer~$b\geq2$, called the \emph{radix},
we denote by~$M$
the \emph{Mahler operator} with regard to~$b$,
which is defined as mapping any function~$y(x)$ to~$y(x^b)$.
For~$r \in \bN$ and polynomials~$\ell_i \in \pol\bK$
with coefficients in a field~$\bK$ to be fixed by the context,
we consider the \emph{linear Mahler equation}
\begin{equation}
\tag{L}\label{eq:linear}
\ell_r M^r y + \dots + \ell_1 M y + \ell_0 y = 0 ,
\end{equation}
as well as the corresponding \emph{linear operator}
\begin{equation*}
  L = \ell_r M^r  + \dots + \ell_1 M  + \ell_0 .
\end{equation*}

We assume~$\ell_r \neq 0$ and call~$r$ the \emph{order} of the equation~\eqref{eq:linear} and of the operator~$L$.
We write~$d$ for the maximal degree $\max_k \deg \ell_k$ of the coefficients.
We call
\begin{equation}
\tag{R}\label{eq:riccati}
\ell_r u \, M u \dotsm M^{r-1} u + \dots + \ell_2 u \, M u + \ell_1 u + \ell_0 = 0 ,
\end{equation}
the \emph{Riccati Mahler equation} associated with~\eqref{eq:linear}.

\end{defi}

\begin{hyp}\label{hyp:l0-neq-0}
Throughout, we further assume
that $\ell_0$~is also nonzero
and that $L$~is primitive, that is, that the family of the~$\ell_i$ has gcd~$1$.
\end{hyp}

\begin{defi}
By a \emph{ramified rational function},
we mean a rational function in some fractional power of~$x$,
so that rational functions are particular ramified rational functions.
\end{defi}

Our goal is to find the right-hand factors $M - u$ of~$L$ where~$u$ is a ramified rational function.
Equivalently, we want to compute the \emph{hypergeometric} solutions of~$L$,
that is, the solutions $y$ of~$L$ that satisfy a first-order linear Mahler
equation with ramified rational coefficients.
Informally, the link between both equations is the relation $u = M y / y$.
However, general solutions of~\eqref{eq:linear} may live in a ring containing zero divisors,
so that quotients are not always well defined
(see the construction of~$\omdr$ in~\S\ref{sec:difference} and Definition~\ref{def:omdr}).

\begin{ex}\label{ex:adamczewski-faverjon-intro}
By searching for a linear dependency among the first coordinates of the iterates
$\vecty, M \vecty, M^2 \vecty, \dots$
of the vector $\vecty$ appearing in Example~\ref{ex:classical}(\ref{item:adamczewski-faverjon}),
one obtains a linear Mahler operator of order~$4$ and degree~$258$ annihilating~$y_1$.
(See Example~\ref{ex:stern-brocot} in §\ref{sec:examples} for details of the method.)
This operator turns out to annihilate all four $y_i$, as one can check by applying the method of Example~\ref{ex:stern-brocot} to each coordinate of the vector~$\vecty$ in turn.
The algorithms developed in the following sections reveal that its hypergeometric solutions are exactly the series
\begin{equation}\label{eq:adamczewski-faverjon-hyp-sols}
\begin{split}
  a \prod_{k=0}^{\infty}(1 - x^{3^k} - x^{2\cdot 3^{k}}), \qquad
  b \prod_{k=0}^{\infty}(1 + x^{3^k} - x^{2\cdot 3^{k}}), \\
  (c_0 + c_1 x) \prod_{k=0}^{\infty}(1 + x^{2\cdot 3^k} + x^{4\cdot 3^k}), \qquad
\end{split}
\end{equation}
for arbitrary constants $a, b, c_0, c_1$.
Comparing with the series expansions of $y_1, \dots y_4$,
we deduce that none of the~$y_i$ is hypergeometric, let alone rational.
By the classical dichotomy for Mahler functions (see~\S\ref{sec:transcendence}),
all the~$y_i$ are therefore transcendental (i.e., not an algebraic function of~$x$).
\end{ex}

\subsection{Contribution}

This article provides algorithms to compute the rational solutions
of Riccati Mahler equations, and more generally their ramified rational solutions.
We develop two approaches.

The first one is adapted from the classical algorithm
by \textcite{Petkovsek-1992-HSL} for finding the hypergeometric solutions of a
linear recurrence equation.
Petkovšek's algorithm searches for first-order factors of difference operators
in a special form called the Gosper--Petkovšek form.
The existence of Gosper--Petkovšek forms relies on the fact that
for any two nonzero polynomials $A$ and~$B$,
the set of integers~$i$ such that $A(x)$ and~$B(x+i)$ have a nontrivial common
divisor is finite.
As this is not true when the shift is replaced by the Mahler operator,
we need to slightly depart from the classical definition of
Gosper--Petkovšek forms.
Doing so, we obtain a first complete algorithm for finding first-order factors
of Mahler equations of arbitrary order.
Note that \textcite{Roques-2018-ARB} recently gave a slightly different adaptation of
Petkovšek's algorithm to the case of Mahler equations of order two.
Like Petkovšek's, these algorithms have to consider an exponential number of
separate subproblems in the worst case.
We discuss several ways of pruning the search space to mitigate this issue in
practice.

Our second algorithm avoids the combinatorial search phase entirely,
at the price of a worst-case exponential behavior of a different nature.
It is based on a relaxation of the problem that can be solved using
Hermite--Padé approximants%
\footnote{We find some pleasant irony
in our application of Hermite--Padé approximants to problems on Mahler operators,
after Mahler himself has introduced similar approximants
in his work \parencite{Mahler-1968-PS}.}.
The idea is to search for series solutions~$y$ of the linear Mahler equation
that make $My/y$ rational.
Roughly speaking,
we first compute a basis $(y_1,\dots,y_\zdim)$ of series solutions,
then search for linear combinations with polynomial coefficients of
$y_1,\dots,y_\zdim,M y_1,\dots,M y_\zdim$
that vanish to a high approximation order~$\sigma$,
and finally isolate, among these relations, those corresponding to
hypergeometric solutions by solving a polynomial system.
Though we are not aware of any exact analogue of our algorithm in the
literature, variants of the same basic idea have been used by several authors in
the differential case (see below for references).

In order to state the algorithms and justify their correctness,
it is useful to work in a ring containing “all” solutions of the linear equation~\eqref{eq:linear},
or at least all solutions needed in the discussion.
Instead of appealing for this to the general Picard--Vessiot theory of linear
difference equations, we introduce suitable conditions (Hypothesis~\ref{def:propertyP}) that
suffice for our purposes and construct a  ring satisfying them, whose
elements (unlike those of Picard--Vessiot rings of Mahler equations) admit
simple representations as formal series.
In passing we define a suitable notion of Mahler-hypergeometric function and
establish its basic properties.

We compare the performance of our two approaches based on an implementation%
\footnote{Authored by Ph.~Dumas. Available at \url{https://mathexp.eu/dumas/dcfun/}.}
in Maple
and observe that the second algorithm turns out to be more efficient
in practice on examples from the mathematical literature.
Finally, we use this implementation in combination with criteria such as the one
mentioned earlier to prove several differential algebraic independence results
between series studied in the literature.

\subsection{Related work}

To the best of our knowledge,
the problem we consider here was first discussed in
the doctoral dissertation of
\textcite[§3.6]{Dumas1993},
which contains an incomplete method for finding hypergeometric solutions of
Mahler equations.
Dumas' method is somewhat reminiscent of Petkovšek's algorithm,
but lacks a proper notion of Gosper--Petkovšek form.
\textcite[§6]{Roques-2018-ARB}, as already mentioned, describes a complete analogue
of Petkovšek's algorithm for Mahler equations,
but restricts himself to equations of order two.
We are not aware of any other reference dealing specifically with the
factorization of Mahler operators.

However, it is natural to try and adapt to Mahler equations algorithms that
apply to differential equations or to difference equations in terms of the
shift operator.
In the differential case, factoring algorithms are a classical topic, going back
at least to Fabry's time \parencite*[§V]{Fabry1885};
we refer to \parencite[§1.4]{BostanRivoalSalvy-2024-MDE}
for a well-documented overview.
The second of our algorithms, using Hermite--Padé approximation, is related to
methods known from the differential case.
Most similar to our approach is maybe an unpublished article by
\textcite[§4]{BronsteinMulders1999}
where they present a heuristic method for solving Riccati differential equations
based on Padé (not Hermite--Padé) approximants of series solutions with
indeterminate coefficients.
The same idea appears in works
by \textcite[§2.5.2]{Pfluegel1997}
and by \textcite[§4.1]{PutSinger-2003-GTL},
though neither of these references discusses in detail how to deal with drops in
the degree of candidate solutions for special values of the parameters.
The core idea of enforcing the vanishing of high-order terms of series solutions
of the Riccati equation so as to reduce to the solution of polynomial equations
in a number of unknowns bounded by the order of the differential equation
already appears in Fabry's 1885 thesis.

In the shift case, the classical algorithm for finding hypergeometric solutions
is that of \textcite{Petkovsek-1992-HSL}.
It is the direct inspiration for our first method.
Petkovsek's algorithm is itself based on Gosper's
hypergeometric summation algorithm \parencite{Gosper1978},
and was previously adapted to $q$-difference equations in
\parencite{AbramovPetkovsek1995,AbramovPaulePetkovsek1998}.
Van Hoeij and Cluzeau \parencite{vanHoeij1999, CluzeauVanHoeij2004} later proposed alternative algorithms that are faster in practice;
it seems likely that the ideas would also be relevant in the Mahler case, but we
do not explore this question~here.

Another line of work aims at unified algorithms for linear functional operators
of various types based on the formalism of Ore polynomials
\parencite[e.g.,][]{BronsteinPetkovsek1993}.
Factoring algorithms, however, remain specific to each individual type of
equation.
In addition, even methods that do apply to almost all types of Ore operators
sometimes fail in the Mahler case because the commutation $M x = x^b M$ does not
preserve degrees with respect to~$x$.

Turning now to the structure and computation of generalized series solutions of
linear Mahler equations, Roques's discussion of the local exponents of Mahler
systems \parencite*[§4]{Roques-2018-ARB} forms the basis for
our §\ref{sec:concrete}.
Further developments (not used here) include recent work by
\textcite{Roques-2022-FMM},
\textcite{FaverjonRoques-2024-HSM},
and
\textcite{FaverjonPoulet2022}.

\subsection{Outline}
\label{sec:outline}

The “structural” results about the solution spaces of Mahler equations come
first in the text.
We first generalize our Mahlerian problem to the context of difference rings,
including the case of nonsurjective morphisms:
in \S\ref{sec:structure},
we derive a parametric description of the right-hand factors
of a difference operator (Theorem~\ref{thm:structure-main}).
To accommodate the Mahler operator in the previous generalized framework,
we then introduce in~\S\ref{sec:concrete}  an explicit difference ring~$\omdr$
containing all the solutions of the linear Mahler equations
that are needed in the parametric descriptions of the solutions~$u$
to Riccati Mahler equations.
In that section, we also define classes of $F$-hypergeometric solutions
for various difference fields~$F$
and, for $F$ the field of Puiseux series,
we partition the  set of Puiseux series solutions of the Riccati equation  according to the coefficient of their valuation term
(Theorem~\ref{thm:structure-puiseux}).
To describe the ramified rational solutions of the Riccati equation,
we then change~$F$ to the field of ramified solutions
in~\S\ref{sec:rational-ramified-sol},
where we obtain a partition (Theorem~\ref{thm:structure-rational-ramified})
that is finer than the partition
induced by the partition in the Puiseux series on ramified solutions.
In preparation for the algorithms,
the technical~\S\ref{sec:bounds-for-rational-solutions} presents
bounds on the degree of numerators and denominators
of rational solutions to a Riccati Mahler equation
(Proposition~\ref{prop:deg-bounds-for-ratfuns}).

We continue with two algorithmic approaches.

We review Petkovšek's classical algorithm for the shift case
and Roques's analogue for order~$2$ in the Mahler case
before developing our generalization in~\S\ref{sec:petkovsek-variant}.
We first propose a brute-force algorithm (Algorithm~\ref{algo:BP}),
which we prove to be correct (Theorem~\ref{thm:BP}).
Next, in~\S\ref{sec:efficiency-improvements}
we propose several pruning criteria to improve its exponential behavior,
leading to an improved algorithm (Algorithm~\ref{algo:IP}).

We then develop our relaxation method based on Hermite--Padé approximations:
after studying in~\S\ref{sec:syzygies} the approximate syzygy module of a basis of truncated solutions of~\eqref{eq:linear} and its limit as precision goes to infinity,
we obtain in~\S\ref{sec:hermite-pade} an algorithm (Algorithm~\ref{algo:HP}) that produces
successive sets of parametrized candidates,
each containing all true solutions,
until it stabilizes on true solutions~only.

To conclude,
we present examples of applications and timings in~\S\ref{sec:impl-and-benchmark},
where we can observe that our relaxation method beats
the other by combinatorial exploration
in a number of natural examples.
We finally apply our implementation in order to prove
properties of differential transcendence on examples
in~\S\ref{sec:transcendence}.

Part~\ref{part:petkovsek} on our generalized Petkovšek algorithm
and Part~\ref{part:hermite-pade} on our ap\-prox\-i\-mants-based algorithm
are completely independent.
They are both based on the structural results of Part~\ref{part:structural}.
Part~\ref{part:implem} is mostly independent from the other parts
if one admits the existence of the various algorithms.
Also, a few sections are very independent from the rest of the text
and might be skipped by readers with a specific interest:
\S\ref{sec:bounds-for-rational-solutions}~provides degree bounds on rational solutions
that are to be used as a black box in Part~\ref{part:hermite-pade};
\S\ref{sec:efficiency-improvements}~presents technical algorithmic improvements to speed up our first plain generalization of Petkovšek's algorithm;
\S\ref{sec:rank-rels}~describes how approximate syzygies corresponding to nonsimilar (Definition~\ref{def:similarity}) solutions interact;
\S\ref{sec:back-to-rat}~sketches how the approach to the rational solving of the Riccati Mahler equation~\eqref{eq:riccati}
adapts to the linear Mahler equation~\eqref{eq:linear};
\S\ref{sec:transcendence}~presents an application
whose background in number theory is a bit away from our main theme,
although it was our original motivation for the work.

\subsection{Notation}

The following is a collection of notation used throughout the text.

\begin{notation}\label{notation:general}
We will denote the field of Puiseux series over a field\/~$\bL$ by
\begin{equation*}
\puiseux\bL = \bigcup_{q\geq1}\rfls\bL{q}
\end{equation*}
and the field
of ramified rational functions
over~$\bL$ by
\begin{equation*}
\ramrat\bL = \bigcup_{q\geq1}\rrat\bL{q} .
\end{equation*}
We use the notation $\bL[x]\langle M\rangle$
to denote the algebra generated by~$M$ over~$\bL[x]$
and satisfying the relation $M f(x) = f(x^b) M$
for any~$f \in \bL[x]$.
Similar definitions apply for other coefficient domains.
All of the operators and functions $M$, $\val$, $\ln$, $\log$
take precedence over additions and products,
which means $My/y = (My)/y$, $\val a \, b = (\val a)b$, etc.
We write $Ly$ for the application of an operator~$L$ to a function, or $L(x,M)\,y(x)$ if needed.
We write~$\nonz{S}$ for a given set~$S$ containing~$0$ to denote~$S \setminus \{0\}$.
The set~$S$ will be the set~$\bN$ of natural numbers, a field, a vector space, a cone (see Definition~\ref{def:F-cone}), a Cartesian product, etc.
A transpose is denoted by an exponent~$T$.
A tuple~$z$ is identified with a row vector, its transpose~$z^T$ with a column vector.
\end{notation}

\part{Structural results}
\label{part:structural}

\section{Structure of the solution set of the Riccati equation}
\label{sec:structure}

The notions introduced in this section will be used when solving Riccati Mahler equations.
As they do not depend on the specific choice of the Mahler operator,
we write them in the broader generality of difference rings.

\subsection{Basic notions of difference algebra}

A difference ring is commonly defined as a pair $(\dring,\sigma)$ formed
by a commutative ring with identity~$\dring$ and an automorphism~$\sigma$ of~$\dring$;
see, e.g., \parencite{Cohn-1965-DA} or \parencite{PutSinger-1997-GTD},
and a difference field is a difference ring that is a field.
An example is the field~$\puiseux\bK$ of Puiseux series,
equipped with the Mahler operator~$M$.
However,
we relax the definition to accept a map~$\sigma$ that is only an injective endomorphism
\parencite[\emph{cf.}][]{Wibmer2013},
since we intend to consider  the restriction of the Mahler operator from~$\rat\bK$ to itself,
which is not an automorphism of the field of rational functions.
When the context is clear,
we write~$\dring$ instead of $(\dring,\sigma)$.

Given a difference ring $(\dring,\sigma)$,
a difference ring extension is a difference ring
$(\dring',\sigma')$ such that $\dring'$~is an extension ring of~$\dring$ and
$\sigma'$ restricted to~$\dring$ is equal to~$\sigma$.
In practice, we will always write again~$\sigma$ for the extended~$\sigma'$.
The ring of constants of a difference ring $\dring$,
denoted~$\dring^\sigma$,
is the subring of elements in~$\dring$ fixed by the endomorphism~$\sigma$.
If~$\dring$ is a difference field, the ring of constants is a field.

In what follows, every difference ring will be
a difference ring extension of the field $\rat\bK$ of rational functions.
For the needed level of generality, let us introduce
\begin{align}
\tag{L${}_\sigma$}\label{eq:linear-sigma}
\ell_r \sigma^r y + \dots + \ell_1 \sigma y + \ell_0 y &= 0 , \\
\tag{R${}_\sigma$}\label{eq:riccati-sigma}
\ell_r u \, \sigma u \dotsm \sigma^{r-1} u + \dots + \ell_2 u \, \sigma u + \ell_1 u + \ell_0 &= 0 ,
\end{align}
where $\ell_0\ell_r \neq 0$ and each~$\ell_i$ is in~$\rat\bK$.
We will refer to these equations as, respectively,
the linear difference equation~\eqref{eq:linear-sigma}
and the Riccati difference equation~\eqref{eq:riccati-sigma}.

\subsection{First-order factors and their solutions}
\label{sec:first-order}

We now explain the link between the linear and the Riccati equations:
the solutions to the Riccati equation are the coefficients~$u$
of the monic first-order right-hand factors $\sigma-u$ of the linear difference operator~$L$.

\begin{lem}\label{lem:linear-riccati-equiv}
Given an element~$u$ of a difference field extension~$F$ of\/~$\rat\bK$,
the operator
\begin{equation*}
L = \ell_r(x) \sigma^r + \dots + \ell_0(x) \in \rat\bK\langle \sigma\rangle
\end{equation*}
associated with the linear difference equation~\eqref{eq:linear-sigma}
admits~$\sigma - u$ as a first-order right-hand factor in~$F\langle \sigma\rangle$
if and only if $u$~satisfies the Riccati difference equation~\eqref{eq:riccati-sigma}.
\end{lem}
\begin{proof}
The ring $F\langle \sigma\rangle$ is a skew Euclidean ring \parencite{BronsteinPetkovsek1993}.
By Euclidean division of~$L$ on the right by~$\sigma - u$ in the ring~$F\langle \sigma\rangle$,
we obtain $L = \tilde L (\sigma - u) + R$ where
$R$~is exactly the left-hand side of~\eqref{eq:riccati-sigma}.
\end{proof}

We define hypergeometric elements in analogy with the classical difference case
for shift operators.
\begin{defi}\label{def:F-hyperg}
Given a difference field extension~$F$ of\/~$\rat\bK$
and a difference ring extension~$\dring$ of~$F$,
an element~$y$ of~$\dring$ is $F$-hypergeometric
if there exists $u \in F$ satisfying $\sigma y = u y$.
\end{defi}
Note that $y$ may be zero in the definition,
but that we will focus on nonzero hypergeometric~$y$ throughout.

The set of $F$-hypergeometric elements is generally not stable under addition,
but it enjoys the structure of an $F$-cone, a notion that we introduce now.

\begin{defi}\label{def:F-cone}
Given a field~$F$,
a~nonempty set stable under multiplication by elements of~$F$
will be called an \emph{$F$-cone}.
\end{defi}

\begin{figure}
\centerline{%
  \begin{tikzpicture}[xscale = 2.0, yscale = 0.9]
    \node (Kx) at (0,1) {$\rat{\bK}$} ;
    \node (F)  at (0,3) {$F$} ;
    \node (R)  at (0,5) {$\dring$} ;

    \node (K)  at (1,0) {$\bK$} ;
    \node (L)  at (1,2) {$\bL = F^\sigma  = \dring^\sigma$} ;
    \node (sol)at (1,4) {$\ker_\dring(L)$} ;

    \draw (Kx) -- (F) ;
    \draw (F)  -- (R) ;

    \draw (K)  -- (Kx);
    \draw (K)  -- (L) ;
    \draw (L)  -- (F) ;
    \draw[dashed] (L)  -- (sol);
    \draw[dashed] (sol)-- (R) ;

  \end{tikzpicture}%
}
  \caption{\label{fig:unsurprizing-ring}The inclusion relations for a difference field extension~$F$ of\/~$\rat\bK$ and a  ring extension~$\dring$ of it satisfying Hypothesis~\ref{def:propertyP}. The two rings $F$ and~$\dring$ share the same field of constants~$\bL$.
  All inclusion but the dashed ones are difference ring inclusions;
  the dashed ones are only inclusions of $\bL$-vector spaces.
  The space $\ker_\dring(L)$ is the space of the solutions of the operator~$L$ underlying the equation~\eqref{eq:linear-sigma}.}
\end{figure}
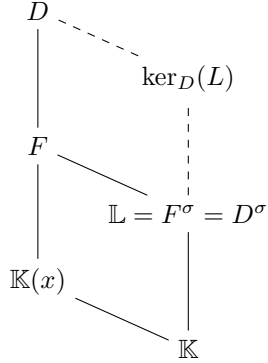

The following hypothesis captures the notion of an extension of some difference field~$F$
that contains “enough” $F$-hypergeometric solutions of the linear equation~\eqref{eq:linear-sigma}
to suitably describe the “full” solution set of the Riccati equation~\eqref{eq:riccati-sigma} in~$F$.
The hypothesis is illustrated in Figure~\ref{fig:unsurprizing-ring}.

\begin{hyp}\label{def:propertyP}
Let $F$ be a difference field extension of\/~$\rat\bK$,
with constant field\/~$\bL$ containing~$\bK$.
A difference ring extension~$\dring$ of~$F$ is
said to satisfy Hypothesis~\ref{def:propertyP}
if
\begin{enumerate}
\item \label{it:universal-constant}
  the constant ring of~$\dring$ is the field\/~$\bL$,
\item \label{it:universal-1}
  for each nonzero~$u \in F$, the equation $\sigma y = uy$ has nonzero solutions in~$\dring$,
\item \label{it:universal-simplicity}
  $\dring$~is simple, meaning that $\sigma$~is an automorphism
  and that the only ideals of~$\dring$ stable under~$\sigma$ are $(0)$ and~$\dring$.

\end{enumerate}

\end{hyp}

In contrast with the classical theory,
we do not consider any universal Picard--Vessiot algebra
\parencites[Prop.~1.33]{PutSinger-1997-GTD}[Theorem~35]{Roques-2018-ARB}:
by point~\ref{it:universal-1},
only equations of order~$1$ are known to have solutions in~$\dring$.
On the other hand, Hypothesis~\ref{def:propertyP} ensures the following natural bound on the dimension of solution space.

\begin{lem}\label{lem:has-Property-U}
Let $\dring \supseteq F \supseteq \rat\bK$ be difference rings satisfying
of Hypothesis~\ref{def:propertyP}
for a suitable constant field\/~$\bL$ containing~$\bK$.

Then, the extension~$\dring$ is such that
  for any linear difference equation~\eqref{eq:linear-sigma} of order~$r$ with coefficients in\/~$\rat\bK$,
  the vector space of solutions in~$\dring$ of this equation
  has dimension at most~$r$ over~$\bL$.

\end{lem}

\begin{proof}
By the assumption of simplicity
\parencites[Def.~1.1]{PutSinger-1997-GTD}[Def.\ 4.1 and~4.4]{Singer-2016-AAA},
the ring~$\dring^\sigma$ of constants
is a field \parencite[Lemma~1.7]{PutSinger-1997-GTD}.
By our assumption $\ell_0\ell_r \neq 0$ made in the introduction (Hypothesis~\ref{hyp:l0-neq-0}),
the companion matrix~$A$ of~$L$ is invertible over~$\rat\bK$
and hence nonsingular over~$\dring$.
Suppose $y_1$, \dots, $y_m$ are $\bL$-linearly independent solutions of~$Ly = 0$,
and, for $1\leq i\leq m$, introduce $Y_i = (y_i,\sigma y_i,\dots,\sigma^{r-1}y_i)^T$,
which satisfies $\sigma Y_i = A Y_i$.
By point~\ref{it:universal-constant} and \parencite[Lemma~4.8]{Singer-2016-AAA},
this family~$\{Y_i\}_{i=1}^m$ of solutions of $\sigma Y = AY$,
which is linearly independent over~$\bL$,
is linearly independent over~$\dring$ in the free $\dring$-module~$\dring^r$.
By the commutativity of~$\dring$, the relation~$m \leq r$ holds,
thus proving the lemma.
\end{proof}

\subsection{Parametrization of the solution set}
\label{sec:parametrization}
Let $F$ be a fixed difference field extension of~$\rat\bK$.
Note that an extension~$\dring$ of~$F$ satisfying Hypothesis~\ref{def:propertyP} is only a ring,
so $\sigma y/y$~is not defined for all~$y\in \dring$.
However, restricting to nonzero $F$-hypergeometric~$y$ makes $1/y$ well defined,
as stated in the following lemma.

\begin{lem}\label{lem:invertible-hypergeom}
Any nonzero $F$-hypergeometric element~$y$
of some  difference ring extension~$\dring$ satisfying Hypothesis~\ref{def:propertyP} over~$F$
has an $F$-hypergeometric inverse in~$\dring$.
\end{lem}
\begin{proof}
For any nonzero $F$-hypergeometric~$y \in \dring$,
introduce~$u\in F$ satisfying $\sigma y = uy$.
The ideal generated by~$y$ in~$D$ contains~$1$ by simplicity,
so $y$~has an inverse~$z$.
Consequently, $\sigma z = 1/(\sigma y) = 1/(uy) = (1/u) \, z$,
proving that $z$~is $F$-hypergeometric.
\end{proof}

We now define similarity, a key concept in what follows.
\begin{defi}\label{def:similarity}
Given a difference field extension $F$ of\/ $\rat\bK$
and a difference ring extension~$\dring$ of $F$,
two elements $y_1$ and~$y_2$ of~$\dring$ are $F$-similar
if there exists a nonzero~$q \in F$ satisfying $y_1 = q y_2$.
\end{defi}
Similarity is an equivalence relation,
and $y_1$ and~$y_2$ in the definition
are either both nonzero or both zero.
In what follows, we focus on nontrivial equivalence classes,
that is, other than~$\{0\}$, which therefore do not contain~$0$.
We will denote
such nontrivial similarity classes by~$\nonz{\fH}$,
reserving for the augmentation by~$0$ of these similarity classes
the notation $\fH = \nonz{\fH} \cup \{0\}$.
(Cf.\ Notation~\ref{notation:general} in the introduction.)

\begin{lem}\label{thm:F-sim-F-hyp-is-L-vec}
For any $F$-similarity class~$\nonz{\fH}$ of $F$-hypergeometric elements,
the set~$\fH$ is an~$\bL$-vector space.
\end{lem}
\begin{proof}
Because $\bL$ is a subfield of~$F$, $\fH$~is an $\bL$-cone.
We show that it is also stable under addition.
Consider two elements $y_1$ and~$y_2$ in~$\nonz{\fH}$ such that $y_1+y_2 \neq 0$,
and $u$ and~$q$ in~$\nonz{F}$ such that $\sigma y_1 = u y_1$ and~$y_2 = q y_1$.
Then, as~$q \neq -1 \neq \sigma q$,
$\sigma(y_1+y_2) = (1+\sigma q) u y_1 = (1+\sigma q)(1+q)^{-1} u (y_1+y_2)$,
from which follows that $y_1+y_2$~is $F$-hypergeometric and $F$-similar to~$y_1$,
thus in~$\nonz{\fH}$.
\end{proof}

We use the notation~$\projzn{V}$
to denote the projectivization of a vector space~$V$,
leaving implicit the field over which $V$~is defined.
We also write~$\bP^n(K)$
for the $K$-projective space $\projzn{K^{n+1}}$ of dimension~$n$ over a field~$K$.

\begin{thm}\label{thm:structure-main}
Let $\dring$~be a  difference ring extension of~$F$ satisfying Hypothesis~\ref{def:propertyP}.
For a fixed~\eqref{eq:linear-sigma}, the following holds:
\begin{enumerate}
\item\label{it:vect-H}
The $\bL$-cone of $F$-hypergeometric solutions of~\eqref{eq:linear-sigma} in~$\dring$
naturally partitions into the class~$\{0\}$
and nontrivial $F$-similarity classes~$\nonz{\fH}$.
Moreover, the $\bL$-cones~$\fH$ are $\bL$-vector spaces in direct sum, and
their dimensions add up to at most~$r$.
\item\label{it:param-H}
The set of solutions in~$F$ of~\eqref{eq:riccati-sigma} partitions
into the images under the map $y \mapsto \sigma y/y$
of all nontrivial\/ $F$-similarity classes
of the set of\/ $F$-hy\-per\-geo\-met\-ric solutions of~\eqref{eq:linear-sigma}.
\item\label{it:param-U}
Each nontrivial\/ $F$-similarity class~$\nonz{\fH} \subseteq \dring$
of the set of\/ $F$-hy\-per\-geo\-met\-ric solutions of~\eqref{eq:linear-sigma}
induces a one-to-one parametrization of its image under $y \mapsto \sigma y/y$
by the $\bL$-projective space~$\projzn{\fH}$.
\item\label{it:disjoint-U}
The sum of the dimensions of the parametrizing $\bL$-projective spaces~$\projzn{\fH}$
does not exceed the order~$r$ of the equation~\eqref{eq:linear-sigma}
minus the number of nontrivial similarity classes~$\nonz{\fH}$.
\end{enumerate}
\end{thm}
\begin{proof}
\begin{trivlist}

\item \ \ \ref{it:vect-H}.
For each $F$-similarity class~$\nonz{\fH}$ of $F$-hypergeometric solutions of~$L$,
Lemma~\ref{thm:F-sim-F-hyp-is-L-vec} implies that $\fH$~is an $\bL$-vector space.
We now prove that all the~$\fH$ are in direct sum.
Given a relation $\lambda_0y_0 + \dots + \lambda_Ny_N = 0$
with coefficients~$\lambda_k \in \nonz{\bL}$
between nonzero $F$-hypergeometric solutions~$y_k$
in distinct $F$-similarity classes,
introduce $u_k\in F$ such that $\sigma y_k = u_ky_k$ for~$0\leq k\leq N$.
Applying~$\sigma - u_0$ yields
\begin{equation}\label{eq:shorter-rel}
\lambda_1(u_1-u_0)y_1 + \dots + \lambda_N(u_N-u_0)y_N = 0 .
\end{equation}
Because the~$y_k$ were taken from distinct classes,
each of the~$y_0/y_k$, which exists by Lemma~\ref{lem:invertible-hypergeom},
cannot be in~$F$, let alone be a constant.
So, $\sigma\bigl(y_0/y_k\bigr) = \bigl(u_0/u_k\bigr) \times \bigl(y_0/y_k\bigr)$ is not~$y_0/y_k$,
and $u_0/u_k$~is not~$1$.
Therefore,
none of the coefficients $\lambda_k(u_k-u_0)$ for $1\leq k\leq N$ is zero,
and \eqref{eq:shorter-rel}~is a shorter nontrivial relation.
Iterating the process, we obtain that some~$y_k$ is zero, a contradiction.
As a consequence, the $\bL$-vector spaces~$\fH$
are in direct sum.
By the conclusion of Lemma~\ref{lem:has-Property-U},
this direct sum is included in an $\bL$-space of dimension at most~$r$,
hence the dimension bound.

\item \ \ \ref{it:param-H}.
Given a solution~$u \in F$ of~\eqref{eq:riccati-sigma},
Lemma~\ref{lem:linear-riccati-equiv} implies
that $\sigma-u$~is a right-hand factor of~$L$.
Because of Hypothesis~\ref{def:propertyP}, there exists~$y\neq0$ such that $\sigma y = uy$,
and $y$~is invertible by Lemma~\ref{lem:invertible-hypergeom}.
As a consequence, $y$~is also a solution of~\eqref{eq:linear-sigma}
and $u$~equals~$\sigma y/y$.
So the map $y \mapsto \sigma y/y$ is a (well-defined) surjection
from the set of nonzero $F$-hypergeometric solutions of~\eqref{eq:linear-sigma}
to the set of solutions of~\eqref{eq:riccati-sigma} in~$F$.
If $u$~can be obtained by two solutions $y_1$ and~$y_2$,
that is, if $\sigma y_1/y_1 = \sigma y_2/y_2$,
then $\sigma(y_1/y_2) = y_1/y_2$ is a constant in~$\dring$, thus in~$\bL$,
so that, $y_1$ and~$y_2$ are $F$-similar.
As a consequence, the partitioning
of the set of nonzero $F$-hypergeometric solutions of~\eqref{eq:linear-sigma}
translates into a partitioning
of the set of solutions of~\eqref{eq:riccati-sigma} in~$F$.

\item \ \ \ref{it:param-U}.
By the proof of point~\ref{it:param-H},
a nontrivial\/ $F$-similarity class~$\nonz{\fH}$
of the set of\/ $F$-hy\-per\-geo\-met\-ric solutions of~\eqref{eq:linear-sigma}
maps under $y \mapsto \sigma y/y$ to a subset~$U$
of the set of solutions of~\eqref{eq:riccati-sigma} in~$F$.
Given $y_1$ and~$y_2$ in the same class~$\nonz{\fH}$,
thus satisfying $y_1 = q y_2$
for some $q \in F$ (which needs to be nonzero),
if both share the same image in~$U$,
that is, $\sigma y_1/y_1 = \sigma y_2/y_2$ is the same element of~$U$,
then the proof of~\ref{it:param-H} has shown that $q$~is a constant from~$\bL$.
This implies that $y_1$ and~$y_2$ are in the same $\bL$-projective class,
thus proving the result.

\item \ \ \ref{it:disjoint-U}.
Point~\ref{it:vect-H} implies that each similarity class~$\fH$ is finite-dimensional.
If it has $\bL$-dimension~$s$,
its associated $\bL$-projective space~$\projzn{\fH}$ has dimension~$s-1$.
The statement is next a direct consequence
of points \ref{it:vect-H} and~\ref{it:param-U}.
\qedhere

\end{trivlist}
\end{proof}

\begin{rem}\label{rem:rational-param}
Point~\ref{it:param-U} of Theorem~\ref{thm:structure-main}
leads to a more explicit parametrization,
owing to the finite $\bL$-dimension of the space~$\fH$.
Fix~$y_0$ in~$\nonz{\fH}$ and let $u_0$ satisfy $\sigma y_0 = u_0 y_0$.
Let $(q_1y_0, \dots, q_s y_0)$ be some basis of~$\fH$.
With coordinates, the bijection of the theorem translates into the map
\begin{equation*}
(c_1:\ldots:c_s) \mapsto
  \frac{\sigma(c_1 q_1 + \dots + c_s q_s)}{c_1 q_1 + \dots + c_s q_s} u_0 ,
\end{equation*}
now a rational bijection from~$\bP^{s-1}(\bL)$ to its image.
\end{rem}

\begin{rem}
The partition obtained by point~\ref{it:param-H} of the theorem can be interpreted
as the partition induced on the solution set of~\eqref{eq:riccati-sigma}
by the equivalence relation on~$\nonz{F}$ defined by:
$u$~and~$u'$ are $F$-equivalent
if there exists~$f \in \nonz{F}$ such that $u/u' = \sigma f/f$.
\end{rem}

\section{Generalized series solutions of the linear equation}
\label{sec:concrete}

In this section, we define a Mahler difference ring, denoted~$\omdr$ (Definition~\ref{def:omdr}), of explicit formal expressions
that is an extension satisfying Hypothesis~\ref{def:propertyP}
of both $\ramrat\bKbar$ and $\puiseux\bKbar$.
This will be a crucial ingredient to describe
the Puiseux series solutions of a Riccati Mahler equation (Theorem~\ref{thm:structure-puiseux}).
To this end,
we describe $\puiseux\bL$-hypergeometric elements
(Proposition~\ref{prop:puiseux-hypergeometric-element})
so as to apply Theorem~\ref{thm:structure-main} for a given field~$\bL$.
The properties of~$\omdr$ are again used in~\S\ref{sec:rational-ramified-sol}
to describe the ramified rational solutions of the Riccati equation
(Theorem~\ref{thm:structure-rational-ramified}).

\subsection{The difference ring \texorpdfstring{$\omdr$}{D}}
\label{sec:difference}

Slightly adapting \parencite[\S5.2]{Roques-2018-ARB},
we introduce the $\puiseux\bKbar$-algebra with basis~$(e_\lambda)_{\lambda \in \nonz{\bKbar}}$ for elements~$e_\lambda$
satisfying the algebraic relations
$e_\lambda e_{\lambda'} = e_{\lambda\lambda'}$ and $e_1 = 1$.
This ring is not a domain, as shown by the product
$(e_\lambda + e_{-\lambda}) \times (e_\lambda - e_{-\lambda}) = 0$.
Next, we equip it with a structure of difference ring
by enforcing $M e_\lambda = \lambda e_\lambda$.
To support the intuition,
we henceforth denote $e_\lambda$ by $\mylog{\lambda}$
for $\lambda \in \nonz{\bKbar}$.
This mere notation bears no analytic meaning,
but it is reminiscent of the relation $\ln(x^b) = b \ln(x)$
in analysis,
as well as of the linear independence over
the field of meromorphic functions
of the family of functions $\left(\mylog{\lambda}\right)_{\lambda\in\nonz{\bC}}$.

\begin{defi}\label{def:omdr}
We define a difference ring by
\begin{equation}\label{eq:def-R}
  \omdr = \bigoplus_{\lambda \in \nonz{\bKbar}}\, \mylog{\lambda}\puiseux\bKbar ,
\end{equation}
where $\mylog{\lambda}$~is an eigenvector of~$M$ with respect to the eigenvalue~$\lambda$.

Given a subfield~$\bL$ of~$\bKbar$, we also define a difference subring~$\omdr_\bL$ of~$\omdr$ by
\begin{equation}\label{eq:def-R_L}
  \omdr_\bL = \bigoplus_{\lambda \in \nonz{\bL}} \, \mylog{\lambda}\puiseux\bL ,
\end{equation}
so that $\omdr$~is the special case $\omdr_{\bKbar}$.

\end{defi}

For an operator of~$\rat{\bKbar}\langle M\rangle$,
\begin{equation}\label{eq:L-in-R}
L = L(x,M) = \sum_{k=0}^r \ell_k(x) M^k ,
\end{equation}
and an element of~$\omdr$,
\begin{equation}\label{eq:y-in-D}
y = \sum_{\lambda\in S} \mylog{\lambda} \, p_\lambda(x) ,
\end{equation}
where the set~$S$ is a finite subset of~$\bKbar$
and the~$p_\lambda(x)$ are Puiseux series,
the action of~$L$ on~$y$ has a special structure given by
\begin{equation}\label{eq:action-L-y}
L(x,M) \, y = \sum_{\lambda\in S} \mylog{\lambda} L(x,\lambda M) \, p_\lambda(x) .
\end{equation}
This formula decomposes~\eqref{eq:linear} into equations on the~$p_\lambda(x)$ as follows.
\begin{lem}\label{lem:split-eqn-to-equiv}
For an operator~$L \in \rat{\bKbar}\langle M\rangle$ as in~\eqref{eq:L-in-R}
and an element~$y$ of~$\omdr$ as in~\eqref{eq:y-in-D},
the equality $L(x,M)\, y = 0$ is equivalent to the finite set of equalities
\begin{equation}\label{eq:indep-eqns}
L(x,\lambda M) \, p_\lambda(x) = 0,\qquad \lambda\in S .
\end{equation}
\end{lem}
\begin{proof}
This follows directly from \eqref{eq:action-L-y},
by the uniqueness of coefficients in~$\omdr$.
\end{proof}

The following lemma can be deduced from
\textcite[Theorem~35]{Roques-2018-ARB} in the case $\bL=\bQbar$.
\begin{lem}\label{lem:R-contants-field}
The field of constants of the ring~$\omdr_\bL$ is~$\bL$.
\end{lem}
\begin{proof}
By Lemma~\ref{lem:split-eqn-to-equiv},
the equation $My = y$ is equivalent to
$\lambda M p_\lambda = p_\lambda$ for $\lambda\in\bKbar$.
For each~$\lambda$, if $p_\lambda$~is nonzero,
its valuation must be~$0$
and its leading coefficient~$c$ must satisfy~$\lambda c = c$,
implying~$\lambda = 1$.
This means that $y$~is in~$\bL$.
We conclude that the
constant field of~$\omdr_\bL$ is $\bL$~itself.
\end{proof}

\subsection{Bounds}
\label{sec:bounds}

The goal of this section is to bound the valuation and ramification order of
the Puiseux series solutions of~\eqref{eq:linear} in~$\omdr$.
We need to introduce notions related to a \emph{lower} Newton polygon attached to~$L$,
which the reader should not confuse with parallel notions attached to the \emph{upper} Newton polygon
that will be defined in~\S\ref{sec:bgpf} (Definition~\ref{def:upper-newton}).

\begin{defi}\label{def:lower-newton}
To the polynomial
\[
  L = \sum_{k=0}^r \ell_k(x) M^k = \sum_{k,j} \ell_{k,j} x^j M^k
\]
we associate
the set of points $(b^k,j) \in \bR^2$ such that $\ell_{k,j}\neq0$.
The lower convex hull of this set is called the \emph{lower Newton polygon} (of~$L$).
\end{defi}

The following lemma is borrowed from
\parencite[Lemma~2.2]{ChyzakDreyfusDumasMezzarobba-2018-CSL}.

\begin{lem}\label{lem:admissibility}
The valuation of any Puiseux series solution of~\eqref{eq:linear}
is the opposite of the slope of an edge of the lower
Newton polygon of\/~$L$.
Moreover the coefficients~$\ell_{k,j}$
corresponding to all points $(b^k,j)$ lying on the edge
must add up to zero.
\end{lem}

For a solution~$y$ of~\eqref{eq:linear}, that is, satisfying $L(x,M)\, y = 0$,
to be of the form $\mylog{\lambda} p(x)$ with $p(x)$ a Puiseux series,
the series~$p(x)$ must by the formula~\eqref{eq:action-L-y}
be a solution of $L(x,\lambda M)\, p = 0$.
The change from~$L(x,M)$ to~$L(x,\lambda M)$ does not modify the
Newton polygon, but the coefficients at abscissa~$b^k$
are multiplied by~$\lambda^k$.
So, by Lemma~\ref{lem:admissibility},
the only~$\lambda$ for which a nonzero solution of the form $\mylog{\lambda} p(x)$ may exist
are  the (necessarily nonzero) roots of the characteristic polynomials defined by~\eqref{eqn:lower-char-poly} in the following definition.

\begin{defi}\label{def:lower-char-poly}
The edges of the lower Newton polygon of~$L$ are indexed so that their slopes form an increasing sequence.
Given the abscissae $b^{k_1}$ and~$b^{k_2}$, with~$k_1\leq k_2$, of the $j$th edge,
the polynomial~$\chi_j$ defined by
\begin{equation}\label{eqn:lower-char-poly}
\chi_j(X) := \sum_{k=k_1}^{k_2}\ell_{k,j}X^{k-k_1}
\end{equation}
is called the \emph{characteristic polynomial} of the $j$th edge.
We define~$\Lambda$ as the union of the roots of these polynomials:
\begin{equation}\label{eq:def-Lambda}
\Lambda = \{ \lambda \in \bKbar : \exists j, \ \chi_j(\lambda) = 0 \} .
\end{equation}
Given~$\lambda$ in~$\Lambda$, the $j$th edge will be called \emph{$\lambda$-admissible} (for~$L$)
if $\chi_j(\lambda) = 0$.
\end{defi}

The present $1$-admissible edges were simply called admissible in \parencite{ChyzakDreyfusDumasMezzarobba-2018-CSL}.
Lemma~\ref{lem:admissibility} extends slightly
to accommodate a logarithmic part,
leading to the following result.
The omitted proof parallels that of \parencite[Lemma~2.2]{ChyzakDreyfusDumasMezzarobba-2018-CSL}.

\begin{lem}\label{lem:extd-admissibility}
The valuation of the Puiseux series~$p(x)$ in any nonzero solution of~\eqref{eq:linear}
of the form $\mylog\lambda p(x)$
is the opposite of the slope of a $\lambda$-admissible edge of the lower
Newton polygon of\/ $L$.
\end{lem}

As already noted for~$\lambda = 1$ in \parencite{ChyzakDreyfusDumasMezzarobba-2018-CSL},
$\lambda$-admissibility is only a necessary condition
for the existence of a series solution with logarithmic part~$\mylog{\lambda}$.

\begin{prop}\label{lem:single-lambda-equiv}
For an operator~$L \in \rat{\bKbar}\langle M\rangle$ as in~\eqref{eq:L-in-R}
and an element $y = \mylog{\lambda} \, p(x)$ of~$\omdr$ given by $\lambda \in \nonz\bKbar$ and~$p \in \puiseux\bKbar$,
the equality $L(x,M)\, y = 0$ is equivalent to the equality
\begin{equation}\label{eq:single-eqn}
L(x,\lambda M) \, p(x) = 0 .
\end{equation}
Moreover, \eqref{eq:single-eqn}~can have a nonzero solution
only if $\lambda$~is the root of some characteristic polynomial of~$L$.
\end{prop}

\begin{proof}
The stated equivalence is a mere specialization of Lemma~\ref{lem:split-eqn-to-equiv}.
The necessary condition on~$\lambda$ has been proved
above when introducing characteristic polynomials.
\end{proof}

In order to discuss ramifications, we import
\parencite[Prop.~2.19]{ChyzakDreyfusDumasMezzarobba-2018-CSL},
which deals with $1$-admissible edges.
In doing so, we slightly reformulate the result
so as to describe solutions over an algebraic extension~$\bL$ of~$\bK$.

\begin{lem}\label{lem:puiseux-denominator-bound}
Any Puiseux series solution of~\eqref{eq:linear} is an element of\/ $\rfls\bKbar{q_1}$,
where $q_1$~is the lcm of the denominators of the slopes of those $1$-admissible edges
such that said denominators are coprime with~$b$.
Furthermore, a $\bKbar$-vector basis of the solutions can be written
in~$\rfls\bK{q_1}$.
Given an intermediate field\/~$\bL$,
such a $\bKbar$-vector basis is also
an $\bL$-vector basis for the $\bL$-space of solutions in~$\puiseux{\bL}$.
\end{lem}
\begin{proof}
The first assertion of the statement is exactly the proposition we referred to,
using $\bKbar$~as a base field.
The second and third assertions are implicit in
\parencite{ChyzakDreyfusDumasMezzarobba-2018-CSL},
which finds solutions by solving linear systems over~$\bK$.
\end{proof}

\begin{defi}
More generally, we denote by $q_\lambda$ the lcm of the denominators coprime with~$b$
of the slopes of the $\lambda$-admissible edges of the lower Newton polygon
of~$L$, that is
\begin{multline}\label{eq:q-lambda}
q_\lambda =
\operatorname{lcm} \, \{q \in \bN : \exists j,p, \ \gcd(b,q) = \gcd(p,q) = 1, \ \chi_j(\lambda) = 0, \\
\text{and the $j$th edge has slope~$p/q$}\} .
\end{multline}
We call~$q_\lambda$ the \emph{ramification bound} with respect to~$\lambda$.
\end{defi}

We can then prove the following result on solutions in~$\omdr_\bL$ (Definition~\ref{def:omdr}).

\begin{prop}\label{prop:structure-sol-space-lin-eq}
Given an intermediate field~$\bL$,
the solution space of~\eqref{eq:linear} in~$\omdr_\bL$ is included in
\[
  \bigoplus_{\lambda\in\Lambda\cap\bL}
  \mylog{\lambda} \, \rfls{\bL}{q_\lambda}.
\]
Furthermore, an $\bL$-basis of this space can be written in
\[
  \coprod_{\lambda\in\Lambda\cap\bL}
  \mylog{\lambda} \, \rfls{\bK[\lambda]}{q_\lambda}.
\]
\end{prop}

\begin{proof}
For a nonzero solution~$y$ of~$L$ in~$\omdr_\bL$,
Lemma~\ref{lem:split-eqn-to-equiv} splits the relation $L y = 0$
into the equivalent finite set of equations~\eqref{eq:indep-eqns}.
This proves the first part,
where the direct sum of the solution space is induced
by the direct sum in the expression~\eqref{eq:def-R_L} of~$\omdr_\bL$.
Next, according to Lemma~\ref{lem:admissibility},
the equation~\eqref{eq:indep-eqns} for a given~$\lambda\in\bL$ has a nonzero solution
only if there exists a $\lambda$-admissible edge, that is, only if
$\lambda$~is an element of~$\Lambda$.
Moreover, the change from $L(x,M)$ to $L(x,\lambda M)$ modifies the linear system
that is used in the proof of Lemma~\ref{lem:puiseux-denominator-bound} to determine the solutions
by moving its coefficients to the algebraic extension~$\bK[\lambda]$.
As a result,
a basis of the Puiseux-series solutions~$p_\lambda$
can be found with coefficients in~$\bK[\lambda]$
and ramification order (bounded by)~$q_\lambda$.
The result is proved.
\end{proof}

The following corollary is already proved as Lemma~\ref{lem:has-Property-U}:
we obtain here an alternative proof
by considering properties of the lower Newton polygon.

\begin{cor}\label{cor:dimension-vs-order}
Given an intermediate field~\/$\bL$,
the solution set in~$\omdr_\bL$ of the linear equation~\eqref{eq:linear}
is an $\bL$-vector space
whose dimension over the field\/~$\bL$
is at most the order~$r$ of the equation.
\end{cor}

\begin{proof}
Proposition~\ref{prop:structure-sol-space-lin-eq} provides us
with an $\bL$-basis of the solutions of~\eqref{eq:linear} in~$\omdr_\bL$,
which can be chosen such that for each~$\lambda$
the basis elements of the form $\mylog{\lambda} \, p(x)$ involve series~$p(x)$ with distinct valuations.
By Lemma~\ref{lem:extd-admissibility}, these valuations are
opposites of slopes of distinct edges of the Newton polygon
whose characteristic polynomials~$\chi_j$ all have~$\lambda$ as a root.
Therefore, the total number of elements of the basis is bounded by
$\sum_j \deg \chi_j \leq r$.
\end{proof}

The order~$r$ is not always reached in Corollary~\ref{cor:dimension-vs-order},
as shown by the following two examples.

\begin{ex}
For any radix~$b\geq2$, the operator $L = M^2-2M+1$ admits the series solution~$1$,
and no nonconstant solution.
This reflects the fact that the nonhomogeneous equation $My-y = 1$
has no solution in~$\omdr$,
while any solution of this equation needs to be killed by~$L$.
\end{ex}

\begin{ex}
Similarly, when~$b\geq2$, $x^{b-1}M^2 - (1+x^{b-1})M + 1$ admits the series solution~$1$,
and otherwise would require a Hahn series $x^{-1/b}+x^{-1/b^2}+x^{-1/b^3}+\dotsb$
to get a second dimension of solutions,
the latter solving $My-y = 1/x$.
See \parencite{FaverjonRoques-2024-HSM} for a definition and developments.
\end{ex}

\subsection{Puiseux-hypergeometric solutions}
\label{sec:puiseux}

In the present section,
we focus on the field $F = \puiseux\bL$ to obtain the Puiseux series
with coefficients in~$\bL$ that are solutions of the Riccati equation.
We describe the $\puiseux\bL$-hypergeometric elements in~$\omdr_\bL$
and deduce that $\omdr_\bL$~satisfies Hypothesis~\ref{def:propertyP} over~$\puiseux\bL$.

\begin{lem}\label{lem:unique-monic-series}
Fix $q\in \nonz{\bN}$ and $u\in \rfls\bL{q}$
of valuation~$\alpha$ and with leading coefficient\/~$\lambda$.
If $\lambda \neq 1$, then the equation $M v = u v$ admits no nonzero solution in~$\puiseux\bL$.
Otherwise,
it
admits a unique series solution
$v$ in $x^{\alpha/(b-1)} \rfps\bL{q}$
of valuation $\alpha/(b-1)$ and with leading coefficient\/~$1$, and the solution space of this equation in~$\puiseux\bL$
is exactly the one-dimensional space\/~$\bL v$.
\end{lem}

\begin{proof}
The Puiseux series~$u$ can be written in the form $u = \lambda x^{\alpha} \hat u$
for some
$\hat u = \sum_{k\in\bN} \hat u_k x^{k/q} \in \rfps\bL{q}$
such that $\hat u_0 = 1$.
Assume that the equation $M v = u v$ admits a nonzero solution $v \in \puiseux\bL$ with valuation~$\beta$ and leading coefficient~$v_0 \neq 0$.
Taking leading terms on both sides yields $v_0 x^{b\beta} = \lambda v_0 x^{\alpha+\beta}$,
which implies~$\lambda = 1$.
This proves the absence of nonzero solutions if~$\lambda \neq 1$
and the announced valuation of~$v$ if~$\lambda = 1$.

We continue with the latter case.
Write $v$ in the form $x^{\alpha / (b-1)}\hat v$ for a series
$\hat v = \sum_ {k\in\bN}\hat v_k x^{k/q}$ and $\hat v_0 = v_0 \neq 0$,
the equation is equivalent to the relation
\begin{equation*}
\sum_{k\in\bN} \hat v_k x^{bk/q} =
\sum_{n\in\bN} \sum_{i\in\bN} \hat u_{n - i} \hat v_i \,  x^{n/q} ,
\end{equation*}
which in turn is equivalent to the recurrence formula
\[ \hat v_n = \hat v_{n / b} - \sum_{i=0}^{n-1} \hat u_{n - i} \hat v_i \]
for $n\in\nonz{\bN}$,
where $\hat v_{n/b}$ is understood to be zero if $n/b$~is not in~$\bN$.
The sequence $(\hat v_n)_{n\in\bN}$, and therefore the solution~$v$,
is fully determined by the choice of $\hat v_0 = 1$.

Next, for any Puiseux series solution~$w$ of the equation $Mw = uw$,
the quotient $w/v$ is in the constant field of~$\puiseux\bL$,
that is, it lies in~$\bL$.
\end{proof}

The following result is to be found
as \parencite[Theorem~35, top of page~342]{Roques-2018-ARB},
where~$\omdr$ is named~$B$ and is defined in \parencite[Lemma~36]{Roques-2018-ARB}.
We reproduce it here for the readers' convenience,
while removing the unnecessary hypothesis
of an algebraically closed field~$\bL$.
\begin{lem}\label{lem:D-is-simple}
$\omdr_\bL$~is a simple difference ring.
\end{lem}
\begin{proof}
Let $J$ be a nonzero difference ideal of~$\omdr_\bL$.
Let $y = \sum_{\lambda \in \nonz{\bL}} \mylog{\lambda} p_\lambda$
for $p_\lambda \in \puiseux\bL$
be a nonzero element of~$J$
such that the cardinality of the support of~$(p_\lambda)_{\lambda \in \nonz{\bL}}$ is minimal.
Let $\lambda_0$~be such that
$p_{\lambda_0} \neq 0$
and replace~$y$ with~$y/p_{\lambda_0}$,
so that we assume that~$p_{\lambda_0} = 1$.
Because $p_{\lambda_0}$~appears with zero coefficient in~$y - M y \in J$
and by the minimality assumption,
$y - M y = 0$.
Consequently, for all~$\lambda \in \nonz{\bL}$,
$p_\lambda = \lambda M p_\lambda$.
For~$\lambda \neq 1$, we conclude~$p_\lambda = 0$
by Lemma~\ref{lem:unique-monic-series}.
Consequently, $y$~reduces to~$p_1$, which is invertible,
thus proving~$J = \omdr_\bL$.
\end{proof}

\begin{cor}\label{cor:DL-1-univ-U}
$\omdr_\bL$ is an extension satisfying Hypothesis~\ref{def:propertyP} over
any intermediate difference field~$F$ between $\rat\bL$ and~$\puiseux\bL$,
so that Theorem~\ref{thm:structure-main} applies to~$\omdr_\bL$, viewed as an extension of~$F$.
\end{cor}
\begin{proof}
This follows directly from
Lemma~\ref{lem:R-contants-field},
 Lemma~\ref{lem:unique-monic-series},
and Lemma~\ref{lem:D-is-simple}.
\end{proof}

We will use
$F = \puiseux\bL$ in Theorem~\ref{thm:structure-puiseux} below,
and $\ramrat\bL$ in Theorem~\ref{thm:structure-rational-ramified}.

\begin{prop}\label{prop:puiseux-hypergeometric-element}
The $\puiseux\bL$-hypergeometric elements in~$\omdr_\bL$ are
its elements $y = \mylog{\lambda} p(x)$
where $\lambda$~is in~$\nonz{\bL}$
and $p(x)$~is a Puiseux series in\/~$\puiseux\bL$.
Two such nonzero elements are $\puiseux\bL$-similar if and only if they share the same~$\lambda$.
\end{prop}

\begin{proof}
Let $y$ be a nonzero $\puiseux\bL$-hypergeometric element in~$\omdr_\bL$,
together with the corresponding Puiseux series~$u$ satisfying the equation~$My = uy$.
The series~$u$ lies in a subfield~$\rfls\bL{q}$ with $q$ a positive integer.
Let its leading term be denoted~$\lambda_0 x^{m/q}$,
for suitable $\lambda_0 \in \bL$ and~$m \in \bZ$.
By Lemma~\ref{lem:split-eqn-to-equiv},
writing~$y$ as in~\eqref{eq:y-in-D}
for a finite subset~$S$ of~$\nonz{\bL}$ and some~$p_\lambda(x)$ in~$\puiseux\bL$
leads to the equations
$\lambda M p_\lambda = u p_\lambda$ for $\lambda\in S$.
Any such equation can have a nonzero solution
only if $\lambda$~is equal to~$\lambda_0$,
forcing~$y$ to be of the form
$y = \mylog{\lambda_0} p_{\lambda_0}(x)$
where~$p_{\lambda_0}(x) \in \puiseux\bL$ is nonzero.
Conversely, any~$y$ of that form
is $\puiseux\bL$-hypergeometric in~$\omdr_\bL$.

Now suppose that
$y_1 = \mylog{\lambda_1} p_1(x)$ and $y_2 = \mylog{\lambda_2} p_2(x)$
are $\puiseux\bL$-similar and nonzero.
The similarity implies the linear dependence of $\mylog{\lambda_1}$ and~$\mylog{\lambda_2}$ over $\puiseux\bL$,
and therefore the equality~$\lambda_1 = \lambda_2$.
Conversely two nonzero hypergeometric solutions of~\eqref{eq:linear} in~$\omdr_\bL$
with the same $\lambda\in\bL$ are similar by definition.
\end{proof}

\begin{rem}\label{rem:puiseux-hypergeometric-element-with-q}
The same approach,
supplemented with Lemma~\ref{lem:unique-monic-series},
shows more precisely
that the  $\rfls\bL{q}$-hypergeometric elements in~$\omdr_\bL$ are of the form
\begin{equation}\label{eq:puiseux-hypergeometric-element}
  y = \mylog{\lambda} x^{m /(q(b-1))} s(x^{1/q})
\end{equation}
with $\lambda$ in~$\nonz{\bL}$,
$q$~a positive integer, $m$~an integer,
and $s(x)$~a series of valuation\/~$0$ in\/~$\fps\bL$.
Observe (e.g., by making~$m=1$) that $My/y$ may have a lower ramification order than~$y$.
Consequently,
obtaining \emph{all} Laurent power series solutions to~\eqref{eq:riccati} in the form~$u = My/y$
generally requires to consider
\emph{ramified} hypergeometric solutions~$y$ of~\eqref{eq:linear}.
A simple example is provided by $L = M - u$ for~$b = 5$ and $u = x  (1+x^3)$,
with solution $y = x^{1/4} \prod_{k\geq0}(1+x^{3\cdot 5^k})^{-1}$.
\end{rem}

\begin{rem}
We could have written the series~$s$ in Proposition~\ref{prop:puiseux-hypergeometric-element}
as an infinite product.
Indeed, this representation will prove useful later:
in Proposition~\ref{prop:ramified-rational-hypergeometric-element}
where we obtain an infinite product of rational functions;
in Proposition~\ref{prop:rational-ramified-hypergeometric}
where we provide a characterization of the similarity classes.
\end{rem}

\begin{defi}\label{def:Lambda'}
Let $\Lambda'$ denote the finite set of~$\lambda \in \Lambda$
(see~\eqref{eq:def-Lambda})
such that
there exists a nonzero solution of~\eqref{eq:linear} in~$\mylog{\lambda} \puiseux\bKbar$.
For an intermediate field\/~$\bL$ and $\lambda \in \Lambda' \cap \bL$, define~$\fH_{\bL,\lambda}$ to be
the set of hypergeometric solutions in~$\mylog{\lambda} \puiseux{\bL}$.
Let\/~$\ricsol_{\bL,\lambda} \subseteq \puiseux\bL$ denote
the image of~$\nonz{(\fH_{\bL,\lambda})}$ under the map $y\mapsto My/y$.

\end{defi}

The elements of~$\fH_{\bL,\lambda}$
are necessarily in~$\mylog{\lambda} \rfls{\bL}{q_\lambda}$
by Proposition~\ref{prop:structure-sol-space-lin-eq}.
With this notation, Theorem~\ref{thm:structure-main} specializes as follows.

\begin{thm}\label{thm:structure-puiseux}
The solution set\/~$\ricsol_{\puiseux\bL}$ of the Riccati equation~\eqref{eq:riccati}
in the ramified rational function field\/~$\puiseux\bL$ is the disjoint union
\begin{equation}\label{eq:R-disj-union}
  \ricsol_{\puiseux\bL} = \coprod_{\lambda\in\bL\cap\Lambda'} \ricsol_{\bL,\lambda}
\end{equation}
indexed by the finite set\/~$\bL\cap\Lambda'$.
Each\/~$\ricsol_{\bL,\lambda}$ is a set of series from\/~$\rfls\bL{q_\lambda}$ with leading coefficient~$\lambda$
and is one-to-one rationally parametrized by the\/ $\bL$-projective space\/ $\projzn{\fH_{\bL,\lambda} }$.
The corresponding parametrization is obtained
by restricting the map $y\mapsto My/y$ from~$\nonz{(\fH_{\bL,\lambda})} $ to its image in\/~$\puiseux\bL$
and projectivizing its source.
Moreover the dimensions of the\/ $\bL$-projective spaces add up to a number
that is at most the order~$r$ of the linear equation~\eqref{eq:linear}
minus the cardinality of\/~$\bL\cap\Lambda'$.
\end{thm}

\begin{proof}
By Corollary~\ref{cor:DL-1-univ-U},
the ring~$\omdr_\bL$ satisfies Hypothesis~\ref{def:propertyP} as an extension of~$\puiseux\bL$,
so Theorem~\ref{thm:structure-main} applies
to $\dring = \omdr_\bL$ and~$F = \puiseux\bL$.
Points \ref{it:vect-H} to~\ref{it:param-U} of that theorem justify
the partitioning according to~\eqref{eq:R-disj-union}
as well as the form of the parametrization.
Moreover,
from the inclusion $\fH_{\bL,\lambda} \subseteq \mylog{\lambda} \rfls{\bL}{q_\lambda}$
follows the inclusion
$\ricsol_{\bL,\lambda} \subseteq \rfls{\bL}{q_\lambda}$.
Point~\ref{it:disjoint-U} proves the bound on dimensions.
\end{proof}

\section{Ramified rational solutions to the Riccati equation}
\label{sec:rational-ramified-sol}

We now study the $F$-hypergeometric solutions of~\eqref{eq:linear}
for $F = \ramrat\bL$, or, equivalently, the ramified rational solutions of
\eqref{eq:riccati} with coefficients in some intermediate field~$\bL$.

\subsection{Hypergeometric elements}
\label{sec:hypergeom}

We begin by characterizing the $\ramrat\bL$-hypergeo\-metric elements
of~$\omdr_\bL$ (Definitions \ref{def:F-hyperg} and~\ref{def:omdr}).

\begin{prop}\label{prop:ramified-rational-hypergeometric-element}
The nonzero\/ $\ramrat\bL$-hypergeometric elements in $\omdr_\bL$
are the elements
\begin{equation}\label{eq:ramified-rational-hypergeometric-element}
  y = c \, \mylog{\lambda} x^{m/q} f(x^{1/q})
\end{equation}
with $c$~and~$\lambda$~in\/~$\nonz{\bL}$, $m$~an integer,
$q$~a positive integer,
and $f(x)$~a formal power series that is expressible as an infinite product
\begin{equation}\label{eq:inf-product}
  f(x) = \prod_{k\geq 0} 1/g(x^{b^k}) \in \fps\bL
\end{equation}
for some $g(x) \in \rat\bL$ satisfying $g(0) = 1$.
\end{prop}

\begin{proof}
Any nonzero $\ramrat\bL$-hypergeometric element~$y$ in~$\omdr_\bL$
is $\puiseux\bL$-hyper\-geo\-met\-ric.
By Proposition~\ref{prop:puiseux-hypergeometric-element},
it can be written $y = \mylog{\lambda} x^{m/q} s(x^{1/q})$
with $\lambda$ in~$\nonz{\bL}$, $q$ a positive integer, $m$ an integer,
and $s(x)$~in~$\fps\bL$ of valuation~$0$.
One then has
$My/y = \lambda x^{(b-1)m/q} s(x^{b/q})/s(x^{1/q}) \in \rfls\bL{q} \cap \ramrat\bL$,
so $g(x) := s(x^b)/s(x)$ is a rational function.
The infinite product~\eqref{eq:inf-product} for this value of~$g$ converges
to~$s$,
hence $y$~is of the form~\eqref{eq:ramified-rational-hypergeometric-element}
with~$f = s$.
\end{proof}

\begin{rem}
If the hypergeometric element~$y$ is such that $My/y$~is in~$\rfls{\bL}{q'}$,
then $q = q'(b-1)$ can be used
in formula~\eqref{eq:ramified-rational-hypergeometric-element},
as a consequence
of Remark~\ref{rem:puiseux-hypergeometric-element-with-q}.
\end{rem}

\subsection{Similarity classes}
\label{sec:similarity-classes}

The set of $\ramrat\bL$-hypergeometric solutions of~\eqref{eq:linear}
decomposes into $\ramrat\bL$-similarity classes that refine the
$\puiseux\bL$-similarity classes~$\fH_{\bL,\lambda} $ described in~\S\ref{sec:concrete} (Definition~\ref{def:Lambda'}).
The following proposition gives a normal form for these new similarity classes.
Recall from Proposition~\ref{prop:structure-sol-space-lin-eq} that the ramification bound~$q_\lambda$
defined by~\eqref{eq:q-lambda} is a bound on the ramification orders
of solutions of~\eqref{eq:linear} in~$\mylog\lambda \, \puiseux\bL$,
and recall as well the definition of~$\Lambda'$ in Definition~\ref{def:Lambda'}.

\begin{prop}\label{prop:rational-ramified-hypergeometric}
Any nontrivial\/ $\ramrat\bL$-similarity class~$\nonz{\fH}$
of\/ $\ramrat\bL$-hyper\-geo\-met\-ric solutions of~\eqref{eq:linear} uniquely determines
a tuple $(\lambda, g, (p_j)_{1\leq j\leq s})$ \hbox{made of}:
\begin{itemize}
  \item $\lambda$ in\/ $\bL\cap\Lambda'$,
  \item an integer~$s>0$ and Laurent polynomials $p_1$, \dots, $p_s$ of\/~$\bL[x,x^{-1}]$,
  \item a rational function~$g$ in\/~$\rat\bL$ satisfying $g(0) = 1$,
\end{itemize}
and such that:
\begin{enumerate}
\item\label{it:basis-using-p} the $s$ generalized series
\begin{equation}\label{eq:rational-ramified-hypergeometric}
  \mylog{\lambda}
    \times p_j(x^{1/q_\lambda}) \times
      \prod_{k\geq 0} 1/g(x^{b^k/q_\lambda}) ,
  \qquad 1 \leq j \leq s,
\end{equation}
form a basis of the $\bL$-vector space~$\fH$,
  \item\label{it:family-of-p} the family $(p_1, \dots, p_s)$ is in
    reduced echelon form w.r.t.\ ascending degree
    (meaning in particular that the coefficient of minimal degree of~$p_j$ is\/~$1$),
  \item\label{it:coprime-p} the elements $p_1$, \dots, $p_s$ are coprime in\/~$\bL[x,x^{-1}]$
    (in other words, the $p_j/x^{\val p_j}$ are coprime in\/~$\pol\bL$).
\end{enumerate}
\end{prop}

\begin{proof}
The $\ramrat\bL$-similarity class~$\nonz{\fH}$
of $\ramrat\bL$-hypergeometric solutions of~\eqref{eq:linear}
is contained in an $\puiseux\bL$-similarity class of
$\puiseux\bL$-hypergeometric elements of~$\omdr_\bL$, which, by
Proposition~\ref{prop:puiseux-hypergeometric-element},
is of the form
$\mylog\lambda \, \nonz{\puiseux\bL}$ for some~$\lambda \in \nonz{\bL}$.
Since, additionally, $\fH$~is contained in the solution space
of~\eqref{eq:linear},
Proposition~\ref{prop:structure-sol-space-lin-eq} implies
$\fH \subseteq \mylog\lambda \, \rfls\bL{q_\lambda}$.
Additionally, $\fH$~has finite dimension over~$\bL$
by Theorem~\ref{thm:structure-main}.

For any element $y \in \nonz{\fH}$, use
Proposition~\ref{prop:ramified-rational-hypergeometric-element}
to write
\begin{equation}\label{eq:ramified-factorization}
  y = c \, \mylog{\lambda} x^{m/q} f(x^{1/q}), \qquad
  f(x) = \prod_{k\geq 0} 1/g(x^{b^k}),
\end{equation}
with $c \in \bL$, $g(x) \in \bL(x)$, and $g(0) = 1$.
We contend that we can without loss of generality
make $q$ equal to the ramification bound~$q_\lambda$ in~\eqref{eq:ramified-factorization}.
First, the valuation~$m/q$ of~$y/\mylog{\lambda}$
is in~$q_\lambda^{-1}\bZ$,
and can thus be written $m/q = \tilde m/q_\lambda$ for~$\tilde m \in \bZ$.
Second, $f(x^{1/q}) = y / (\mylog{\lambda} x^{\tilde m/q_\lambda})$
is in~$\rfls\bL{q_\lambda}$,
so that $h(x) := f(x^{q_\lambda/q})$~is in~$\fls\bL$.
From $Mf/f = g \in \rat\bL$ follows the membership of
$g(x^{q_\lambda/q}) = Mh/h$ in $\bL(x^{q_\lambda/q}) \cap \fls\bL$.
Write $q_\lambda/q = n/d$ in lowest terms
and observe that $g$~is a series in~$x^d$.
Define $\tilde g(x) := g(x^{n/d})$, a rational series satisfying $\tilde g(0) = 1$.
Define $\tilde f(x) := \prod_{k\geq 0} 1/\tilde g(x^{b^k})$ and observe
$f(x^{1/q}) = \tilde f(x^{1/q_\lambda})$.
We have obtained a new expression of~$y$ of the form~\eqref{eq:ramified-factorization}
with $(q, m, g)$ replaced with~$(q_\lambda, \tilde m, \tilde g)$.
In particular, we have made~$q$ independent of~$y$.

Pick an element $y_0 \in \nonz{\fH}$ and obtain
its decomposition of the form~\eqref{eq:ramified-factorization},
thus fixing $(q, m, g)$ to some~$(q_\lambda, m_0, g_0)$.
Replace~$y_0$ with~$y_0/c$ so as to set~$c$ to~$1$.
Since $\nonz{\fH}$ is an $\ramrat\bL$-similarity class,
the set $V := \{ y/y_0 : y \in \fH \}$
is an $\bL$-subspace of $\ramrat\bL$,
and thus, by the previous paragraph,
a finite-dimensional $\bL$-subspace of $\rrat\bL{q_\lambda}$.
Let $a(x^{1/q_\lambda})$
be the least common denominator of the elements of~$V$,
with valuation coefficient normalized to one.
Write
$a(x^{1/q_\lambda}) = x^{v/q_\lambda} \tilde a(x^{1/q_\lambda})$
for $\tilde a$ of valuation zero.
Then $a(x^{1/q_\lambda}) V$ is a finite-dimensional subspace of
$\bL[x^{1/q_\lambda}]$;
as such, it admits a basis
$(\tilde p_1(x^{1/q_\lambda}), \dots, \tilde p_s(x^{1/q_\lambda}))$
in reduced echelon form w.r.t.\ ascending degree.
The polynomials $\tilde p_1, \dots, \tilde p_s$ are coprime by the minimality
of~$a$.
Setting $g = g_0 a/Ma$ and $p_j = x^{m_0-v} \tilde p_j$, the
expressions~\eqref{eq:rational-ramified-hypergeometric}
form a basis of~$\fH$ that satisfies the conditions
\ref{it:family-of-p} and~\ref{it:coprime-p}.

We already have noticed that $\lambda$~is uniquely determined by~$\fH$.
Suppose that $\fH$~admits a second basis
of the form~\eqref{eq:rational-ramified-hypergeometric},
with parameters $g', p'_1, \dots, p'_s$ also satisfying
\ref{it:family-of-p} and~\ref{it:coprime-p}.
Let $f' = \prod_{k\geq 0} 1/g'(x^{b^k})$.
Then $(p_1, \dots, p_s)$ and $(p'_1 f'/f, \dots, p'_s f'/f)$
are two bases of the same $\bL$-vector space.
In particular, the $p'_j f'/f$ are Laurent polynomials.
A Bézout relation $\sum_j u_j p'_j = 1$ for Laurent polynomials~$u_j$
then implies that $f'/f = \sum_j u_j (p'_j f'/f)$ is a Laurent polynomial.
By symmetry, $f/f'$~is also Laurent,
and in fact equal to~$1$ because $g(0) = g'(0) = 1$.
Then the uniqueness of the reduced echelon form of a given vector space implies
$(p_1, \dots, p_s) = (p'_1, \dots, p'_s)$.
We have thus proved the uniqueness of the tuple $(\lambda, g, p_1, \dots, p_s)$
under conditions
\ref{it:basis-using-p},
\ref{it:family-of-p}, and~\ref{it:coprime-p}.
\end{proof}

A direct consequence of Proposition~\ref{prop:rational-ramified-hypergeometric}
is that for a given~\eqref{eq:linear},
$\ramrat\bL$-similarity classes are uniquely identified by the pairs~$(\lambda, g)$
extracted from the tuples~$(\lambda, g, p)$ of the proposition.
Indeed,
if two classes share the same~$(\lambda, g)$, then,
in view of the form~\eqref{eq:rational-ramified-hypergeometric} of the bases associated with the two classes,
any element of one class must be $\ramrat\bL$-similar to any element of the other class.
These two classes must therefore be equal.
This leads us to the following definition that the reader will compare with Definition~\ref{def:Lambda'}.

\begin{defi}\label{def:classes-L-lambda-g}
Given a nontrivial\/ $\ramrat\bL$-similarity class
of\/ $\ramrat\bL$-hyp\-er\-geo\-met\-ric solutions of~\eqref{eq:linear},
let $(\lambda, g, (p_j)_{1\leq j\leq s})$ be
the uniquely determined tuple of Proposition~\ref{prop:rational-ramified-hypergeometric}.
Then the similarity class will be denoted~$\fH_{\bL,\lambda,g}$.
We will also write~$\ricsol_{\bL,\lambda,g}$ for the image of~$\nonz{\fH_{\bL,\lambda,g}}$ under $y \mapsto My/y$.
\end{defi}

Summing up the previous discussion, the full solution set of the Riccati
equation can be described as follows.

\begin{thm}\label{thm:structure-rational-ramified}
The solution set\/~$\ricsol_{\ramrat\bL}$ of the Riccati
equation~\eqref{eq:riccati}
in~$\ramrat\bL$ is a disjoint union
\begin{equation}\label{eq:structure-rational-ramified}
  \ricsol_{\ramrat\bL} = \coprod_{(\lambda,g)} \ricsol_{\bL,\lambda,g}
\end{equation}
indexed by a finite set of
pairs $(\lambda,g)$, where~$\lambda$ is in~$\bL\cap\Lambda'$
and~$g$ is a rational function satisfying $g(0) = 1$,
such that the~$\fH_{\bL,\lambda,g}$ partition~$\fH_{\bL,\lambda}$.
Each\/~$\ricsol_{\bL,\lambda,g}$ can be parametrized bijectively
by
\begin{equation*}
  u(x) = \lambda\,
        \frac{\sum_{j=1}^s c_j p_j(x^{b/q_\lambda}) }
        {\sum_{j=1}^s c_j p_j(x^{1/q_\lambda}) }
        g(x^{1/q_\lambda})
  \quad\text{for}\quad
  (c_1:\ldots:c_s) \in \bP^{s-1}(\bL)
\end{equation*}
where the~$p_j$ are those in the uniquely determined tuple of Proposition~\ref{prop:rational-ramified-hypergeometric}.
The dimensions~$s-1$ of the projective spaces add up to a number
that is at most the order~$r$ of the linear equation~\eqref{eq:linear}
minus the cardinality of the set of pairs~$(\lambda,g)$ indexing the disjoint union~\eqref{eq:structure-rational-ramified}.
\end{thm}

\begin{proof}
The classes~$\nonz{\fH}$ of Theorem~\ref{thm:structure-main},
applied to $\dring = \puiseux\bL$ and $F = \ramrat\bL$,
are the~$\fH_{\bL,\lambda,g}$ of Definition~\ref{def:classes-L-lambda-g}.
The result then follows, using Remark~\ref{rem:rational-param}
and the explicit basis~\eqref{eq:rational-ramified-hypergeometric} of~$\fH_{\bL,\lambda,g}$
provided by Proposition~\ref{prop:rational-ramified-hypergeometric}.
\end{proof}

\section{Degree bounds for rational solutions}
\label{sec:bounds-for-rational-solutions}

In the present section, we derive
degree bounds for rational solutions of the Riccati equation
in terms of its order~$r$ and coefficient degree~$d$.
Such bounds will be useful
in~\S\ref{sec:approx-syz}
to select relevant Hermite--Padé approximants.

The following easy consequence of Lemma~3.1(c) in
\parencite{ChyzakDreyfusDumasMezzarobba-2018-CSL}
will be used throughout;
we prove it for self-containedness.
\begin{lem}\label{lem:M-Bezout}
Given a field~$K$, $P$ and~$Q$ in~$K[x]$, and any~$j \in \bN$,
$P$ and~$Q$ are coprime if and only if so are $M^j P$ and~$M^j Q$.
\end{lem}
\begin{proof}
If $P$ and~$Q$ are coprime, there exist polynomials $A$ and~$B$
such that the relation $AP + BQ = 1$ holds.
Therefore, $M^jA \, M^jP + M^jB \, M^jQ = 1$,
so that $M^jP$ and~$M^jQ$ are coprime.
Conversely, if $M^jP$ and~$M^jQ$ are coprime,
then there exist polynomials $A$ and~$B$
such that $A \, M^jP + B \, M^jQ = 1$.
If some $G \in K[x]$ divides both $P$ and~$Q$,
then $M^jG$~divides~1, so $G$~must be a constant.
Therefore, $P$ and~$Q$ are coprime.
\end{proof}

\begin{prop}\label{prop:deg-bounds-for-ratfuns}
  Assume that $P / Q \in \rat\bKbar$, with $\gcd(P, Q) = 1$, is a solution
  of~\eqref{eq:riccati}. Then, the following degree bounds hold:
    \begin{equation}\label{eq:bound-P}
     \deg P \leq \Bnum :=
       \begin{cases}
         2d & (b=2) , \\
         {4d}/{b^{r-1}} & (b\ge3) ,
       \end{cases}
    \end{equation}
    \begin{equation}\label{eq:bound-Q}
    \deg Q \leq \Bden :=
      \begin{cases}
        2(1-1/2^{r})d & (b=2) , \\
        {3d}/{b^{r-1}} & (b\ge3) .
      \end{cases}
    \end{equation}
\end{prop}

\begin{proof}
By Lemma~\ref{lem:M-Bezout} over~$K=\bKbar$,
the assumption $\gcd(P, Q) = 1$ implies $\gcd(M^{r-1} P, M^{r-1} Q) = 1$,
that is, $M^{r-1} P$ and~$M^{r-1} Q$ are coprime.
For $n\in \bN$, let $\rfact{n}{y}$ denote the ``rising factorial''
$\prod_{i=0}^{n-1} M^i y$.
In particular, $\rfact{0}{y}=1$.
Multiplying~\eqref{eq:riccati} by $\rfact{r}{Q}$ yields
\begin{equation}\label{eq:cleared-riccati}
\sum_{i=0}^r \ell_i \rfact{i}{P} \rfact{r-i}{(M^i Q)} = 0 .
\end{equation}
Since the factor $M^{r - 1} Q$ appears in all terms on the left-hand side
of~\eqref{eq:cleared-riccati} but the term for~$i = r$,
and since $M^{r - 1} Q $ and $M^{r - 1} P$ are coprime,
$M^{r - 1} Q$ must divide $\ell_r \, \rfact{r-1}{P}$.
Hence, since $\deg \ell_r \leq d$,
the degrees of $Q$ and~$P$ are related by the inequality
\begin{equation}\label{eq:first-bound-PQ}
   b^{r-1} \deg Q \leq d + \frac{b^{r-1} - 1}{b - 1} \deg P .
\end{equation}
Furthermore,
for~\eqref{eq:cleared-riccati} to hold,
at least one term of index~$i<r$
must have a degree greater than or equal to
the degree of the term of index~$r$.
Simplifying by the common factor $\rfact{i}{P}$
yields
\begin{equation*}
\deg\left(\ell_r \, \rfact{r-i}{(M^i P)}\right) \leq
\deg\left(\ell_i \, \rfact{r-i}{(M^i Q)}\right) ,
\end{equation*}
from which follows
\begin{equation*}
   \deg P - \deg Q \leq
   \frac{\deg \ell_i - \deg \ell_r}{(b^{r} - b^{i})/(b-1)} \leq
   \frac{d}{b^{r-1}} .
\end{equation*}
The above inequality yields an upper bound on~$\deg P$,
\begin{equation}\label{eq:second-bound-PQ}
\deg P \leq \frac{d}{b^{r-1}} + \deg Q .
\end{equation}
Substituting it into~\eqref{eq:first-bound-PQ},
we find
\begin{equation*}
b^{r-1} \deg Q \leq d + \frac{b^{r-1} - 1}{b - 1} \left(\frac{d}{b^{r-1}}+\deg Q\right)
= d\frac{b^r-1}{(b - 1)b^{r-1}} +\frac{b^{r-1} - 1}{b - 1}\deg Q .
\end{equation*}
Multiplying by~$b-1$ yields
\begin{equation}\label{eq:third-bound-PQ}
(b^r-2b^{r-1}+1) \deg Q \leq \frac{d}{b^{r-1}} (b^r-1) \leq db .
\end{equation}
For the case~$b\geq3$,
using the last term,~$db$, in~\eqref{eq:third-bound-PQ},
and the inequality $b^r \leq 3(b^r-2b^{r-1}+1)$
yields~\eqref{eq:bound-Q}.
For the case~$b=2$,
specializing the left inequality in~\eqref{eq:third-bound-PQ} at~$b=2$
yields precisely~\eqref{eq:bound-Q}.
In both cases, \eqref{eq:bound-P}~is a direct consequence
of \eqref{eq:second-bound-PQ} and~\eqref{eq:bound-Q}.
\end{proof}

\begin{rem}
The elementary result of Proposition~\ref{prop:deg-bounds-for-ratfuns} is remarkable
as it provides a uniform polynomial bound
in terms of scalar parameters describing the size of the equation,
namely its order~$r$ and degree~$d$.
This is unreachable in the shift operator case, where the degrees of rational
solutions of Riccati equations may be exponential in the bit size of the
equation.
For example, the Charlier polynomials
\[
  C_n(x,a) =
  {}_2F_0\left(
    \begin{array}{c}
      -n,-x\\ -
    \end{array}; -a^{-1}
  \right)
\]
viewed as functions of~$x$ with parameters $n\in\bN$ and $a > 0$ satisfy the recurrence equation \parencite[\S18.22]{DLMF}
\[
  a C_n(x+1,a)
  -(a + x) C_n(x,a)
  +x C_n(x-1,a)
  + n C_n(x,a)
  = 0
\]
whose degrees of coefficients are bounded by~$d = 1$
and whose last coefficient~$n$ has bit size~$\log n$,
while $C_n(x,a)$ has degree~$n$.
As a result the associated Riccati equation has solutions $C_n(x+1,a)/C_n(x,a)$ whose numerator and denominator cannot be bounded solely as a function of~$d$.
The fact that the Mahler operator drastically increases the degree is sometimes a difficulty,
but here it allows us to obtain the bounds in Lemma~\ref{prop:deg-bounds-for-ratfuns}.
\end{rem}

\begin{rem}
The exponents of $b$ and~$d$ in the bounds \eqref{eq:bound-P} and~\eqref{eq:bound-Q} are tight,
as we show now.
For a given nonzero rational function~$u = P/Q$ in its lowest terms with both $P$ and~$Q$ monic,
let us consider any operator of order~$r\geq 2$
\begin{multline*}
  L = \biggl(\sum_{k=0}^{r-1} c_k M^k\biggr) (QM - P) \\
  = c_{r-1}(M^{r-1}Q) M^r + \sum_{k=1}^{r-1} \left( c_{k-1} M^{k-1}Q - c_k M^k P  \right) M^k - c_0 P,
\end{multline*}
with constant coefficients~$c_0$, \dots, $c_{r-1}$ in~$\bKbar$
satisfying $c_{r-1} \neq 0$.
The operator admits $M - u$ as a right-hand factor, hence $u$ is a solution of the Riccati equation according to Lemma~\ref{lem:linear-riccati-equiv}.
The leading coefficient of~$L$ has degree $b^{r-1}\deg Q$.
In case $\deg Q \geq b \deg P$, all other coefficients have lower degree.
In case $\deg Q < b \deg P$, the coefficient of $M^{r-1}$ has degree $b^{r-1}\deg P$ and the remaining ones have lower degree.
Hence, in both cases
the maximum degree of the coefficients of~$L$ is equal to $d = b^{r-1} \max(\deg P,\deg Q)$.
Upon choosing $b \geq 3$ and~$\deg P = \deg Q$,
the bounds \eqref{eq:bound-P} and~\eqref{eq:bound-Q} become $\Bnum = 4\deg P$ and $\Bden = 3\deg Q$,
showing that they overshoot by no more than a constant factor~$4$.
\end{rem}

\part{Algorithm by exploration of the possible singularities}
\label{part:petkovsek}

\section{Mahlerian variant of Petkovšek's algorithm}
\label{sec:petkovsek-variant}

In this section, we present an algorithm for computing hypergeometric solutions
of linear Mahler equations adapted from Petkovšek's algorithm for difference
equations in the usual shift operator.
In doing so, we generalize to arbitrary order a previous adaptation of
Petkovšek's algorithm to Mahler equations of order~2 due to Roques.

Equations are given with polynomial coefficients over a field~$\bK$,
and we solve them for solutions
with coefficients in a field~$\bL$ satisfying $\bKbar \supseteq \bL \supseteq \bK$.
As a consequence of~\S\ref{sec:bounds}, particularly the definition~\eqref{eq:q-lambda} for the ramification bound~$q_\lambda$,
we have a bound
\begin{equation}\label{eq:q-L}
  q_\bL = \lcm_{\lambda\in\bL\cap\Lambda} q_\lambda
\end{equation}
on the ramification orders of solutions in~$\ramrat\bK$ of the Riccati equation,
so that computing all ramified rational solutions reduces to computing the plain
rational solutions of a modified equation.
Additionally, the method to be used in \S\ref{sec:roques} and~\S\ref{sec:bgpf}
requires a supplementary ramification in its intermediate calculations:
whatever the target ramification order,
our method requires
the working ramification order  to be multiplied by a factor~$b^{r-1}$.

\subsection{Petkovšek's classical algorithm}
\label{sec:petkovsek}

In the case of the linear classical difference equation
\begin{equation}\label{eq:shift-linear}
\ell_r(x) y(x+r) + \dots + \ell_0(x) y(x) = \sum_{i=0}^r \ell_i(x) y(x+i) = 0
\end{equation}
and the corresponding Riccati equation
\begin{equation}\label{eq:shift-riccati}
\ell_r(x) u(x) \dotsm u(x+r-1) + \dots + \ell_1(x) u(x) + \ell_0(x) = \sum_{i=0}^r \ell_i(x) \prod_{j=0}^{i-1} u(x+j) = 0 ,
\end{equation}
an algorithm due to \textcite{Petkovsek-1992-HSL} is known
to solve~\eqref{eq:shift-riccati} for all its rational function solutions,
or equivalently to find
all first-order right-hand factors of~\eqref{eq:shift-linear}.

The algorithm is based on the concept of a \emph{Gosper--Petkovšek form}:
for any rational function $u(x) \in \nonz{\rat\bL}$,
there exist
a constant~$\zeta \in \nonz{\bL} $
and monic polynomials $A(x),B(x),C(x)$ in~$\pol\bL$
satisfying:
\begin{enumerate}
\item $u(x) = \zeta \frac{C(x+1)}{C(x)} \frac{A(x)}{B(x)}$,
\item $A(x)$ and~$C(x)$ are coprime,
\item $B(x)$ and~$C(x+1)$ are coprime,
\item\label{it:petkovsek-AB}
  $A(x)$ and~$B(x+i)$ are coprime for all $i \geq 0$.
\end{enumerate}
Now, any potential nonzero rational solution~$u$
of the Riccati equation~\eqref{eq:shift-riccati}
leads to the necessary relation
\begin{equation}\label{eq:denom-free-shift-riccati}
\sum_{i=0}^r \ell_i(x) \, \zeta^i \, C(x+i) \, \biggl(\prod_{j=0}^{i-1} A(x+j)\biggr) \biggl(\prod_{j=i}^{r-1} B(x+j)\biggr) = 0 .
\end{equation}
Here, $A(x)$~appears in all the terms of the sum but the one for~$i=0$.
As it is coprime to all forward shifts of~$B(x)$ and to~$C(x)$,
it must divide~$\ell_0(x)$.
Similarly, $B(x+r-1)$~appears in every term but the one for~$i=r$,
and so must divide~$\ell_r(x)$.
This motivates iterating on all pairs of monic divisors of $\ell_0$ and~$\ell_r$.
Given monic $A(x) \divides \ell_0(x)$ and $B(x) \divides \ell_r(x-r+1)$,
where divisibility is meant in~$\pol\bL$,
the leading coefficient of the left-hand side of~\eqref{eq:denom-free-shift-riccati}
is independent of the choice of a monic~$C(x)$.
So this leading coefficient yields
an algebraic equation in~$\zeta$.
Each choice of a solution~$\zeta \in \bL$
turns~\eqref{eq:denom-free-shift-riccati} into an equation
that can be solved for polynomial solutions $C(x)\in\pol\bL$.
The algorithm then gathers and returns all found $(\zeta, A(x), B(x), C(x))$.

Remark that the classical definition of the literature,
quantified over all integers~$i \in \bN$
and reproduced as point~\ref{it:petkovsek-AB} above,
is stronger than needed:
the proof has used the property for $0\leq i \leq r-1$ only.
In the Mahler case,
it will prove important to use the weaker constraint
to get an algorithm.
The reader should also compare point~\ref{it:petkovsek-AB} of the present section
with its analogues in \S\ref{sec:roques} and~\S\ref{sec:bgpf}.

\subsection{Roques's algorithm for order~2}
\label{sec:roques}

Inspired by Hendriks' works
for usual difference equations \parencite{Hendriks-1998-ADD}     
and for $q$-difference equations \parencite{Hendriks-1997-ACS},  
\textcite[\S6.2]{Roques-2018-ARB} recently presented
an analogue of Petkovšek's algorithm for Mahler equations of order~$2$,
that is, equations of the form
\begin{equation}\label{eq:mahler-riccati-2}
\ell_2(x) u(x) u(x^b) + \ell_1(x) u(x) + \ell_0(x) = 0 .
\end{equation}
Roques's original presentation
determines along his algorithm
a suitable extension of~$\bK$        
sufficient to obtain all solutions from~$\ramrat\bKbar$.
Here, we present a variant that computes with an input field~$\bL$,
so as to be consistent with the rest of our text.

After finding a bound~$q$ such that all solutions~$u \in \ramrat\bL$
are in fact in~$\rrat\bL{q}$,
let us introduce a new indeterminate~$t$ for which $x = t^{qb}$,
so as to prove, for any given~$u \in \ramrat\bL$,
the existence
of a nonzero constant~$\zeta \in \nonz{\bL}$
and monic polynomials $A,B,C \in \bL[t]$
satisfying
\begin{enumerate}
\item\label{it:u-of-ABC-in-BGPF}
  $u(x) = \zeta \frac{C(x^{1/q})}{C(x^{1/qb})} \frac{A(x^{1/qb})}{B(x^{1/qb})}$,
  that is, $u(t^{qb}) = \zeta \frac{C(t^b)}{C(t)} \frac{A(t)}{B(t)}$,
\item $A$ and~$C$ are coprime,
\item $B$ and~$M C$ are coprime,
\item\label{it:roques-AB-bounded}
  $A$ and~$M^i B$ are coprime for all $i \in \{0,1\}$,
\item\label{it:roques-C-MC} $C$ and~$M C$ are coprime,
\end{enumerate}
where $M$~acts on~$\bL[t]$ by substituting $t^b$ for~$t$.

Now, any nonzero ramified rational solution~$u$
of the Riccati equation~\eqref{eq:mahler-riccati-2}
leads to the necessary relation
\begin{multline}\label{eq:roques-necessary}
\ell_2(t^{qb}) \zeta^2 C(t^{b^2}) A(t) A(t^b)
\\ \hbox{} + \ell_1(t^{qb}) \zeta C(t^b) A(t) B(t^b)
+ \ell_0(t^{qb}) C(t) B(t) B(t^b) = 0 .
\end{multline}
This time, one finds that
$A(t)$~must divide~$\ell_0(t^{qb})$,
while $B(t^b)$~must divide~$\ell_2(t^{qb})$,
implying that $B(t)$~must divide~$\ell_2(t^q)$.

Roques's approach continues in a way similar in spirit to Petkovšek's,
although technical reasons require to exchange the order of steps.
Roques indeed finds a linear constraint on the degree of~$C$ by putting apart the term with~$\zeta^2$ in~\eqref{eq:roques-necessary}.
Finding the coefficients of~$C$ then amounts to solving a finite linear system.
As a side remark, in Roques's text valuations and degrees are treated interchangeably, and the calculation of the coefficient~$\zeta$ is obscured.
We will clarify this in the next section.

\subsection{A new algorithm for higher-order Mahler equations}
\label{sec:bgpf}

In this section,
we deal with Riccati Mahler equations~\eqref{eq:riccati}
of general order~${r \geq 2}$.
The core of the section describes an algorithm for solving
for (plain) rational functions from~$\rat\bL$.

In order to go beyond order~2,
we observe that
point~\ref{it:roques-AB-bounded} in Roques's definition
requires a coprimality only for~$0 \leq i \leq  r-1 = 1$ (when~$r = 2$),
whereas point~\ref{it:petkovsek-AB} in the shift situation
requires a coprimality for~$0 \leq i$,
although the proof in~\S\ref{sec:petkovsek} for an equation of order~$r$
only uses the cases~$0 \leq i \leq r-1$.
This motivates us
to introduce the following notion of a bounded Gosper--Petkovšek form.
\begin{defi}\label{def:BGP}
Given a rational function $u(x) \in \nonz{\rat\bL}$ and an integer~$r \geq 2$,
a \emph{bounded Gosper--Petkovšek form}
of order~$r$ for~$u(x)$ is a tuple
$(\zeta, A, B, C) \in \nonz{\bL} \times \bL[t]^3$,
with $A$, $B$, $C$ monic polynomials in a new indeterminate~$t$,
\hbox{such that}:
\begin{enumerate}
\item\label{it:gen-def-u}
  $u(t^{b^{r-1}}) = \zeta \frac{C(t^b)}{C(t)} \frac{A(t^{b^{r-1}})}{B(t)}$,
\item\label{it:coprimality-Mr-1A-C}
  $M^{r-1} A$ and~$C$ are coprime,
\item\label{it:coprimality-B-MC}
  $B$ and~$M C$ are coprime,
\item\label{it:bounded-AB}
  $M^i A$ and~$B$ are coprime for all $i \in \{0,\dots,r-1\}$,
\end{enumerate}
where $M$~acts on~$\bL[t]$ by substituting $t^b$ for~$t$.
\end{defi}
By Lemma~\ref{lem:M-Bezout} over~$K=\bL$,
point~\ref{it:bounded-AB} above can be restated equivalently into:
$M^{r-1} A$ and~$M^i B$ are coprime for all $i \in \{0,\dots,r-1\}$.
Also remark that
for a polynomial triple $(A, B, C)$ satisfying our definition for~$r = 2$,
the polynomial triple $(MA, B, C)$ satisfies Roques's definition (with~$q = 1$),
at least provided $\gcd(C, MC) = 1$.
A~constructive proof of the existence
of bounded Gosper--Petkovšek forms
will be provided in~\S\ref{sec:bounded-gpf-computation}.

Now, consider any potential nonzero rational solution~$u$
of the Riccati equation~\eqref{eq:riccati},
represented by one of its bounded Gosper--Petkovšek forms.
Substituting the Gosper--Petkovšek form for~$u$
and $t^{b^{r-1}}$ for~$x$ in~\eqref{eq:riccati} leads
after canceling denominators
to the necessary relation
\begin{equation}\label{eq:mahler-necessary-1}
  \tilde L(t, \zeta M)\, C = 0 ,
\end{equation}
where $\tilde L$~is the operator in~$\bL[t]\langle M \rangle$ defined by
\begin{equation}\label{eq:mahler-necessary-2}
  \tilde L (t, M) =
\sum_{k=0}^r \biggl( M^{r-1}\ell_k(t) \times \prod_{j=0}^{k-1} M^{r-1+j}A(t) \times \prod_{j=k}^{r-1} M^jB(t) \biggr) M^k .
\end{equation}
In this formula, note that beside the polynomials $A$ and~$B$ that are by definition in the ring~$\bL[t]$,
we have noted $\ell_k(t)$ for the result of the substitution
of~$t$ for~$x$ in~$\ell_k$.
The factor $M^{r-1} A$~appears in all the terms
of the expansion of~\eqref{eq:mahler-necessary-1}, except for the one corresponding to~$k=0$,
\[ M^{r-1}\ell_0(t) \times C(t) \times \prod_{j=0}^{r-1} M^jB(t) ; \]
using points \ref{it:coprimality-Mr-1A-C} and~\ref{it:bounded-AB} in Definition~\ref{def:BGP},
$M^{r-1} A$~must divide~$M^{r-1} \ell_0$,
and Lemma~\ref{lem:M-Bezout} implies that $A$~must divide~$\ell_0$.
Similarly, $M^{r-1} B$~appears in all the terms of~\eqref{eq:mahler-necessary-1} , except for the one corresponding to~$k=r$,
\[ M^{r-1}\ell_r(t) \times \zeta^r M^rC(t) \times \prod_{j=0}^{r-1} M^{r-1+j}A(t) ; \]
thus, $M^{r-1} B$~must divide~$M^{r-1} \ell_r$,
and Lemma~\ref{lem:M-Bezout} implies that $B$~must divide~$\ell_r$.

To complete the pairs~$(A,B)$ into candidate tuples $(\zeta, A, B, C)$
delivering rational functions in~$\bL(t)$,
we interpret~\eqref{eq:mahler-necessary-1} as an analogue of~\eqref{eq:single-eqn} in Proposition~\ref{lem:single-lambda-equiv}.
To this end,
we start by introducing
in analogy with the construction of~$\omdr$ in~\S\ref{sec:difference}
a ring whose coefficients are series in~$t^{-1}$, in the form
\begin{equation*}
  \omdr' = \bigoplus_{\lambda \in \nonz{\bKbar}}\, \mylogalt{\lambda}\puiseuxalt\bKbar .
\end{equation*}
Here, the~$\mylogalt{\lambda}$ denote elements of some new family~$(e'_\lambda)_{\lambda\in\nonz\bKbar})$ of generators
satisfying the algebraic relations
$e'_\lambda e'_{\lambda'} = e'_{\lambda\lambda'}$ and $e'_1 = 1$.
The theory of~$\omdr'$ parallels that of~$\omdr$,
including the results developed
in~\S\ref{sec:bounds}
to find the possible
logarithmic parts~$\mylog\lambda$ and valuations of the corresponding Puiseux series in a
solution of~\eqref{eq:linear}
by
using the lower Newton polygon.

\begin{defi}\label{def:upper-newton}
Let $L$ be a general linear Mahler operator.
The upper convex hull of the set of points $(b^k,j) \in \bR^2$ of Definition~\ref{def:lower-newton} is called the \emph{upper Newton polygon} (of~$L$).
In analogy with Definition~\ref{def:lower-char-poly},
we index its edges by decreasing slopes and
denote by~$\xi_j(X)$ the \emph{characteristic polynomial} of its $j$th edge.
Define~$Z(L)$ to be the union of the sets of roots in~$\bKbar$ of the~$\xi_j$,
\begin{equation}\label{eq:def-Zeta}
Z(L) = \{ \zeta \in \bKbar : \exists j, \ \xi_j(\zeta) = 0 \} .
\end{equation}
The $j$th edge will be called \emph{$\zeta$-admissible} (for~$L$) if~$\xi_j(\zeta) = 0$.
\end{defi}

Now, an analogue of Proposition~\ref{lem:single-lambda-equiv} holds,
with the difference that the parameter~$\lambda$ of the new proposition
needs to be the root of some characteristic polynomial
associated with an edge of an upper Newton polygon,
rather than a lower Newton polygon.
Consequently,
the constant~$\zeta$ in~\eqref{eq:mahler-necessary-1}
must be an element of~$Z(\tilde L)$.
In this situation, the degree of~$C$ in~$t$
must be the opposite of the slope of some edge
associated with the upper Newton polygon:
the degree of a solution~$C$ of~\eqref{eq:mahler-necessary-1} in~$t$ is the opposite of the valuation with respect to~$t^{-1}$ of~$C$,
which is given by the lower Newton polygon of~$\tilde L(t^{-1},M)$;
the upper Newton polygon of~$\tilde L(t,M)$ is the symmetric with respect to a horizontal axis of the lower Newton polygon of~$\tilde L(t^{-1},M)$;
edges associated by this symmetry have opposite slopes;
the result follows.

We summarize the previous discussion in the following result.

\begin{prop}\label{prop:bGP-necessary-conds}
Let $(\zeta, A, B, C)$ be a bounded Gosper--Petkovšek form of order~$r$
of any rational solution~$u$ of~\eqref{eq:riccati},
and let $\tilde L$ be defined by~\eqref{eq:mahler-necessary-2}.
Then, necessarily:
\begin{itemize}
\item $\zeta \in Z(\tilde L)$,
\item $A$~divides~$\ell_0(t)$,
\item $B$~divides~$\ell_r(t)$,
\item $C$~is a nonzero polynomial solution of~\eqref{eq:mahler-necessary-1},
  implying $\deg C$~is the opposite of a nonpositive integer slope of an edge
  associated with the upper Newton polygon of~$\tilde L$.
\end{itemize}
\end{prop}

\begin{algo}
\inputs{
  A Riccati Mahler equation~\eqref{eq:riccati} with coefficients $\ell_k(x) \in \pol\bK$.
  Some intermediate field~$\bL$, that is, a field satisfying $\bKbar \supseteq \bL \supseteq \bK$.
}
\outputs{
  The set of rational functions $u \in \rat\bL$ that solve~\eqref{eq:riccati}.
}
\caption{\label{algo:BP}%
Rational solutions to a Riccati Mahler equation.
Compare with the efficiency improvements in Algorithm~\ref{algo:IP}.}
\begin{deepenum}[label=(\Alph*)]
\item Set $\cU := \varnothing$.
\item \label{algo:BP:main-loop}
      For each monic factor $A(t) \in \bL[t]$ of~$\ell_0(t)$, for each monic factor $B(t) \in \bL[t]$ of~$\ell_r(t)$ such that $M^iA$ and~$B$ are coprime for $0 \leq i < r$:
      \begin{deepenum}[label=(\arabic*)]
      \item \label{it:compute-L-tilde}
            compute $\tilde L(t,M)$ by~\eqref{eq:mahler-necessary-2},
      \item \label{it:upper-newton-polygon}
            compute the upper Newton polygon of~$\tilde L$, the set~$Z(\tilde L) \cap \bL$ of roots~$\zeta$ in~$\bL$ of the associated characteristic polynomials,
      \item \label{it:loop-over-zeta}
            for each~$\zeta$ in~$Z(\tilde L) \cap \bL$:
            \begin{deepenum}[label=(\alph*)]
            \item \label{it:candidate-degree-C}
                  compute the maximum $\Delta_\zeta$ of the integer values of the opposites of the slopes of the $\zeta$-admissible edges (for~$\tilde L$),
            \item \label{it:unchanged-steps}
                  if $\Delta_\zeta \geq 0$:
                  \begin{deepenum}[label=(\roman*)]
                  \item \label{it:solving-for-C}
                        compute a basis $(C_i)_{1\leq i \leq s}$ of solutions in $\bL[t]_{\leq \Delta_\zeta}$ of the equation
                        $\tilde L(t, \zeta M)\, C = 0$,
                  \item \label{it:parametrizing-for-C}   
                        if $s>0$:
                        \begin{deepenum}[label=(\greek*)]
                        \item \label{it:param-C}
                              set~$C := \sum_{i=1}^s c_i C_i$ for formal parameters $c_i$,
                        \item \label{it:identify-rf}
                              normalize the rational function \\
                              \begin{myequation*}
                              \tilde u(t) := \zeta \frac{C(t^b)}{C(t)} \frac{A(t^{b^{r-1}})}{B(t)} ,
                              \end{myequation*}\\[1ex]
                              which is an element of $\bL(c_1,\dots,c_s)(t^{b^{r-1}})$,
                              so as to identify $u(x) \in \bL(c_1,\dots,c_s)(x)$ such that $u(t^{b^{r-1}}) = \tilde u(t)$,
                        \item augment~$\cU$ with~$u$.
                        \end{deepenum} 
                  \end{deepenum} 
            \end{deepenum} 
      \end{deepenum} 
\item \label{it:algo-return}
      Remove redundant elements from~$\cU$ by the method described before Proposition~\ref{prop:BP-no-loss} and return the resulting~$\cU$.
\end{deepenum} 
\end{algo}

The previous considerations lead to Algorithm~\ref{algo:BP},
whose general structure is the following:
\begin{itemize}
\item a loop over candidates~$(A,B)$ is set up at step~\ref{algo:BP:main-loop} from the \emph{lower} Newton polygon of the input~$L$;
\item candidates~$\zeta$, then degrees for candidates~$C$ at step~\ref{algo:BP:main-loop}\ref{it:loop-over-zeta}\ref{it:candidate-degree-C}, are obtained from the \emph{upper} Newton polygon of the auxiliary operator~$\tilde L$;
\item solving for~$C$ at step~\ref{algo:BP:main-loop}\ref{it:loop-over-zeta}\ref{it:unchanged-steps}\ref{it:solving-for-C} is done by appealing
to our algorithm for polynomial solutions of bounded degree
in \parencite[\S2.6, Algorithm~5]{ChyzakDreyfusDumasMezzarobba-2018-CSL};
\item to avoid redundancy in the output,
  the cleaning step~\ref{it:algo-return} makes sure to return a partition
  by enforcing that no parametrization is included in another.
\end{itemize}

Regarding the last item,
the cleaning step~\ref{it:algo-return} removes redundant elements of~$\cU$ by a simple loop:
while $\cU$~contains two distinct elements $u^{(1)}$ and~$u^{(2)}$
with inclusion of~$u^{(1)}$ in~$u^{(2)}$,
it removes $u^{(1)}$ from~$\cU$.
This is effective provided we have an algorithmic inclusion test available.
To describe one, observe that,
given two parametrized rational functions,
$u^{(1)}$~with parameters $c^{(1)} = (c^{(1)}_1,\dots,c^{(1)}_{s_1})$
and $u^{(2)}$~with parameters $c^{(2)} = (c^{(2)}_1,\dots,c^{(2)}_{s_2})$,
inclusion of~$u^{(1)}$ in~$u^{(2)}$ is only possible if~$s_1 \leq s_2$,
in which case getting rid of denominators in the equation~$u^{(1)} = u^{(2)}$
and equating like powers of~$x$
results in a linear system in~$c^{(2)}$ linearly parametrized by~$c^{(1)}$.
Either this system has no nonzero solution, proving noninclusion,
or it provides a parametrization proving inclusion.
We thus have an algorithm by linear algebra over~$\bL$.

We now prove a first part of the correctness of the algorithm:
no rational solution is lost.

\begin{prop}\label{prop:BP-no-loss}
Algorithm~\ref{algo:BP} computes the set of rational solutions of the Riccati equation~\eqref{eq:riccati}
as a union of sets parametrized by finite-dimensional\/ $\bL$-projective spaces.
\end{prop}

\begin{proof}

Observe that Algorithm~\ref{algo:BP}
iterates on all tuples $(\zeta, A, B, \Gamma(c))$
in $\nonz{\bL} \times \bL[t]^2 \times \bL(c_1,c_2,\dots)[t]$ such that,
for any values of the parameters~$c_i$ in~$\bL$,
the tuple $(\zeta, A, B, C)$ where $C = \Gamma(c)$ satisfies:
the necessary conditions (i), (ii), and~(iii) of Proposition~\ref{prop:bGP-necessary-conds},
point~\ref{it:bounded-AB} of Definition~\ref{def:BGP},
a degree bound on~$C$ that is necessarily satisfied by the polynomial kernel of~$\tilde L(t,\zeta M)$, cf. the necessary condition~(iv) in the same proposition.
The algorithm therefore represents a set that is less constrained that
the set of bounded Gosper--Petkovšek forms of order~$r$ of~$u$
for any rational solution~$u$ of~\eqref{eq:riccati};
in particular, it represents them all, and,
by the existence of bounded Gosper--Petkovšek forms (see~\S\ref{sec:bounded-gpf-computation}),
the algorithm must find any solution~$u$ of~\eqref{eq:riccati}.
Conversely, for any~$u(x)$ element of the output~$\cU$,
obtained by the algorithm from a tuple $(\zeta, A, B, C)$,
expanding~\eqref{eq:mahler-necessary-1} by using~\eqref{eq:mahler-necessary-2},
then dividing  by the relevant product
shows that $u(x)$~solves~\eqref{eq:riccati}.
We have thus proved that Algorithm~\ref{algo:BP} computes all rational solutions.
\end{proof}

For the end of the present section, we call a \emph{block}
the image of a similarity class of nonzero hypergeometric elements
under $y \mapsto M y/y$.
The following theorem states the correctness of Algorithm~\ref{algo:BP}.

\begin{thm}\label{thm:BP}
Algorithm~\ref{algo:BP} computes the set of rational solutions of the Riccati equation~\eqref{eq:riccati}
as a disjoint union of sets bijectively parameterized by finite-dimensional\/ $\bL$-projective spaces.
Equivalently, the output~$\cU$ of Algorithm~\ref{algo:BP} consists
of formal parametrizations of blocks, with exactly one parametrization for each block,
and each such parametrization is a bijection.
\end{thm}
\begin{proof}
Each element~$U(c) = (c_1,\dots,c_s)$ of~$\cU$ can be interpreted as a parametrization
\begin{equation}\label{eq:BP}
\gamma = (\gamma_1:\ldots:\gamma_s) \in \bP^{s-1}(\bL)
\mapsto
U(\gamma) = \zeta \frac{MC}{C} \frac{M^{r-1}A}{B} \in \rat\bL .
\end{equation}
Fix $y_0$, a nonzero solution of~$My_0 = \zeta \frac{M^{r-1}A}{B} y_0$.
The solutions of $My = U(\gamma)y$ when $\gamma$~ranges in~$\nonz{(\bL^s)}$
form a vector space~$\fG$ that is exactly $C(\bL^s)y_0$.
Because $C$~is a polynomial in~$x$,
the set~$\nonz{\fG}$ is included in a similarity class.

For each similarity class~$\nonz{\fH}$,
the finite-dimensional space~$\fH$
is by Proposition~\ref{prop:BP-no-loss}
covered as a finite union of vector spaces~$\fG$ obtained from elements~$U\in\cU$.
Let $d$~denote the dimension of~$\fH$.
Assume that no~$\fG$ has dimension~$d$.
Then, each~$\fG$ is in some hyperplane defined by some nonzero linear form~$\phi$.
Fix a coordinate system of~$\fH$ and
let $S$ be the parametrized curve $t \mapsto (1,t,\dots,t^{d-1})$,
whose image is included in~$\fH$.
The curve can only meet a given~$\fG$ at finitely many intersections,
provided by the zeros of the polynomial $\phi(S(t))$.
Because the curve has infinitely many points,
this contradicts that $\fH$~is covered.
So at least one of the~$\fG$ is equal to~$\fH$.

By the absence of redundancy enforced at step~\ref{it:algo-return} in the algorithm,
no two distinct elements of~$\cU$ can produce vector spaces $\fG_1$ and~$\fG_2$
with $\fG_1 \subseteq \fG_2$.
So any vector space~$\fH$ of a similarity class~$\nonz{\fH}$
is obtained exactly once as a~$\fG$,
and only such~$\fH$ are obtained.
The result for blocks then follows.

Because the algorithm constructs~$C$ in~\eqref{eq:BP}
as a formal sum $C = \sum_{j=1}^s c_j C_j$
for $\bL$-linearly independent polynomials~$C_j \in \pol\bL$,
the~$C_jy_0$ form an $\bL$-basis of~$C(\bL^s)y_0$.
The implied bijection from~$\bL^s$ to the suitable~$\fH$
induces a bijection from~$\bP^{s-1}(\bL)$ to the block associated with~$\nonz{\fH}$.
\end{proof}

\begin{cor}\label{cor:riccati-to-hypergeom}
The set of\/ $\rat\bL$-hypergeometric solutions of the Mahler equation~\eqref{eq:linear}
can be described as a union of\/ $\bL$-vector spaces~$\fH$ in direct sum
pa\-ram\-e\-triz\-ed by the output~$\cU$ as follows:
for given~$U \in \cU$, the vector space~$\fH$ is spanned by a basis $(y_1,\dots,y_s)$
where each~$y_i$ is any nonzero solution of the equation $My_i = U(e_i) y_i$
for the $i$th element of the canonical basis.
\end{cor}

\begin{proof}
The statement is a direct consequence of the last paragraph
in the proof of Theorem~\ref{thm:BP}.
\end{proof}

\begin{rem}[efficiency of Algorithm~\ref{algo:BP}]
The double loop over $A$ and~$B$
induces an exponential behavior,
making this basic algorithm inefficient.
We will discuss several pruning strategies in~\S\ref{sec:efficiency-improvements}.
Furthermore, the algorithm considers tuples~$(\zeta,A,B,C)$ that provide solutions
but are not bounded Gosper--Petkovšek forms.
These tuples are redundant because the algorithm also produces the same
solutions in bounded Gosper--Petkovšek form.
From the point of view of efficiency,
the problem is not in rejecting tuples that are not bounded Gosper--Petkovšek forms,
but in not computing too many (redundant) candidates in the first place.
The necessity to determine the polynomials~$C$ late in the algorithm
makes it impossible to consider only bounded Gosper--Petkovšek forms,
and is the main cause for getting repeated solutions in the course of the algorithm
(see also the Remark~\ref{rem:GPforms-not-unique} below).
\end{rem}

\begin{rem}[solving large linear systems]
When searching for polynomial solutions of the auxiliary equations
$\tilde L(t,\zeta M)\, C = 0$,
it is crucial for performance to use the fast algorithms
introduced in~\parencite{ChyzakDreyfusDumasMezzarobba-2018-CSL}.
We take the opportunity to clarify a minor point of confusion in that work.
In~\parencite[Algorithm~2]{ChyzakDreyfusDumasMezzarobba-2018-CSL},
we appeal to the linear solving algorithm of~\textcite{IbarraMoranHui-1982-GFL}
for computing the kernel of an $m \times n$ matrix~$A$ in $O(m^{\omega-1}n)$ field
operations.
This algorithm is based on computing an LSP decomposition of~$A$, that is,
a decomposition as a product of a lower square matrix~$L$, a semi-upper
triangular matrix~$S$, and a permutation matrix~$P$.
Strictly speaking, however, this approach is only valid when~$m \leq n$
(fewer rows than columns).
As our need is for an $m\times n$ matrix~$A$ with~$m \geq n$ (more rows than columns),
a workaround is to apply the LSP decomposition algorithm to the transpose~$A^T$,
thus obtaining a PSU decomposition of~$A$,
that is, as a permutation matrix~$P$, a semi-lower triangular matrix~$S$, and an upper square matrix~$U$.
To solve for the (right) kernel, we observe $\ker A = U^{-1} (\ker S)$.
As $S$~reduces to a lower triangular matrix with nonzero diagonal elements
when the zero columns are deleted,
finding its kernel is immediate.
Inverting~$U$ takes~$O(n^\omega)$ operations.
So~$\ker A$ is obtained in $O(m/n)$ square matrix products,
hence in complexity $O(mn^{\omega-1})$.
Taking the rank~$r$ of~$A$ into account and refining the analysis
yields the complexity $O(mn r^{\omega - 2})$ announced
in~\parencite[Prop.~2.14]{ChyzakDreyfusDumasMezzarobba-2018-CSL}.
\end{rem}

To compute all \emph{ramified} rational solutions of~\eqref{eq:riccati},
we now propose Algorithm~\ref{algo:ramP},
which is a simple variant of Algorithm~\ref{algo:BP}.
By the property of~$q_\bL$ to be a uniform bound
on the ramification order of ramified rational solutions,
Algorithm~\ref{algo:ramP} is also correct, in the sense that it satisfies
an analogue of Theorem~\ref{thm:BP},
where rational solutions are replaced by ramified rational solutions.
An obvious adaptation of Corollary~\ref{cor:riccati-to-hypergeom}
to $\ramrat\bL$-hypergeometric solutions also holds.

\begin{algo}
\inputs{
  A Riccati Mahler equation~\eqref{eq:riccati} with coefficients $\ell_k(x) \in \pol\bK$.
  Some intermediate field~$\bL$, that is, a field satisfying $\bKbar \supseteq \bL \supseteq \bK$.
}
\outputs{
  The set of ramified rational functions $u \in \ramrat\bL$ that solve~\eqref{eq:riccati}.
}
\caption{\label{algo:ramP}%
Ramified rational solutions to a Riccati Mahler equation.}
\begin{deepenum}[label=(\Alph*)]
\item Compute the lower Newton polygon of $L := \sum_{k=0}^r \ell_k(x) M^k$, the associated characteristic polynomials, the set $\Lambda\cap\bL$, and the ramification bound~$q_\bL$ by~\eqref{eq:q-L}.
\item Call Algorithm~\ref{algo:BP} for computing the rational solutions of~$L(x^{q_\bL}, M)$ in~$\rat\bL$.
\item Substitute $x^{1/q_\bL}$ for~$x$ in the obtained solutions and return the resulting set.
\end{deepenum}
\end{algo}

\subsection{Existence and computation of bounded Gosper--Petkovšek forms}
\label{sec:bounded-gpf-computation}

The following lemma proves the existence of bounded Gosper--Petkovšek forms at all orders,
and its proof provides an algorithm for putting a rational function
in bounded Gosper--Petkovšek form of a given order.

\begin{lem}\label{lem:exists-bgpf}
Given any two coprime monic polynomials $P$ and~$Q$ in~$\pol\bK$,
define
a sequence of triples of polynomials given for all~$k \in \nonz{\bN}$ by
\begin{align}
\label{eq:def-ABC-1}
(A_1, B_1, C_1) &= (P, Q, 1) \\
\label{eq:def-ABC-k}
(A_{k+1}, B_{k+1}, C_{k+1}) &= \left(\frac{A_k}{G_k}, \frac{M B_k}{G_k}, M C_k \times (M^0 G_k \dotsm M^{k-1} G_k)\right) ,
\end{align}
where $G_k = \gcd(A_k, MB_k)$.
Then, for all~$k \in \nonz{\bN}$,
$(1, A_k, B_k, C_k)$ is a bounded Gosper--Petkovšek form of order~$k$
for the rational function~$P/Q$.
Additionally, there exists $k \in \nonz{\bN}$ satisfying
\begin{equation}\label{eq:stationary-ABC}
A_{k+i} = A_k , \quad B_{k+i} = M^i B_k , \quad C_{k+i} = M^i C_k ,
\qquad\text{for all~$i \in \bN$.}
\end{equation}
\end{lem}
\begin{proof}
The proof is by induction on~$k$.
The case~$k=1$ is immediate by~\eqref{eq:def-ABC-1}.
For some~$k \in \nonz{\bN}$, assume
that $(1, A_k, B_k, C_k)$ is a bounded Gosper--Petkovšek form of order~$k$
for the rational function~$P/Q$
and that \eqref{eq:def-ABC-k}~holds.
This provides the formulas:
\begin{align}
\label{eq:bgpf-ABC}
M^{k-1} \frac{P}{Q} &= \frac{M C_k}{C_k} \frac{M^{k-1} A_k}{B_k} , \\
\label{eq:bgpf-coprime-MkA-C}
\gcd(M^{k-1} A_k, C_k) &= 1 , \\
\label{eq:bgpf-coprime-B-MC}
\gcd(B_k, M C_k) &= 1 , \\
\label{eq:bgpf-coprime-MiA-B}
\gcd(M^i A_k, B_k) &= 1 \quad\text{for all}\quad 0 \leq i < k .
\end{align}
Then, first observe
\begin{equation*}
\frac{M C_{k+1}}{C_{k+1}} = \frac{M^2 C_k}{M C_k} \frac{M^k G_k}{G_k}
\end{equation*}
and apply~$M$ to~\eqref{eq:bgpf-ABC}, so that
\begin{equation*}
M^k \frac{P}{Q} = \frac{M^2 C_k}{M C_k} \frac{M^k (A_{k+1} G_k)}{B_{k+1} G_k}
  = \frac{M C_{k+1}}{C_{k+1}} \frac{M^k A_{k+1}}{B_{k+1}} .
\end{equation*}
Applying Lemma~\ref{lem:M-Bezout} with $K=\bK$ and~$j=1$ to~(\ref{eq:bgpf-coprime-MkA-C}--\ref{eq:bgpf-coprime-MiA-B}) implies
\begin{gather*}
\gcd(M^{i+1} (A_{k+1} G_k), B_{k+1} G_k) = 1 \quad\text{for all}\quad 0 \leq i < k , \\
\gcd(M^k (A_{k+1} G_k), M C_k) = \gcd(B_{k+1} G_k, M^2 C_k) = 1 ,
\end{gather*}
so that in particular
\begin{gather}
\label{eq:several-gcd}
\gcd(M^{i+1} A_{k+1}, B_{k+1}) = \gcd(M^{i+1} A_{k+1}, G_k) = \gcd(M^{i+1} G_k, B_{k+1}) = 1 , \\
\label{eq:several-more-gcd}
\gcd(M^k A_{k+1}, M C_k) = \gcd(B_{k+1}, M^2 C_k) = 1 .
\end{gather}
The construction~\eqref{eq:def-ABC-k} ensures $\gcd(M^0 A_{k+1}, B_{k+1}) = 1$,
so combining with the left gcd in~\eqref{eq:several-gcd}
shows~\eqref{eq:bgpf-coprime-MiA-B} at~$k+1$.
Applying Lemma~\ref{lem:M-Bezout} with~$j=k-i-1$ to the middle gcd in~\eqref{eq:several-gcd} next yields
\begin{equation*}
\gcd(M^k A_{k+1}, M^{k-i-1} G_k) = 1 .
\end{equation*}
Considering those gcds for~$0 \leq i < k$ as well as the left gcd in~\eqref{eq:several-more-gcd}
proves~\eqref{eq:bgpf-coprime-MkA-C} at~$k+1$.
Combining the right gcd in~\eqref{eq:several-gcd} for $0 \leq i < k$
together with the right gcd in~\eqref{eq:several-more-gcd}
proves~\eqref{eq:bgpf-coprime-B-MC} at~$k+1$.
We have obtained
that $(1, A_{k+1}, B_{k+1}, C_{k+1})$ is a bounded Gosper--Petkovšek form of order~$k+1$
for the rational function~$P/Q$.

Last, in view of~\eqref{eq:def-ABC-k},
the existence of~$k$ satisfying~\eqref{eq:stationary-ABC}
is equivalent
to the existence of~$k$ such that $G_{k+i} = 1$ for all~$i \in \bN$.
This~$k$ exists because $\deg A_k$~cannot decrease indefinitely.
\end{proof}

\begin{rem}\label{rem:GPforms-not-unique}
Bounded Gosper--Petkovšek forms are not unique,
even for a given order.
An example with~$b = 2$ is given for order~$r = 3$
by $(\zeta, A, B, C) = (1, x+1, 1, 1)$
and $(\zeta', A', B', C') = (1, 1, 1, x^4-1)$,
where the first form is computed
by the direct use of the recurrence in Lemma~\ref{lem:exists-bgpf}.
\end{rem}

\begin{rem}
Point~\ref{it:roques-C-MC} in our presentation of Roques's definition
(see~\S\ref{sec:roques})
is not enforced in our definition of bounded Gosper--Petkovšek forms.
(Also compare with the constraints on $r$, $s$, and~$t$ in \parencite[\S6.2, top of page~345]{Roques-2018-ARB}.)
The mere construction of~$C_{k+1}$ in~\eqref{eq:def-ABC-k} makes
the coprimality between $C$ and~$MC$ generally impossible,
at least for~$r\geq3$.
\end{rem}

\section{A more practical algorithm}
\label{sec:efficiency-improvements}

As experiments show
(see~\S\ref{sec:discussion} and especially the timings in Table~\ref{tab:Basic-vs-Improved-GP}),
a direct implementation of Algorithm~\ref{algo:BP}
is inefficient because of the large number of pairs of divisors~$(A,B)$
it has to consider.
In the present section, we describe several ways to mitigate this issue.
The results are rather technical.
They are not used in the rest of the article,
apart from the fact that the resulting improved variant algorithm
is tested beside our other algorithms in~\S\ref{sec:impl-and-benchmark}.

\subsection{Ensuring coprimality}
\label{sec:cartesian-prod}

Given a polynomial~$f \in \bL[t]$,
let $\irred(f)$ denote the set of its \mif factors
and $\mult f p$ denote the multiplicity of a \mif polynomial~$p$ in~$f$.
By representing a monic divisor of~$\ell_0$ by the family $(\alpha_p)_{p\in\irred (\ell_0)}$ of exponents
in its  factorization $\prod_{p\in\irred (\ell_0)} p^{\alpha_p}$ into \mif factors,
the set of such divisors can be viewed as the Cartesian product
\[
  \cA = \prod_{p\in\irred (\ell_0)} \intinv{0}{\mult{\ell_0}p} ,
\]
where $\intinv{a}{b}$~denotes the integer interval~$\{a,\dots,b\}$.
In the same way, the set of monic divisors of~$\ell_r$ is represented by
\[
  \cB = \prod_{q\in\irred (\ell_r)} \intinv{0}{\mult{\ell_r}q} .
\]
So, the main loop~\ref{algo:BP:main-loop} in Algorithm~\ref{algo:BP}
can be viewed as parametrized by a pair~$(\alpha,\beta)$ in the product~$\cA\times\cB$.
However, we are  only interested in pairs~$(A,B)$
that satisfy $\gcd(M^s A, B) = 1$ for all $0 \leq s < r$,
and the basic algorithm has to test all those gcds
in the double loop over $A$ and~$B$.

To explain how to prune $\cA\times\cB$,
let us imagine that $\ell_r$ has a \mif factor~$q$ with multiplicity~$b > 0$
and that $\ell_0$ has a \mif factor~$p$ with multiplicity~$a > 0$
such that $q$~is a factor of some~$M^s p$, $0 \leq s < r$.
In such a situation, we say that the factor~$p$ is \emph{forbidden} (in~$A$) by the factor~$q$ (of~$B$).
When a pair~$(\alpha,\beta)$ of tuples is in~$\cA\times\cB$,
the pair~$(\alpha_p,\beta_q)$ of integers lies in the Cartesian product $\intinv{0}{a} \times \intinv{0}{b}$.
Nonetheless,
only the parts $\intinv{0}{a} \times \{0\}$ and $\{0\} \times \intinv{1}{b}$
have to be considered,
as $(\alpha_p,\beta_q)$ outside of these parts violate the coprimality condition.
The previous considerations are independent of the other coordinates $\alpha_{p'}$ and~$\beta_{q'}$,
so only a fraction $(a + b + 1)/((a + 1)(b + 1))$ of the whole product~$\cA\times\cB$ is useful.
Taking several pairs~$(p,q)$ into account, we expect only a small fraction of~$\cA\times\cB$ to remain.

We now make this observation algorithmic.
Given a \mif factor~$q$ of~$\ell_r$,
we denote by~$\forbiddenby(q)$ the set of \mif factors~$p$ of~$\ell_0$ forbidden by~$q$.
Let $\cR$ be the subset of~$\irred(\ell_r)$ consisting
of those factors~$q$ of~$\ell_r$ such that $\forbiddenby(q)$~is nonempty.
Given a~$B$, a \mif factor~$q$ of~$B$
restricts the choice of the useful~$A$ if it is in~$\cR$,
and places no restriction on~$A$ if it is not in~$\cR$.
Next, for a subset~$\pi \subseteq \cR$,
let $\forbiddenby(\pi)$ denote $\bigcup_{q\in\pi}\forbiddenby(q)$.
The set of useful pairs~$(A,B)$ can now be seen to be
the disjoint union over the subsets $\pi \subseteq \cR$ of the Cartesian products
\begin{equation*}
\cC_\pi = \prod_{p\in \irred(\ell_0)} \cA(p,\pi)
   \times \prod_{q\in \irred(\ell_r)} \cB(\pi,q) ,
\end{equation*}
where
\begin{equation*}
\cA(p,\pi) = \begin{cases}
\intinv{0}{\mult{\ell_0}{p}} & \text{if $p\not\in\forbiddenby(\pi)$}, \\
\{0\} & \text{if $p\in\forbiddenby(\pi)$},
\end{cases}
\qquad\!
\cB(\pi,q) = \begin{cases}
\intinv{0}{\mult{\ell_r}{q}} & \text{if $q\not\in\cR$}, \\
\intinv{1}{\mult{\ell_r}{q}} & \text{if $q\in\pi$}, \\
\{0\} & \text{if $q\in\cR\setminus\pi$}.
\end{cases}
\end{equation*}
Rather than precomputing the $\cA(p,\pi)$ and~$\cB(\pi,q)$ explicitly,
in our algorithm
we discard useless pairs on the fly.

In the following, we call \emph{factored representation} of a polynomial
its representation as a power product of \mif factors.
An optimization of Algorithm~\ref{algo:BP}
is obtained by just changing the iteration of loop~\ref{algo:BP:main-loop}
into
\begin{itemize}
\item[\ref{algo:IP-B-loop}]
      For all $B := \prod_{q \in \irred(\ell_r)} q^{\beta_q}$ such that $0 \leq \beta_q \leq \mult{\ell_r}{q}$ for all~$q \in \irred(\ell_r)$:
      \begin{enumerate}[label=(\alph*), start=21]
      \item \label{algo:prose-IP-forbidden-by-actual-factors}
            let $F$ be the union of the~$\forbiddenby(q)$ over all~$q \in \irred(\ell_r)$ such that~$\beta_q>0$,
      \item \label{algo:prose-IP-prohibit-forbidden-factors}
            for $p \in \irred(\ell_0)$, set~$a_p$ to~$0$ if~$p\in F$, to~$\mult{\ell_0}{p}$ otherwise,
      \item for all $A := \prod_{p \in \irred(\ell_0)} p^{\alpha_p}$ such that $0 \leq \alpha_q \leq a_p$ for all~$p \in \irred(\ell_0)$,
            \begin{enumerate}[label=(\arabic*)]
            \item compute $\tilde L(t,M)$ by~\eqref{eq:mahler-necessary-2}, \\ \dots
            \end{enumerate}
      \end{enumerate}
\end{itemize}
\noindent
Here,
by the definition of~$\cB(\pi,q)$,
those~$q$ for which~$\beta_q>0$ are either in~$\pi$ or in~$\irred(\ell_r) \setminus \cR$.
Therefore, the set~$F$ computed at step~\ref{algo:prose-IP-forbidden-by-actual-factors}
is~$\forbiddenby(\pi)$,
because $\forbiddenby(q) = \varnothing$ if~$q \not\in \cR$.
The integer~$a_p$ defined at step~\ref{algo:prose-IP-prohibit-forbidden-factors}
reflects the definition of~$\cA(p,\pi)$.

Additionally, the set~$\forbiddenby(q)$ for a given~$q$ can be computed efficiently
by using the Gräffe operator~$G$ defined by
$Gf = \operatorname{Res}_y (y^b - x, f(y))$
\parencite[\emph{cf.}][\S3.2, especially Lemma~3.1]{ChyzakDreyfusDumasMezzarobba-2018-CSL},
as expressed in the following lemma.
The point of this approach is to avoid carrying out divisions
of high-degree polynomials~$M^s p$ by the polynomials~$q$.

\begin{lem}
  For any \mif factor~$q$ of~$\ell_r$, one has
  \[ \forbiddenby(q) = \irred(\ell_0) \cap \{ q, \sqrt Gq, \dots, \sqrt G^{r-1}q \} , \]
  where $\sqrt G f$ denotes the squarefree part of~$Gf$,
  computed by forcing exponents to~$1$ in the factored representation of~$Gf$.
\end{lem}

\begin{proof}
From the definition of~$G$ follow the multiplicativity formula $G(fg) = Gf\,Gg$
and the degree-preservation formula $\deg Gp = \deg p$.
In particular, $G$~preserves divisibility.
One has the crucial relations
$u\divides M^iG^iu$ and ${G^iM^iu = u^{b^i}}$ for any \mif polynomial~$u$ and~$i\in\bN$.
In terms of~$\sqrt G$, the formulas become
$u\divides M^i\sqrt G^iu$ and~$\sqrt G^iM^iu = u$ for~$i\in\bN$.

Fix an integer~$s$ such that~$0\leq s < r$.
If $q\divides M^s p$, then upon applying~$\sqrt G^s$ we find $\sqrt G^sq\divides \sqrt G^sM^s p = p$, so~$\sqrt G^sq = p$.
Conversely, if $\sqrt G^sq = p$, then upon applying~$M^s$ we find $q\divides M^s\sqrt G^sq = M^s p$, so~$q\divides M^s p$.
Now, a \mif factor~$p$ of~$\ell_0$ is forbidden by a \mif factor~$q$ of~$\ell_r$
if and only if there exists an integer~$s$ such that $0\leq s < r$ and $q\divides M^s p$,
that is, if and only if there exists an integer~$s$ such that $0\leq s < r$ and $p = \sqrt G^s q$.
\end{proof}

\subsection{Avoiding redundant pairs}
\label{sec:avoid-redundant-pairs}

Because some of the coprimality conditions defining bounded Gosper--Pet\-kov\-šek forms
cannot be taken into account in Algorithm~\ref{algo:BP}
until parametrizations of solutions~$C$ have been obtained,
solutions are naturally found with repetitions.
Here, we alleviate this
by predicting some of the repetitions from divisibility conditions on $A$ and~$B$.
By the design of Algorithm~\ref{algo:BP},
a rational solution~$u(x)$ will be generated multiple times
(before the final step that removes redundancy)
if and only if two triples $(A,B,C)$ and~$(A',B',C')$ of monic polynomials satisfying
\begin{equation}\label{eq:ABC-is-A'B'C'}
\frac{MC}{C} \frac{M^{r-1}A}{B} = \frac{MC'}{C'} \frac{M^{r-1}A'}{B'}
\end{equation}
are considered for the same~$\zeta$ during the run of the algorithm.

Recall that in Algorithm~\ref{algo:BP},
ramification is taken care of by introducing a new indeterminate~$t$
that plays the role of some appropriate root of~$x$.
The results described in~\S\ref{sec:cartesian-prod} for polynomials in~$\pol\bL$
apply with trivial modifications for polynomials of~$\bL[t]$,
and for the end of~\S\ref{sec:efficiency-improvements},
we continue with the indeterminate~$t$.
In particular, we let $\ell_0$ and~$\ell_r$ denote $\ell_0(t)$ and~$\ell_r(t)$, respectively.

We will study several scenarios
in which, if the pair $(A,B)$ is considered in the run of the algorithm
and, for a certain~$\zeta$,
leads to a polynomial~$C$ and a rational solution~$u$,
then another triple $(A',B',C')$ leads to the same~$u$ for the same~$\zeta$.
We will thus obtain rules that attempt to:
minimize the multiplicity of~$t$ in~$A$,
maximize the multiplicity of~$t$ in~$B$,
remove factors of the form~$Mp/p$ from~$A$,
introduce factors of the form~$Mp/p$ in~$B$,
replace a divisor~$Mp$ with the factor~$p$ in~$A$,
replace the factor~$p$ with a divisor~$Mp$ in~$B$.
In all cases, the polynomial~$C$ will be adjusted
to compensate for the change.
This compensation will operate by multiplications/divisions by elements of~$\bL[t]$
obtained independently of~$\zeta$.
This makes the actual value of~$\zeta$ irrelevant
when considering equalities of the form~\eqref{eq:ABC-is-A'B'C'},
which justifies our pruning independently of~$\zeta$.

\paragraph{Pruning rules}

The following four situations provide instances of~\eqref{eq:ABC-is-A'B'C'}
that will permit us to discard pairs~$(A,B)$.
Note that in each case, monic $A$, $B$, $C$, and~$p$ induce
monic $A'$, $B'$, and~$C'$:
\begin{itemize}
\item If $B = p\tilde B$, define $A'= A$, $B' = Mp \, \tilde B$, $C' = pC$ to get:
\begin{equation*}
    \frac{MC}{C} \frac{M^{r-1}A}{B} =
    \frac{Mp \, MC}{Mp \, C}\frac{M^{r-1}A'}{p\tilde B} =
    \frac{MC'}{C'} \frac{M^{r-1}A'}{B'} .
\end{equation*}
\item If $A = Mp \, \tilde A$, define $A'= p \tilde A$, $B' = B$, $C'= M^{r-1}p \, C$ to get:
\begin{multline*}
    \frac{MC}{C} \frac{M^{r-1}A}{B} =
    \frac{MC}{C} \frac{M^{r-1}(Mp \, \tilde A)}{B} = \\
    \frac{M(M^{r-1}p \, C)}{M^{r-1}p \, C} \frac{M^{r-1}(p\tilde A)}{B} =
    \frac{MC'}{C'} \frac{M^{r-1}A'}{B'} .
\end{multline*}
\item If $p \divides Mp$, define $A'= A$, $B' = (Mp/p) B$, $C' = pC$ to get:
\begin{equation*}
    \frac{MC}{C} \frac{M^{r-1}A}{B} =
    \frac{Mp \, MC}{pC} \frac{M^{r-1}A}{(Mp/p)B}
     = \frac{MC'}{C'} \frac{M^{r-1}A'}{B'} .
\end{equation*}
\item If $p \divides Mp$ and $A = (Mp/p) \, \tilde A$, define $A'= \tilde A$, $B' = B$, $C'= M^{r-1}p \, C$ to get:
\begin{multline*}
    \frac{MC}{C} \frac{M^{r-1}A}{B} =
    \frac{MC}{C} \frac{M^{r-1}((Mp/p)\tilde A)}{B} = \\
    \frac{M(M^{r-1}p \, C)}{M^{r-1}p \, C} \frac{M^{r-1}A'}{B} =
    \frac{MC'}{C'} \frac{M^{r-1}A'}{B'} .
\end{multline*}
\end{itemize}
For any monic divisor~$A$ of~$\ell_0$, any monic divisor~$B$ of~$\ell_r$,
these guarded formulas provide a rule to discard the pair~$(A,B)$
if there exists a \mif polynomial~$p$ satisfying any of the predicates:
\begin{enumerate}[label=(P\arabic*), ref=(P\arabic*)]
  \item\label{pred:alpha-case-B}
    $B$~is of the form~$p\tilde B$ and $Mp \, \tilde B$ divides~$\ell_r$.
  \item\label{pred:alpha-case-A}
    $A$~is of the form~$Mp \, \tilde A$ and $p\tilde A$ divides~$\ell_0$.
  \item\label{pred:beta-case-B}
    $p$~divides~$Mp$ and $(Mp/p) B$ divides~$\ell_r$.
  \item\label{pred:beta-case-A}
    $p$~divides~$Mp$ and $A$~is of the form~$(Mp/p) \tilde A$.
\end{enumerate}
As an example,
we legitimate the use of~\ref{pred:beta-case-A} after assuming the last situation listed above.
Because $A' \divides A \divides \ell_0$ and $B' = B \divides \ell_r$,
if $(A,B,C)$ is considered in the algorithm for a certain~$\zeta$ and leads to a solution~$u$,
then the pair~$(A', B')$ is also considered,
and the solution~$u$ found with~$(A,B)$ after the polynomial solving step gets~$C$
will also be found with~$(A',B')$ and the same~$\zeta$ after the polynomial solving step gets~$C'$.
Using the pair~$(A',B')$ instead of~$(A,B)$ is better,
because it potentially leads to more solutions,
including polynomials~$C'$ not divisible by~$M^{r-1} p$.
The use of the other predicates is justified by a similar reasoning.

It is clear that the monic nature of~$p$ is no restriction:
it is needed to relate a monic~$B$ and a monic~$\tilde B$ in~\ref{pred:alpha-case-B},
and a monic~$A$ and a monic~$\tilde A$ in~\ref{pred:alpha-case-A};
only~$Mp/p$ plays a role in \ref{pred:beta-case-B} and~\ref{pred:beta-case-A}
and it depends on~$\bL p$ only.

Furthermore, the reasoning on the predicates does not rely on the irreducibility of~$p$:
one easily proves that there exists a (monic) irreducible~$p$ making a rule apply
if and only if there exists a (monic) not necessarily irreducible~$p$
making the same (generalized) rule apply.
So it is algorithmically permissible to restrict to (monic) irreducible~$p$,
and doing so avoids having to consider exponentially more factors~$p$.

Finally, making~$p = t$ in \ref{pred:beta-case-B} and~\ref{pred:beta-case-A}
results in slightly more explicit formulations, respectively:
\begin{enumerate}[label=(P\arabic*), ref=(P\arabic*), resume]
  \item\label{pred:maximize-t-in-B}
    $\mult{B}{t} \leq \mult{\ell_r}{t} - b + 1$.
  \item\label{pred:minimize-t-in-A}
    $\mult{A}{t} \geq b - 1$.
\end{enumerate}
Note that making~$p = t$ in \ref{pred:alpha-case-B} results in special cases of~\ref{pred:maximize-t-in-B},
and doing so in \ref{pred:alpha-case-A} results in special cases of~\ref{pred:minimize-t-in-A}.

\paragraph{Iteration of the rules}

An optimization for factors is obtained
from two distinct \mif polynomials $p$ and~$q$ satisfying~$Mp = pq$:
considering the final result of a repeated use of predicates
makes it possible to fix certain multiplicities
before entering the loops over $A$ and~$B$.
If for~$k\geq1$, $q^kB$~divides~$\ell_r$,
an iterated use of~\ref{pred:beta-case-B} shows that all of $B$, \dots, $q^{k-1}B$ can be skipped,
and we get the predicate:
\begin{enumerate}[label=(P\arabic*), ref=(P\arabic*), resume]
  \item\label{pred:iterated-beta-case-B}
    $Mp = pq$ for~$p\neq q$ and $\mult{B}{q} < \mult{\ell_r}{q}$.
\end{enumerate}
Similarly, if for~$k\geq1$, $A$~is of the~form~$q^k\tilde A$ and divides~$\ell_0$,
an iterated use of~\ref{pred:beta-case-A} shows that all of $q^k\tilde A$, \dots, $q\tilde A$ can be skipped,
and we get the predicate:
\begin{enumerate}[label=(P\arabic*), ref=(P\arabic*), resume]
  \item\label{pred:iterated-beta-case-A}
    $Mp = pq$ for~$p\neq q$ and $\mult{A}{q} > 0$.
\end{enumerate}
Note that $p\neq q$ requires $p\neq t$ and~$q\neq t$.

\begin{footnotesize}
\begin{algo}
\inputs{
  A Riccati Mahler equation~\eqref{eq:riccati} with coefficients $\ell_k(x) \in \pol\bK$.
  Some intermediate field~$\bL$, that is, a field satisfying $\bKbar \supseteq \bL \supseteq \bK$.
}
\outputs{
  The set of rational functions $u \in \rat\bL$ that solve~\eqref{eq:riccati}.
}
\caption{\label{algo:IP}%
Improved variant of Algorithm~\ref{algo:BP}.}
\begin{deepenum}[label=(\Alph*), start=23]
\item \label{algo:IP-upper-Newton-polygon}
      Compute the upper Newton polygon of~$L := \sum_{k=0}^r \ell_k(x) M^k$, the associated characteristic polynomials, the set~$Z\cap\bL$ where $Z = Z(L)$, and for each~$\zeta \in Z\cap\bL$, the set of indices~$j$ making the $j$th edge $\zeta$-admissible.
\item From now on, let $\ell_k$ denote~$\ell_k(t)$.
\item \label{algo:IP-memoization}
      Compute factored representations of relevant polynomials and related data:
      \begin{deepenum}[label=(\alph*), start = 10]
      \item for $q \in \irred(\ell_r)$, set $\forbiddenby(q) := \irred(\ell_0) \cap \{ q, \sqrt Gq, \dots, \sqrt G^{r-1}q \}$,
      \item \label{algo:Mp-p-in-ell_r}
            set $\cD_r$ to the set of those $f \in \{ Mp/p : p \in \irred(G\ell_r), \ p\neq t \} \cap \bL[t]$ that divide~$\ell_r$,
            \hfill \emph{\tiny [\ref{pred:beta-case-B}, not \ref{pred:maximize-t-in-B} or \ref{pred:iterated-beta-case-B}]}
      \item \label{algo:Mp-p-in-ell_0}
            set $\cD_0$ to the set of those $f \in \{ Mp/p : p \in \irred(G\ell_0), \ p\neq t\} \cap \bL[t]$ that divide~$\ell_0$.
            \hfill \emph{\tiny [\ref{pred:beta-case-A}, not \ref{pred:minimize-t-in-A} or \ref{pred:iterated-beta-case-A}]}
      \end{deepenum} 
\item \label{algo:IP-refine-loop-bounds}
      Refine bounds for the loops below:
      \begin{deepenum}[label=(\alph*), start = 10]
      \item if $t\divides\ell_r$, set $\check b_t := \max(0, \mult{\ell_r}{t} - b + 2)$,
            \hfill \emph{\tiny [\ref{pred:maximize-t-in-B}]}
      \item if $t\divides\ell_0$, set $\check a_t := b-2$,
            \hfill \emph{\tiny [\ref{pred:minimize-t-in-A}]}
      \item \label{algo:IP-iterated-beta-case-B-for-Mahler-bifactors}
            for $q \in \irred(\ell_r)$ with~$q \neq t$, set~$\check b_q$ to~$\mult{\ell_r}{q}$ if $Mp = p q$ for $p = \sqrt Gq$, to~$0$ otherwise,
            \hfill \emph{\tiny [\ref{pred:iterated-beta-case-B}]}
      \item \label{algo:IP-iterated-beta-case-A-for-Mahler-bifactors}
            for $q \in \irred(\ell_0)$ with~$q \neq t$, set~$\check a_q$ to~$0$ if $Mp = p q$ for $p = \sqrt Gq$, to~$\mult{\ell_0}{q}$ otherwise.
            \hfill \emph{\tiny [\ref{pred:iterated-beta-case-A}]}
      \end{deepenum} 
\setcounter{deepenumi}{0}
\item Set $\cU := \varnothing$.
\item \label{algo:IP-B-loop}
      For all $B := \prod_{q \in \irred(\ell_r)} q^{\beta_q}$ such that $\check b_q \leq \beta_q \leq \mult{\ell_r}{q}$ for all~$q \in \irred(\ell_r)$:
      \begin{deepenum}[label=(\alph*), start=20]
      \item \label{algo:IP-B-skip}
            continue to the next~$B$ if either of the following conditions holds:
            \begin{deepenum}[label=(\greek*)]
            \item \label{algo:IP-B-skip-beta-case-B}
                  $Mp \, (B/p) \divides \ell_r$ for some $p \in \irred(B)$ with $p \neq t$,
                  \hfill \emph{\tiny [\ref{pred:alpha-case-B}, not \ref{pred:maximize-t-in-B}]}
            \item \label{algo:IP-B-skip-alpha-case-B}
                  $fB \divides \ell_r$ for some $f \in \cD_r$,
                  \hfill \emph{\tiny [\ref{pred:beta-case-B}, not \ref{pred:maximize-t-in-B} or \ref{pred:iterated-beta-case-B}]}
            \end{deepenum}
      \item \label{algo:IP-forbidden-by-actual-factors}
            let $F$ be the union of the~$\forbiddenby(q)$ over all~$q \in \irred(\ell_r)$ such that~$\beta_q>0$,
      \item \label{algo:IP-prohibit-forbidden-factors}
            for $p \in \irred(\ell_0)$, set~$a_p$ to~$0$ if~$p\in F$, to~$\mult{\ell_0}{p}$ otherwise,
      \item \label{algo:IP-A-loop}
            for all $A := \prod_{p \in \irred(\ell_0)} p^{\alpha_p}$ such that $0 \leq \alpha_p \leq \min(\check a_p, a_p)$ for all~$p \in \irred(\ell_0)$,
            \begin{deepenum}[label=(\arabic*), start=0]
            \item \label{algo:IP-A-skip}
                  continue to the next~$A$ if either of the following conditions holds:
                  \begin{enumerate}[label=(\greek*)]
                  \item \label{algo:IP-A-skip-alpha-case-A}
                        for~$p \in \irred(\ell_0)$ with~$p \neq t$,
                        $Mp$~divides~$A$ and $p \, (A/Mp)$~divides~$\ell_0$,
                        \hfill \emph{\tiny [\ref{pred:alpha-case-A}, not \ref{pred:minimize-t-in-A}]}
                  \item \label{algo:IP-A-skip-beta-case-A}
                        $fA \divides \ell_0$ for some $f \in \cD_0$,
                        \hfill \emph{\tiny [\ref{pred:beta-case-A}, not \ref{pred:minimize-t-in-A} or \ref{pred:iterated-beta-case-A}]}
                  \end{enumerate}
            \item \label{algo:IP-Ltilde}
                  compute $\tilde L(t,M)$ by~\eqref{eq:mahler-necessary-2},
            \item \label{algo:IP-zeta-loop}
                  for each~$\zeta$ in~$Z\cap\bL$:
                  \begin{deepenum}[label=(\alph*)]
                  \item \label{algo:IP-degree-bound-for-poly-sol}
                        compute the maximum $\Delta_\zeta$ of the integer values taken by~$\operatorname{sh}(\delta)$ in~\eqref{eq:from-delta-to-Delta} when $\delta$~ranges over the opposites of the slopes of the $\zeta$-admissible edges of~$L$,
                  \item \label{algo:IP-solve-for-polynomials}
                        proceed as in step~\ref{algo:BP:main-loop}\ref{it:loop-over-zeta}\ref{it:unchanged-steps} of Algorithm~\ref{algo:BP}.
                  \end{deepenum} 
            \end{deepenum} 
      \end{deepenum} 
\item \label{algo:IP-return}
      Finish as in step~\ref{it:algo-return} of Algorithm~\ref{algo:BP}.
\end{deepenum} 
\end{algo}
\end{footnotesize}

\paragraph{Implementation of the discarding rule}

The optimization of loop~\ref{algo:BP:main-loop}
that we introduced in~\S\ref{sec:cartesian-prod}
is further refined in Algorithm~\ref{algo:IP}:
before~(u) in the loop over~$B$, we insert a step~(t) to take \ref{pred:alpha-case-B} and~\ref{pred:beta-case-B} into account,
and before~(1) in the loop over~$A$, we insert a step~(0) to take \ref{pred:alpha-case-A} and~\ref{pred:beta-case-A} into account.
Notwithstanding,
the specializations \ref{pred:maximize-t-in-B} to~\ref{pred:iterated-beta-case-A}
are taken care of at step~\ref{algo:IP-refine-loop-bounds},
by restricting the exponents used in $A$ and~$B$
to ranges in which the rules
\ref{pred:alpha-case-B} to~\ref{pred:beta-case-A}
cannot~apply.

Because \ref{pred:beta-case-B}~tests those polynomials~$Mp/p$ that divide~$\ell_r$,
a precomputation at step~\ref{algo:Mp-p-in-ell_r} determines those quotients.
For any \mif factor~$q$ of any such~$Mp/p$,
$\sqrt G q$~is equal to~$p$,
thus restricting the search for~$p$ to~$\irred(G\ell_r)$.
A~similar discussion applies to \ref{pred:beta-case-A} and step~\ref{algo:Mp-p-in-ell_0}.

\subsection{Avoiding redundant computations of Newton polygons}

Given an operator~$P$,
recall that the~$\xi_j$ denote the characteristic polynomials of its upper Newton polygon
(Definition~\ref{def:upper-newton}),
and that $Z(P)$~is the set of all roots~$\zeta$ of these characteristic polynomials, according to~\eqref{eq:def-Zeta}.
Step~\ref{algo:BP:main-loop}\ref{it:compute-L-tilde} of Algorithm~\ref{algo:BP}
computes the characteristic polynomials~$\xi_j(X)$ associated with~$\tilde L$
for every pair~$(A,B)$ of divisors of~$\ell_0$ and~$\ell_r$.
It turns out that the roots of the~$\xi_j(X)$ do not depend on~$(A,B)$.
In this section, the \emph{degree} of a finite Puiseux series is defined
as the maximal (rational) exponent appearing with a nonzero coefficient.

\begin{lem}\label{lem:shearing}
The set~$Z(\tilde L)$ is equal to~$Z(L)$
for every pair~$(A,B)$ of polynomials in~$\bL[t]$.
Moreover, the set~$\Delta(\tilde L)$
of degrees of finite Puiseux series solutions of $\tilde L y = 0$
is obtained from the set~$\Delta(L)$
of degrees of finite Puiseux series solution of $L y = 0$
by an affine transformation:
\begin{equation}\label{eq:from-delta-to-Delta}
  \Delta(\tilde L) = \operatorname{sh}(\Delta(L))
  \quad\text{where}\quad
  \operatorname{sh}(\delta) = - \frac{b^{r-1} \deg A - \deg B}{b - 1} + b^{r - 1}\delta .
\end{equation}
\end{lem}

\begin{proof}
The upper Newton polygon of~$L$, respectively~$\tilde L$, is the upper part of the convex hull
of the points $(b^k, d_k)$, $0 \leq k \leq r$,
respectively of the points $(b^k, D_k)$, $0 \leq k \leq r$,
with the relation
\begin{equation*}
  D_k = b^{r-1}  d_k
  + b^{r-1} \frac{b^k - 1}{b - 1} \deg A
  + \frac{b^r - b^k}{b - 1} \deg B
  = b^{r-1}  d_k + \alpha b^k + \beta
\end{equation*}
that results directly from~\eqref{eq:mahler-necessary-2} after setting
\begin{equation*}
\alpha = \frac{b^{r-1} \deg A - \deg B}{b - 1}
\quad\text{and}\quad
\beta = \frac{b^r \deg B - b^{r-1} \deg A}{b - 1} .
\end{equation*}
As the coefficient $b^{r-1}$ is positive, the upper convex hull of the points $(b^k,d_k)$, $0 \leq k \leq r$,
is mapped onto the upper convex hull of the points $(b^k,D_k)$, $0 \leq k \leq r$.
As a consequence,
the degrees of the finite Puiseux series solutions of $L$ and~$\tilde L$
are related by~\eqref{eq:from-delta-to-Delta},
since they are the opposite of the slopes of the edges~%
\parencite[Lemma~2.5]{ChyzakDreyfusDumasMezzarobba-2018-CSL}.
Moreover, as the polynomials $A$ and~$B$ are monic, the coefficients of the monomials
$x^{d_k} M^k$ in~$L$ and~$x^{D_k}M^k$ in~$\tilde L$ are the same for $0 \leq k \leq r$.
So the characteristic polynomials are the same for the edges of both upper Newton polygons.
\end{proof}

A first consequence of Lemma~\ref{lem:shearing} is that
we can change the structure of Algorithm~\ref{algo:BP}:
instead of recomputing~$Z(\tilde L)$ at step~\ref{it:upper-newton-polygon}
for each new~$(A,B)$,
a single computation of~$Z(L)$ is done
before entering loop~\ref{algo:BP:main-loop},
and step~\ref{it:loop-over-zeta} becomes a loop over~$\zeta \in Z(L)$.
In the latter loop,
we still need to compute a bound on the degree of a solution~$C$ for each pair~$(A,B)$,
but it is deduced from the lemma
by considering the largest nonnegative integer in~$\Delta(\tilde L)$
as obtained from~$\Delta(L)$
by the parametrization~\eqref{eq:from-delta-to-Delta}.

\part{Algorithm by reconstruction from series}
\label{part:hermite-pade}

\section{Solutions as syzygies}
\label{sec:syzygies}

We develop another approach to finding
all ramified rational function solutions
of the Riccati Mahler equation~\eqref{eq:riccati}.
In comparison to the adaption of Petkovšek's algorithm in~\S\ref{sec:petkovsek-variant},
which focuses on fixing the possible singularities of a solution of~\eqref{eq:riccati},
this second approach will proceed in a guess-and-check manner,
refining calculations on approximate solutions of the linear equation~\eqref{eq:linear}
until a whole set of candidate solutions is proven to be exact solutions.
This approach will be embodied in our Algorithm~\ref{algo:HP}.

Our algorithm will search for the set~$\ricsol_{\ramrat\bKbar}$ of solutions
with coefficients in~$\bKbar$,
but for the first part of the theory,
we more generally consider an intermediate field $\bKbar \supseteq \bL \supseteq \bK$,
as in~\S\ref{sec:petkovsek-variant}.
We compute Hermite--Padé approximants for auxiliary problems,
before recombining them to obtain structured approximants.
These auxiliary problems show no logarithm and ramification,
making it possible to focus on formal power series calculations.
Each auxiliary problem solves a modified Riccati equation
for its  rational solutions (in~$\rat\bL$)
whose series expansions have nonnegative valuation and their leading coefficient is equal to~$1$.

Let us focus on the space of formal power series solutions of the linear equation~\eqref{eq:linear}
and consider an $\bL$-basis $(z_1,\dots,z_\zdim)$,
from which we deduce solutions~$u$ of~\eqref{eq:riccati} in~$\fps\bL$
by the parametrization
\begin{equation}\label{eq:u-of-a}
  u = \frac{M(a_1 z_1 + \dots + a_\zdim z_\zdim)}
              {a_1 z_1 + \dots + a_\zdim z_\zdim},\qquad
  a = (a_1:\ldots:a_\zdim) \in \projzn{\bL^\zdim}.
\end{equation}
We need to determine those~$a$ for which the series~$u$ lies in~$\rat\bL$.
Suppose that $u$~is indeed some $P/Q$ in~$\rat\bL$, with coprime $P$ and~$Q$.
After canceling denominators and using linearity,
we obtain the equation
\begin{equation}\label{eq:HP-structured-rel}
  (- a_1 P) \, z_1 + \dots + (- a_\zdim P) \, z_\zdim +
(a_1 Q) \, M z_1 + \dots + (a_\zdim Q) \, M z_\zdim = 0 .
\end{equation}
This linear relation with polynomial coefficients from~$\pol\bL$
between the series $z_1$,~\dots, $z_\zdim$, $M z_1$,~\dots, $M z_\zdim$ from $\fps\bL$
is a special case of a linear relation
\begin{equation}\label{eq:HP-exact-rel}
  P_1 z_1 + \dots + P_\zdim z_\zdim + Q_1 \, M z_1 + \dots + Q_\zdim \, M z_\zdim  = 0 ,
\end{equation}
for polynomial coefficients $P_i$ and~$Q_i$ from~$\pol\bL$.
A second level of relaxation is obtained, for any~$\sigma \in \bN$,
by Hermite--Padé approximants $(P_1,\dots,P_\zdim,Q_1,\dots,Q_\zdim)$
in~$\pol\bL^{2\zdim}$, to order~$\sigma$,
that is, approximate linear relations of the form
\begin{equation}\label{eq:HP-approx-rel}
  P_1 z_1 + \dots + P_\zdim z_\zdim + Q_1 \, M z_1 + \dots + Q_\zdim \, M z_\zdim  = O(x^\sigma) .
\end{equation}

Recombining exact relations~\eqref{eq:HP-exact-rel}
into \emph{structured} relations of the form~\eqref{eq:HP-structured-rel}
reduces to solving a polynomial system in $a_1, \dots, a_\zdim$
(see Definition~\ref{def:Sigma-pol-sys}), so that we can in principle go back from
\eqref{eq:HP-exact-rel} to~\eqref{eq:HP-structured-rel}.
Unfortunately, given bounds on the degrees of the $P_i$ and~$Q_i$,
it is not clear how to find an accuracy~$\sigma$ such that
\eqref{eq:HP-approx-rel} implies~\eqref{eq:HP-exact-rel}.
Instead, Algorithm~\ref{algo:HP} computes the approximate
relations~\eqref{eq:HP-approx-rel} compatible with our degree bounds on $P$ and~$Q$
and attempts to reconstruct
hypergeometric solutions starting from these relations.
This proceeds by guessing candidate relations for increasing~$\sigma$
and rejecting wrong ones until we get no false solutions.
We will show (Theorem~\ref{thm:HP-correct})
that this process eventually yields all hypergeometric solutions,
and nothing but hypergeometric solutions.

\begin{table}

\begin{tabular}{|l|l|p{80mm}|}
\hline
$\sigma$ & &
  the truncation order of solution series~$z_i$ to~\eqref{eq:linear} that serve to obtain approximate syzygies to the same order;
  increased when the algorithm detects it cannot conclude \\
\hline
$\cS\trsigma$ & Def.~\ref{def:syzygies} &
  $\pol\bL$-module of approximate syzygies (“approximate syzygy module”);
  nonincreasing and ultimately decreasing w.r.t.~$\sigma$; constant rank \\
\hline
$\cT\trsigma$ & Def.~\ref{def:truncated-module} &
  vectorial truncation of the approximate syzygy module to total degree at most~$B_\infty$;
  nonincreasing and ultimately constant w.r.t.~$\sigma$;
  a finite-dimensional $\bL$-vector space \\
\hline
$\rho\trsigma$ & Def.~\ref{def:rank-rho} &
  rank of the module $\pol\bL\cT\trsigma$;
  nonincreasing and ultimately constant w.r.t.~$\sigma$ \\
\hline
$\cA\trsigma$ & \parbox[t]{16mm}{\raggedright Def.~\ref{def:cone-A-sigma}, \\ Prop.~\ref{prop:cone-A-sigma-alt}} &
  $\bL$-cone of the~$a$ for which structured approximate syzygies exist (“characteristic cone of structured approximate syzygies”, “characteristic cone” for short);
  nonincreasing and ultimately constant w.r.t.~$\sigma$;
  converges to a direct sum of\/ $\bL$-vector spaces \\
\hline
$\cV\trsigma$ & Def.~\ref{def:cone-V-sigma} &
  $\bL$-cone that is a relaxation of~$\cA\trsigma$, thus generally providing only candidates~$a$ (“relaxed cone of structured approximate syzygies”, “relaxed cone” for short);
  converges to the same limit as~$\cA\trsigma$ \\
\hline
$W\trsigma$ & Def.~\ref{def:truncated-module} &
  matrix computed as a minimal basis from the~$z_i$ (“generating matrix”);
  its rows generate $\pol\bL\cT\trsigma$\\
\hline
$W_a\trsigma$ & Def.~\ref{def:W-a} &
  augmentation of the matrix~$W\trsigma$ (“augmented generating matrix”);
  singular if and only if~$a \in \cV\trsigma$ \\
\hline
$\Sigma\trsigma$ & Def.~\ref{def:Sigma-pol-sys} &
  polynomial system obtained from $W_a\trsigma$;
  generates an ideal~$\ideal{\Sigma\trsigma}$ whose variety is~$\cV\trsigma$ \\
\hline
$\sqrt{\ideal{\Sigma\trsigma}}$ & Thm.~\ref{thm:radical-as-linear-ideals} &
  radical of the ideal~$\ideal{\Sigma\trsigma}$;
  used to test linearity of the irreducible components of~$\cV\trsigma$ \\
\hline
\end{tabular}

\caption{\label{tab:cast-of-characters}Inventory of the objects to be introduced in \S\ref{sec:syzygies} and~\S\ref{sec:hermite-pade}.
Each of the quantities~$X\trsigma$ above has a limit when $\sigma$~goes to~$\infty$, denoted~$X\trinfty$.}
\end{table}

The remainder of~\S\ref{sec:syzygies} is dedicated to the details of this strategy,
leading to an algorithm presented in~\S\ref{sec:algo-sieve}.
The discussion involves a number of intermediate objects
whose definitions are summarized in Table~\ref{tab:cast-of-characters}.
In terms of the notation in this table,
our approach proceeds by computing an overapproximation~$\cV\trsigma$
of the limit~$\cA\trinfty$ we have to compute.
At the limit, the relaxed cone~$\cV\trinfty$ and the characteristic cone~$\cA\trinfty$ are equal,
and are equal to a direct sum of vector spaces.
Because, by means of Gröbner basis calculations, we are able to detect
nonlinear components of the relaxed cone~$\cV\trsigma$,
we can develop an algorithm by rejection,
which increases~$\sigma$ and refines~$\cV\trsigma$
until we get the exact set of solutions to~\eqref{eq:riccati}.
A plot summary of the analysis to prove our approach is the following:

\begin{itemize}
\item we show that the situation is trivial unless $1 \leq \rho\trsigma \leq 2\zdim-1$;
\item we describe the relaxed cone~$\cV\trsigma$ as the vanishing set of a polynomial system~$\Sigma\trsigma$
  that reflects the degeneracy of the augmented generating matrix~$W_a\trsigma$;
\item some case distinction is then needed:
  \begin{itemize}
  \item if $\rho\trsigma = 2\zdim-1$, all the series solutions of~\eqref{eq:linear} are hypergeometric
    and $\cV\trsigma = \bL^\zdim$;
  \item if $\rho\trsigma < 2\zdim-1$, we isolate the irreducible components
    of~$\cV\trsigma$;
    if one irreducible component is described by a nonlinear system,
    $\sigma$~is too small and the algorithm needs to restart with an increased value;
  \end{itemize}
\item in both cases,
  the relaxed cone~$\cV\trsigma$ is described as a union of linear spaces,
  and there corresponds to each of those spaces a parametrization of a set of candidate solutions to~\eqref{eq:riccati};
\item each linear parametrization is tested by substituting into the left-hand side of~\eqref{eq:riccati}:
  a nonzero evaluation invalidates~$\sigma$, and the algorithm restarts with an increased value;
\item otherwise, the parametrizations thus obtained describe exactly all solutions,
  up to possible inclusions of a parametrization into another.
\end{itemize}

\begin{rem}
The same idea can be adapted to find the rational solutions of a linear Mahler equation.
We further comment on this in~\S\ref{sec:back-to-rat}.
The latter, as well as \S\ref{sec:rank-rels}~on the rank of syzygy modules,
are only incidental to the flow of the text.
\end{rem}

\begin{rem}
The problem of determining those~$a$ in~\eqref{eq:u-of-a} making $u$ rational
simplifies a lot when~$\zdim=1$,
as it  reduces to the question whether
$M z_1/z_1$~is rational.
Using Hermite--Padé approximation together with the bounds
derived in~\S\ref{sec:bounds-for-rational-solutions}
gives an efficient and simple procedure in that context.
\end{rem}

\subsection{Approximate syzygies}
\label{sec:approx-syz}

In order to describe
the linear relations~\eqref{eq:HP-exact-rel}, hereafter called \emph{syzygies},
in connection with the approximate linear relations~\eqref{eq:HP-approx-rel}, hereafter called \emph{approximate syzygies},
we recall and specialize to our setting some useful notions and facts
from the classical theory of Hermite--Padé approximants.

\begin{defi}\label{def:syzygies}
Given an integer~$m \in \bN$ and a tuple $f = (f_1,\dots,f_m) \in \fps\bL^m$,
an element~$w \in \pol\bL^m$ is a \emph{syzygy} of~$f$
when $w\cdot f^T = 0$.
Given an integer $\sigma \in \bN$,
an \emph{approximate syzygy} of~$f$(to order~$\sigma$) is an element~$w \in \pol\bL^m$ satisfying~$w \cdot f^T = O(x^\sigma)$.
The module of syzygies is denoted~$\cS\trinfty$.
The module of approximate syzygies to order~$\sigma$ is denoted~$\cS\trsigma$.
\end{defi}

It is clear that the sequence of the approximate syzygy modules~$\cS\trsigma$ is nonincreasing,
meaning that for any~$\tau \leq \sigma$, $\cS^{[\tau]} \supseteq \cS\trsigma$.
And because a series is~$0$ if and only if it is~$O(x^\sigma)$ for all~$\sigma\in\bN$,
the limit of the sequence is the syzygy module~$\cS\trinfty$.

\begin{rem}
For an approximate syzygy~$w$ of~$f$
there exists a series $q \in \fps\bL$
satisfying $(w,q)\cdot (f,x^\sigma)^T = 0$.
Therefore, the approximate syzygy module is in general larger than
the projection to its first $m$ components of the syzygy module of $(f_1,\dots,f_m,x^\sigma)$,
as the latter would restrict the implied~$q$ to polynomials.
\end{rem}

It is well known \parencite{Derksen-1994-ACG}
that the approximate syzygy module~$\cS\trsigma$ is a free $\pol\bL$-module of rank~$m$,
and that this module admits
a basis of a specific shape, called a \emph{minimal basis} in the recent literature.

\begin{defi}
A basis $(w_1,\dots,w_m)$ of some free $\pol\bL$-module
is called a \emph{minimal basis} if
for each~$i$, the $m$-tuple~$w_i = (p_1,\dots,p_m)$ satisfies
$\deg p_j \leq \deg p_i > \deg p_{j'}$ whenever $j \leq i < j'$.
\end{defi}

Minimal bases are also well known to be Gröbner bases of the approximate syzygy module
for a term-over-position (TOP) ordering \parencite{Neiger2016}.
Fast algorithms for their calculation have been provided,
most notably in \parencite{BeckermannLabahn-1994-UAF};
see also \parencite{BeckermannLabahn-2000-FFC}.

To link the previous definitions
with the notation of the introduction of~\S\ref{sec:syzygies}
and in particular \eqref{eq:HP-structured-rel}, \eqref{eq:HP-exact-rel}, and~\eqref{eq:HP-approx-rel},
we set~$m := 2\zdim$
and, for $1\leq i\leq \zdim$,  $f_i := z_i$ and $f_{i+\zdim} := M z_i$.
If the Riccati equation has nontrivial solutions,
then for a large enough $\sigma$
we expect the existence of an approximate syzygy of the form
\begin{equation}\label{eq:structured-syz}
    w_0 := (-a_1P,\dots,-a_\zdim P,a_1Q,\dots,a_\zdim Q)
\end{equation}
that is in fact an exact syzygy
satisfying $\deg P \leq \Bnum$ and $\deg Q \leq \Bden$
for the bounds introduced in Proposition~\ref{prop:deg-bounds-for-ratfuns}.

Let us consider a minimal basis $(w_1\trsigma,\dots,w_{2\zdim}\trsigma)$ of the approximate syzygy module~$\cS\trsigma$.
Any $w \in \cS\trsigma$ reduces to zero by reduction by the minimal basis
viewed as a Gröbner basis for the TOP order.
Thus, there are polynomials $R_i$ such that
\begin{equation*}
    w = R_1 w_1\trsigma + \dots + R_{2\zdim} w_{2\zdim}\trsigma
\end{equation*}
with $\deg(R_i w_i\trsigma) \leq \deg w$ for $1\leq i \leq 2\zdim$:
for the quotient~$R_i$ to be nonzero, we need $\deg w_i\trsigma \leq \deg w$.
In particular, for any~$d \in \bN$,
any $w$ of degree at most~$d$ reduces to zero by the elements of the minimal basis
themselves of degree at most~$d$.
With our goal to be able to generate structured syzygies~$w_0$
of degree bounded as in~\eqref{eq:structured-syz},
what just precedes justifies retaining only those elements of the minimal basis
whose degrees are at most
\begin{equation}\label{eq:B-infty}
  B_\infty := \max(\Bnum,\Bden) .
\end{equation}

\begin{defi}\label{def:truncated-module}
The\/ \emph{vectorial truncation of the approximate syzygy module}  to degree~$B_\infty$
is the $\bL$-vector space
\begin{equation}\label{eq:approx-syz-mod}
\cT\trsigma :=
\cS\trsigma \cap \pol\bL_{\leq B_\infty}^{2\zdim}
  = \sum_{\deg w_i\trsigma \leq B_\infty} \pol\bL_{\leq B_\infty - \deg w_i\trsigma} w_i\trsigma .
\end{equation}
We write~$W\trsigma$ for the matrix
whose rows are those~$w_i\trsigma$ with degree at most~$B_\infty$,
for $1\leq i\leq 2\zdim$,
and we call it the \emph{generating matrix}.
\end{defi}

Note that the minimal basis remains implicit in the notation of $\cT\trsigma$ and~$W\trsigma$.
This is no problem in practice as, once~$\sigma$ is fixed,
our analysis and algorithm consider a single minimal basis to order~$\sigma$.

\begin{rem}\label{rem:left-deg-right-deg}
Instead of $\cS\trsigma \cap \pol\bL_{\leq B_\infty}^{2\zdim}$ in~\eqref{eq:approx-syz-mod},
one could suggest using
$\cS\trsigma \cap \bigl(\pol\bL_{\leq \Bnum}^\zdim \times \pol\bL_{\leq \Bden}^\zdim\bigr)$,
which fits the degree structure of expected exact syzygies~\eqref{eq:structured-syz} more tightly.
While that should allow for some savings in the initial calculations,
the algorithm would anyway have to deal with much higher degrees
when the kernel computation
(step~\ref{it:loop-over-lambda}\ref{it:HP-iteration}\ref{it:rho<2t-1}\ref{it:loop-over-primes}\ref{it:left-kernel} in Algorithm~\ref{algo:subtask-of-HP})
would result in an output of degree bounded by~$\zdim \times (\Bnum + \Bden) \leq 2\zdim \times B_\infty$.
\end{rem}

We have already said that, for general $f \in \pol\bL^{2\zdim}$,
the sequence of the approximate syzygy modules~$\cS\trsigma$ is nonincreasing
and the rank of~$\cS\trsigma$ remains constant at~$2\zdim$.
We now prove that, for our specific~$f$ defined from~$z_i$,
the sequence is also ultimately decreasing.
To this end, without
loss of generality we can assume that $v := \val z_1 \leq \val z_i$
for~$1\leq i\leq \zdim$.
Thus, for~$\sigma\geq v$, $(x^{\sigma-v},0,\dots,0)$ is in~$\cS\trsigma$ but not in~$\cS^{[\sigma+1]}$.

\begin{lem}
The limit~$\cT\trinfty = \bigcap_{\sigma\in\bN}\cT\trsigma$ is equal to $\cS\trinfty \cap \pol\bL_{\leq B_\infty}^{2\zdim}$,
and the module it spans, $\pol\bL\cT\trinfty$, is an $\pol\bL$-submodule of~$\cS\trinfty$.
\end{lem}

\begin{proof}
The sequence of the~$\cT\trsigma$ is nonincreasing,
and, as a subspace of a fixed finite-dimensional vector space, is ultimately constant.
The result follows after intersecting with~$\pol\bL_{\leq B_\infty}^{2\zdim}$.
\end{proof}

\begin{rem}
Because the degree bound~$B_\infty$ holds a priori only for relations on~$(z_i,Mz_i)$,
the module~$\pol\bL\cT\trinfty$ can be a strict submodule of the syzygy module~$\cS\trinfty$.
\end{rem}

Altogether, we have just proven the next lemma,
which we state after a useful definition.

\begin{defi}\label{def:rank-rho}
For any~$\sigma$, possibly~$\infty$, let $\rho\trsigma$ denote the rank
of the module~$\pol\bL\cT\trsigma$.
\end{defi}

The rank~$\rho\trsigma$ is also the height of the generating matrix~$W\trsigma$
whose rows in~$\pol\bL^{2\zdim}$ generate the module
(see Definition~\ref{def:truncated-module}).

\begin{lem}\label{lem:nonincr-ultim-constant}
For each $\sigma \in \bN$, the $\pol\bL$-module $\pol\bL\cT\trsigma$ is generated
by the rows of the generating matrix~$W\trsigma$.
The sequence of these modules is nonincreasing
and ultimately constant with limit the submodule
$\pol\bL\cT\trinfty$ of\/~$\cS\trinfty$.
A similar property holds for the sequence of their ranks, which we write
\begin{equation}\label{eq:rank-inequalities}
2\zdim \geq \rho\trsigma  \geq \rho\trinfty  .
\end{equation}
\end{lem}

\begin{rem}\label{rem:syzygies-dont-grow}
Given a tuple $f = (f_1,\dots,f_m) \in \fps\bL^m$,
some precision~$\sigma$,
and a minimal basis $w = (w_1,\dots,w_m)$
of the corresponding $\pol\bL$-module~$\cS\trsigma$ of approximate syzygies,
it is immediate, by reduction,
that the basis~$w$ is also
a minimal basis of its $\pol{\bL'}$-module of approximate syzygies
at precision~$\sigma$
for any superfield~$\bL'$ of~$\bL$.
\end{rem}

\subsection{Structured syzygies}
\label{sec:structured-syz}

Given an operator $L \in \pol{\bK}\langle M\rangle$ of order~$r$
with a space of formal power series solutions assumed of $\bL$-dimension~$\zdim > 0$,
we fix a basis $z := (z_1,\dots,z_\zdim)$ of the space of series solutions.
For any $\sigma \in \bN$,
the $\pol\bL$-module $\pol\bL\cT\trsigma$,
which consists of unstructured (approximate) syzygies,
may contain structured syzygies.
We consider in the next definition the set
of those~$a \in \nonz{(\bL^\zdim)}$ that correspond to a structured syzygy,
which together with~$0$ make an $\bL$-cone (see Definition~\ref{def:F-cone}).

\begin{defi}\label{def:cone-A-sigma}
For each~$\sigma \in \bN$,
the \emph{characteristic cone of structured approximate syzygies},
or \emph{characteristic cone} for short,
is the $\bL$-cone
\begin{multline}\label{eq:def-cone-A-sigma-easy}
\cA\trsigma = \{ 0 \} \cup \{ a \in \nonz{(\bL^\zdim)} \mid
  \exists P \in \nonz{\pol\bL}, \
  \exists Q \in \nonz{\pol\bL}, \\
  (-Pa,Qa) \in \pol\bL \cT\trsigma
\} .
\end{multline}
\end{defi}

As a consequence of Lemma~\ref{lem:nonincr-ultim-constant},
the sequence of the characteristic cones~$\cA\trsigma$ is nonincreasing and ultimately constant,
with limit a cone~$\cA\trinfty$ that we proceed to describe.

\begin{lem}\label{lem:direct-sum}
Let $s_1,s_2,\dots$~denote the dimensions of the nontrivial similarity classes~$\nonz{(\fH_1)},\nonz{(\fH_2)},\dots$
of the hypergeometric series solutions of~\eqref{eq:linear}.
There exist full-rank matrices~$H_1,H_2,\dots$
such that:
\begin{itemize}
\item the $\bL$-vector spaces~$\bL^{s_i} H_i$ are in direct sum in~$\bL^\zdim$;
\item the limiting characteristic cone~$\cA\trinfty$ is the union of the vector spaces~$\bL^{s_i} H_i$.
\end{itemize}
\end{lem}

\begin{proof}
Consider an $\rat\bL$-similarity class
of hypergeometric series solutions~$\nonz{\fH}$ of~\eqref{eq:linear},
described by a basis $h := (h_1,\dots,h_s)$.
Expressing this family in the basis~$z$ provides us
with a matrix $H \in \bL^{s\times \zdim}$ such that $h^T = Hz^T$.
\hbox{Each $c \in \nonz{(\bL^s)}$} yields a rational-function solution~$M(ch^T)/(ch^T)$ of~\eqref{eq:riccati},
also equal to~$M(cHz^T)/(cHz^T)$.
After writing $P/Q$ for this rational function,
we get that $(-PcH, QcH)$~is a structured approximate syzygy to any order~$\sigma$,
so that $cH$~is in the cone~$\cA\trinfty$.
This cone therefore contains~$\bL^s H$.
Iterating over the nontrivial similarity classes
$\nonz{(\fH_1)}$, $\nonz{(\fH_2)}$,~\dots
of hypergeometric series solutions,
we therefore get dimensions~$s_1,s_2,\dots$ and matrices~$H_1,H_2,\dots$
such that each~$\bL^{s_i} H_i$ is contained in~$\cA\trinfty$.
By point~\ref{it:vect-H} of Theorem~\ref{thm:structure-main} applied
to $\dring = \omdr_\bL$ (see Definition~\ref{def:omdr}) and $F = \fls\bL$,
the~$\fH_i$ are in direct sum in~$\omdr_\bL$,
inducing, by the map $a \mapsto az^T$,
that the~$\bL^{s_i} H_i$ are in direct sum in the $\bL$-space generated by~$\cA\trinfty$.
Conversely, any nonzero~$a \in \cA\trinfty$ yields
a nonzero rational solution~$P/Q$ of the Riccati Mahler equation,
thus, by point~\ref{it:param-H} of Theorem~\ref{thm:structure-main}
must come from some~$\nonz{(\fH_i)}$.
So, the cone~$\cA\trinfty$ consists solely of the union of the $\bL$-vector spaces~$\bL^{s_i} H_i$.
\end{proof}

In view of~\eqref{eq:rank-inequalities},
if some rank~$\rho\trsigma$ is~$0$, then~$\rho\trinfty = 0$,
meaning there are no nontrivial syzygies at all,
and our algorithm will terminate in this case.
So, we continue the analysis by assuming~$\rho\trinfty > 0$.
Additionally,
if the rank~$\rho\trinfty$ were~$2\zdim$,
there would exist some finite~$\sigma$ for which~$\rho\trsigma = 2\zdim$,
which would yield a nonsingular square matrix~$W\trsigma$.
We could then perform row operations over~$\pol\bL$ so as to
produce $2\zdim$ nonzero polynomials~$p_i$ satisfying
$p_i z_i = p_{\zdim+i} M z_i = 0$ for $1 \leq i \leq \zdim$,
but this would be a contradiction to the non-nullity of the~$z_i$.
Thus, we have~$0 < \rho\trinfty < 2\zdim$.
Consequently, we also have~$\rho\trsigma < 2\zdim$ for large enough~$\sigma$. If, conversely, we observe~$\rho\trsigma = 2\zdim$,
then we know that $\sigma$~is too small to provide a description of the wanted solutions,
and our algorithm will continue with increased~$\sigma$.

\begin{ex}\label{ex:adamczewski-faverjon-case-no}
We use as a working example throughout~\S\ref{sec:syzygies} the Mahler operator of order~$r = 4$ and degree~$d = 258$ introduced in Example~\ref{ex:adamczewski-faverjon-intro}.
This operator admits a four-dimensional vector space of formal power series solutions in~$\fps\bQ$.
For any $\sigma \leq 120$,
the computation of the generating matrix~$W\trsigma$ (Definition~\ref{def:truncated-module})
leads to a rank $\rho\trsigma = 8 = 2\zdim$  (Definition~\ref{def:rank-rho}),
proving that such~$\sigma$ are too low to provide candidates.
\end{ex}

\begin{hyp}\label{hyp:rho-infty>0}
For \S\ref{sec:linalg-cond}, \S\ref{sec:cand-sols}, and~\S\ref{sec:validation},
we assume $1 \leq \rho\trinfty$.
\end{hyp}

\subsection{A linear-algebra condition on structured syzygies}
\label{sec:linalg-cond}

From now on, we assume~$1 \leq \rho\trsigma \leq 2\zdim-1$.
The following lemma is a variant of Cramer's rule.

\begin{lem}\label{lem:n+1-times-n}
Let $M$ be a matrix of size $(n+1) \times n$ and rank~$n$,
with coefficients in a field~$F$.
Then, the left kernel of~$M$ has dimension~$1$ over~$F$
and a nonzero kernel element is
$K := \bigl(\Delta_1, \dots, (-1)^{i+1}\Delta_i, \dots, (-1)^n\Delta_{n+1}\bigr)$,
where $\Delta_i$~denotes the determinant
of the square submatrix obtained by removing the $i$th row.
\end{lem}
\begin{proof}
Augmenting~$M$ by any column~$C$ of it on its left
yields a matrix with determinant~$0$.
Expanding the determinant with respect to the first column
shows that $K$~satisfies $KC = 0$.
Combining all columns~$C$ yields ${KM = 0}$.
Because~$\rk M = n$,
the minors~$\Delta_i$ cannot all be~$0$ simultaneously,
which implies~$K\neq0$.
\end{proof}

The characteristic cone~$\cA\trsigma$ was defined by~\eqref{eq:def-cone-A-sigma-easy}.
We proceed to formulate an alternative description of it in
terms of the generators of the module~$\pol\bL\cT\trsigma$
provided by the rows of~$W\trsigma$.

\begin{defi}\label{def:W-a}
Let~$W_a\trsigma$ denote the matrix obtained by augmenting~$W\trsigma$ at its bottom
with the two-row matrix
\begin{equation}\label{eq:two-row-matrix}
\begin{pmatrix}
a & 0 \\
0 & a
\end{pmatrix} =
\begin{pmatrix}
a_1 & \dots & a_\zdim & 0 & \dots & 0 \\
0 & \dots & 0 & a_1 & \dots & a_\zdim
\end{pmatrix} .
\end{equation}
We call it the \emph{augmented generating matrix}.
\end{defi}

\begin{prop}\label{prop:cone-A-sigma-alt}
The characteristic cone~$\cA\trsigma$ (Definition~\ref{def:cone-A-sigma}) satisfies
\begin{multline}\label{eq:def-cone-A-sigma}
\cA\trsigma = \{ 0 \} \cup \{ a \in \nonz{(\bL^\zdim)} \mid
  \exists P \in \nonz{\pol\bL}, \
  \exists Q \in \nonz{\pol\bL}, \
  \exists \leftkern \in \pol\bL^{\rho\trsigma}, \\
  (\leftkern_1,\dots,\leftkern_{\rho\trsigma},P,-Q)W_{a}\trsigma = 0
\} .
\end{multline}

\end{prop}

Note that this characterization
is in fact independent of the  choice of the generating matrix~$W\trsigma$.
We introduce the cone~$\cV\trsigma$ defined analogously
by the looser constraint $(\leftkern,P,-Q) \neq 0$
in place of $P \neq 0$ and~$Q \neq 0$:
this is the cone of values of~$a$
for which $W_{a}\trsigma$ has a nontrivial left kernel.

\begin{defi}\label{def:cone-V-sigma}
For each~$\sigma \in \bN$,
the \emph{relaxed cone of structured approximate syzygies},
or \emph{relaxed cone} for short,
is the $\bL$-cone
\begin{equation}\label{eq:def-cone-V-sigma}
\cV\trsigma = \{ 0 \} \cup \{ a \in \nonz{(\bL^\zdim)} \mid
  \exists K \in \nonz{(\pol\bL^{\rho\trsigma+2})}, \ K W_a\trsigma = 0
\} .
\end{equation}
\end{defi}

Note the inclusion $\cA\trsigma \subseteq \cV\trsigma$
of the characteristic cone into the relaxed cone.
The following lemma describes a case of equality.

\begin{lem}\label{lem:P-Q-nonzero}
For $\sigma$~large enough, meaning $\cT\trsigma = \cT\trinfty$,
assume $1 \leq \rho\trsigma \leq 2\zdim-1$.
Then:
\begin{enumerate}
\item\label{it:P-Q-nonzero}
if, for some nonzero~$a$,
the row $K = (\leftkern_1,\dots,\leftkern_{\rho\trsigma},-P,Q) \neq 0$
satisfies $K W_{a}\trsigma = 0$,
then neither~$P$ nor~$Q$ can be zero;
\item\label{it:sigma-large-cV=cA} $\cA\trsigma = \cV\trsigma = \cA\trinfty$.
\end{enumerate}
\end{lem}
\begin{proof}
Taking~$a$ as in the first point implies
$Pa z^T = Qa Mz^T + \leftkern W\trsigma (z,Mz)^T$,
which reduces to~$Pa z^T = Qa Mz^T$
because $W\trsigma (z,Mz)^T = 0$ by the hypothesis on~$\sigma$.
Now, if $P$~is zero, then $Qa = 0$
because the entries of~$M z$ are independent over~$\bL$,
forcing~$Q = 0$ because~$a\neq0$.
Similarly, $Q = 0$ implies $P = 0$.
So, if $P$ or~$Q$ is zero, then both must be zero,
which contradicts that
the generating matrix~$W\trsigma$ has independent rows over~$\rat\bL$,
thus proving $\cV\trsigma \subseteq \cA\trsigma$,
and therefore $\cA\trsigma = \cV\trsigma$.
Moreover, if two integers $\sigma_1$ and~$\sigma_2$ are
such that $\cT^{[\sigma_1]} = \cT^{[\sigma_2]}$,
then $\cA^{[\sigma_1]} = \cA^{[\sigma_2]}$
by the definition~\eqref{eq:def-cone-A-sigma-easy}.
In particular, for all~$\sigma$ satisfying $\cT\trsigma = \cT\trinfty$,
all the cones~$\cA\trsigma$ are equal and equal to~$\cA\trinfty$.
\end{proof}

\begin{lem}
The limit~$\cV\trinfty$ of relaxed cones exists and is equal to the limit~$\cA\trinfty$ of characteristic cones.
\end{lem}
\begin{proof}
The existence of~$K$ in the definition of~$\cV\trsigma$ is equivalent
to the existence of a nonzero element in
\begin{equation*}
\pol\bL^{\rho\trsigma} W\trsigma \cap
\pol\bL^2 \left(\begin{smallmatrix} a & 0 \\ 0 & a \end{smallmatrix}\right)
= \pol\bL \cT\trsigma \cap
\pol\bL^2 \left(\begin{smallmatrix} a & 0 \\ 0 & a \end{smallmatrix}\right) .
\end{equation*}
By Lemma~\ref{lem:nonincr-ultim-constant},
this intersection is nonincreasing and ultimately constant,
so that $\cV\trsigma$~is nonincreasing and ultimately constant as well.
This shows the existence of~$\cV\trinfty$.
Now, by Hypothesis~\ref{hyp:rho-infty>0},
if $\rho\trsigma < 2\zdim$ for some~$\sigma$,
then for all~$\tau \geq \sigma$, we have $1 \leq \rho^{[\tau]} \leq 2\zdim-1$.
So if some~$\sigma$ satisfies the hypotheses of Lemma~\ref{lem:P-Q-nonzero},
then all~$\tau \geq \sigma$ satisfy them as well.
We deduce~$\cV\trinfty = \cA\trinfty$.
\end{proof}

\begin{rem}
In order to compute~$\cA\trinfty$,
Algorithm~\ref{algo:HP} will not manipulate~$\cA\trsigma$ directly
and will manipulate~$\cV\trsigma$ instead.
It will therefore work by increasing~$\sigma$
until the equalities of Lemma~\ref{lem:P-Q-nonzero}(\ref{it:sigma-large-cV=cA}) are satisfied.
\end{rem}

\subsection{Candidate solutions}
\label{sec:cand-sols}

Two cases need different analyses:
\emph{(i)\/}~$\rho\trsigma = 2\zdim-1$;
\emph{(ii)\/}~$\rho\trsigma \leq 2\zdim-2$.

\subsubsection{Case~$\rho\trsigma = 2\zdim-1$.}\label{sec:case-i}

In the first case,
the $(\rho\trsigma+2) \times (\rho\trsigma+1)$-matrix~$W_{a}\trsigma$ has
a nontrivial left kernel for any value of~$a$:
the relaxed cone~$\cV\trsigma$ is the whole space~$\bL^\zdim$.
Additionally,
the $\rho\trsigma \times (\rho\trsigma+1)$-matrix~$W\trsigma$ has full rank over~$\rat\bL$,
and thus has a $1$-dimensional right kernel.
Applying Lemma~\ref{lem:n+1-times-n} to the transposed matrix
provides a polynomial generator
$K^T = (\Delta_1,\dots,(-1)^{i+1}\Delta_i,\dots,-\Delta_{2\zdim})^T$
of this right kernel,
given by minors~$\Delta_i$ obtained by removing columns.
As ranks are unchanged by extending the base field to~$\fls\bL$,
the generating matrix~$W\trsigma$ also has a $1$-dimensional right kernel over~$\fls\bL$,
and the same~$K^T$ is a generator of this $\fls\bL$-vector space.

\begin{lem}\label{lem:formula-when-rho=2t-1}
For $\sigma$~large enough, meaning $\cT\trsigma = \cT\trinfty$,
assume $\rho\trsigma = 2\zdim-1$.
Then:
\begin{itemize}
\item $\cA\trsigma = \cV\trsigma = \bL^\zdim$,
\item  those rational solutions $u\in \rat\bL$ to the Riccati equation~\eqref{eq:riccati}
that are  quotients of the form $M w / w$
for $w = \sum_{i=1}^\zdim a_i z_i$ and $a \in \nonz{(\bL^\zdim)}$
are given by the parametrization
\begin{equation}\label{eq:cand-from-Delta}
(a_1:\ldots:a_\zdim) \mapsto u =
  \frac{\sum_{i=1}^\zdim a_i (-1)^{i+\zdim+1} \Delta_{i+\zdim}}{\sum_{i=1}^\zdim a_i (-1)^{i+1} \Delta_i} ,
\qquad
\text{with $a \in \bP^{\zdim-1}(\bL)$.}
\end{equation}
\end{itemize}
\end{lem}

\begin{proof}
By the assumption on~$\sigma$,
$(z_1,\dots,z_\zdim,Mz_1,\dots,Mz_\zdim)^T$ must be in the kernel of the generating matrix~$W\trsigma$,
and thus must be a multiple of~$K^T$ by a series~$z_0 \in \nonz{\fls\bL}$.
In particular, $z_i = (-1)^{i+1}\Delta_i z_0$ for~$1\leq i\leq \zdim$,
so all the~$z_i$ are $\rat\bL$-similar.
For any $a \in \nonz{(\bL^\zdim)}$,
the nonzero series $w = \sum_{i=1}^\zdim a_i z_i$ is a solution to~\eqref{eq:linear}
and an immediate calculation shows that the quotient~$M w / w$
is a rational solution from~$\rat\bL$
given by the formula in~\eqref{eq:cand-from-Delta}.
That is, $w$~is hypergeometric and,
by point~\ref{it:param-H} of Theorem~\ref{thm:structure-main},
$u := M w / w$ is a solution to~\eqref{eq:riccati}.
This shows that \eqref{eq:cand-from-Delta}~parametrizes
all solutions~$u$ that are quotients~$M w / w$ and rational.
Any such solution, given by~$a\neq0$, provides
nonzero $P$ and~$Q$ satisfying $Mw/w = P/Q$.
The nullity of
$(P,-Q)
\left(\begin{smallmatrix}
a & 0 \\
0 & a
\end{smallmatrix}\right)
(z, Mz)^T$
implies that
$(P,-Q)
\left(\begin{smallmatrix}
a & 0 \\
0 & a
\end{smallmatrix}\right)$
is a linear combination of the rows of~$W\trsigma$ with rational function coefficients,
that is, there exists~$\leftkern \in \rat\bL^{\rho\trsigma}$ satisfying
$(\leftkern, P, -Q) W_{a}\trsigma = 0$.
After clearing denominators, this shows~$a \in \cA\trsigma$.
So we proved the equality $\cA\trsigma = \bL^\zdim$,
from which derives the equality $\cV\trsigma = \bL^\zdim$
because $\cA\trsigma \subseteq \cV\trsigma \subseteq \bL^\zdim$.
\end{proof}

The algorithm does not decide at first if $\cT\trsigma = \cT\trinfty$.
Instead, the formula~\eqref{eq:cand-from-Delta} of Lemma~\ref{lem:formula-when-rho=2t-1}
introduces
a parametrized candidate for rational solutions.
If this parametrized rational function
is later verified to be an actual parametrized solution,
in other words if the algorithm reaches a~$\sigma$ large enough to ensure
\begin{equation*}
u = \frac{\sum_i a_i z_{i+\zdim}}{\sum_i a_i z_i} ,
\end{equation*}
then the similarity class of hypergeometric solutions of the~$z_i$
is proved to have $\bL$-dimension~$\zdim$.

\begin{ex}[continuing from Example~\ref{ex:adamczewski-faverjon-case-no}]
\label{ex:adamczewski-faverjon-case-i}
For each $\sigma$ between $121$ and~$123$,
the computation of the generating matrix~$W\trsigma$
leads to a rank $\rho\trsigma = 7 = 2\zdim-1$.
In each case,
the rational candidate~\eqref{eq:cand-from-Delta}
has a numerator of degree~$81$ in~$x$
and a denominator of degree~$83$,
both involving linear parameters $a_1,\dots,a_4$.
It will turn out, however, that this candidate need not
parametrize only solutions
(see Example~\ref{ex:adamczewski-faverjon-case-conclusion}).
\end{ex}

\subsubsection{Case~$1\leq \rho\trsigma \leq 2\zdim-2$.}\label{sec:case-ii}

In the second case,
the matrix~$W_{a}\trsigma$ is of size $(\rho\trsigma+2) \times (2\zdim)$
with $\rho\trsigma+2 \leq 2\zdim$.
This matrix is full rank over $\bL(x,a_1,\dots,a_\zdim)$ and
we want to find all choices of~$a = (a_1,\dots,a_\zdim)$ in~$\bL^\zdim$
that make the rank drop.
They are described by canceling all maximal minors.
These minors are polynomials in $a$ and~$x$
that are homogeneous of degree~$2$ in the~$a_i$.
We gather their coefficients with respect to~$x$ into a
system~$\Sigma\trsigma$ of polynomials of~$\bL[a]$.

\begin{defi}\label{def:Sigma-pol-sys}
Let $\Sigma\trsigma \subseteq \polm\bL$ denote the polynomial system
consisting of the coefficients according to the monomial basis~$(x^m)_{m\in\bN}$
of the maximal minors of the augmented generating matrix~$W_a$.
\end{defi}

So, the cone~$\cV\trsigma$ of all~$a$ that make $W_{a}\trsigma$ rank deficient is
the variety of common zeros in~$\bL^\zdim$ of~$\Sigma\trsigma$.
The following lemma makes the description of the equality case $\cA\trsigma = \cV\trsigma$
in Lemma~\ref{lem:P-Q-nonzero} more explicit.

\begin{lem}\label{lem:spec-rk-is-rho+1}
For $\sigma$~large enough, meaning $\cT\trsigma = \cT\trinfty$,
assume $1 \leq \rho\trsigma \leq 2\zdim-2$.
Then,
for any nonzero~$a$ of\/~$\bL^\zdim$ that makes~$W_{a}\trsigma$ rank deficient,
the rank over $\rat\bL$ of~$W_{a}\trsigma$ is $\rho\trsigma + 1$.
\end{lem}

\begin{proof}
Because $\rk W\trsigma = \rho\trsigma$,
the rank~$\rk W_{a}\trsigma$ is either~$\rho\trsigma$ or~$\rho\trsigma + 1$.
In the former case,
the first $\rho\trsigma$ rows of~$W_{a}\trsigma$ generate its image $\rat\bL^{\rho\trsigma+2} W_{a}\trsigma$,
and there exists $K = (\leftkern_1,\dots,\leftkern_{\rho\trsigma},-P,0) \in \pol\bL^{\rho\trsigma + 2}$ with~$P \neq 0$
satisfying $K W_{a}\trsigma = 0$.
Consequently,
$Pa z^T = (\leftkern_1 w_1\trsigma + \dots + \leftkern_{\rho\trsigma} w_{\rho\trsigma}) (z,Mz)^T$,
which is~$0$ because $W\trsigma (z,Mz)^T = 0$ by the hypothesis on~$\sigma$.
It follows that~$az^T=0$, which contradicts that $z$~is an $\bL$-basis.
So the rank must be~$\rho\trsigma + 1$.
\end{proof}

By Lemma~\ref{lem:direct-sum},
the limit cone~$\cA\trinfty$ is a union of linear spaces.
The relaxed cone~$\cV\trsigma$ approximates the characteristic cone~$\cA\trsigma$ by containing it
and is equal to it for sufficiently large~$\sigma$.
Our algorithm therefore computes
a primary decomposition of the radical of the ideal generated by~$\Sigma\trsigma$,
then tests whether each prime ideal thus obtained describes a linear space.

There exist algorithms for computing primary decompositions over an algebraically closed field like~$\bKbar$.
Algorithms to be found in the literature \parencite{GianniTragerZacharias-1988-GBP,DeckerGreuelPfister-1999-PDA} return primary ideals represented by (minimal reduced) Gröbner bases,
and we will prove by Theorem~\ref{thm:radical-as-linear-ideals}
that testing linearity of the irreducible components of the relaxed cone~$\cV\trsigma$
amounts to checking that all returned ideals are generated by linear forms.

So, on the one hand,
if any element of any of the (minimal reduced) Gröbner bases is nonlinear,
then our algorithm restarts with a larger~$\sigma$.

On the other hand,
if all Gröbner bases define linear spaces,
our algorithm extracts a more explicit parametrization of candidates~$u = P/Q$.
To this end, take each Gröbner basis in turn,
solve the system it represents for~$a$ to determine
parameters~$(g_1,\dots,g_v)$ and some matrix~$S$ of size $v\times \zdim$
such that $a = gS$.
The following lemma shows the resulting form of candidate solutions~$P/Q$,
parametrized by~$g$.

\begin{lem}\label{lem:when-rho<2t-1}
Assume $1 \leq \rho\trsigma \leq 2\zdim-2$ and
let $g \in \bL^v$ satisfy $gS \in \nonz{(\cV\trsigma)}$.
Then, the left kernel of\/~$W_{a}\trsigma$
is generated over\/~$\rat\bL$ by
$K = (\leftkern_1,\dots,\leftkern_{\rho\trsigma},-P,Q)$
where $P$ and~$Q$ are nonzero and depend linearly on~$g$,
and the~$\leftkern_i$ are quadratic in~$g$.
\end{lem}
\begin{proof}
The matrix~$W_{a}\trsigma$
has size $(\rho\trsigma + 2) \times (2\zdim)$
and by Lemma~\ref{lem:spec-rk-is-rho+1} it has $\rat\bL$-rank $\rho\trsigma + 1$.
As a consequence, we can select $\rho\trsigma + 1$ linearly independent columns
to obtain a matrix
$\Omega(g)$ of rank~$\rho\trsigma + 1$,
made of an upper $\rho\trsigma \times (\rho\trsigma + 1)$ block extracted from~$W_{a}\trsigma$ and independent of~$g$,
and of two rows depending linearly on~$g$.
Lemma~\ref{lem:n+1-times-n} provides
a nonzero~$K = (\leftkern_1,\dots,\leftkern_{\rho\trsigma},-P,Q)$
satisfying~$K\Omega(g) = 0$,
with nonzero $P$ and~$Q$ by Lemma~\ref{lem:P-Q-nonzero}.
Since the columns of~$\Omega(g)$ generate those of~$W_{a}\trsigma$,
the product~$K W_{a}\trsigma$ is zero as well.
The definition of~$K$ by minors
and the degree structure of~$W_{a}\trsigma$ with respect to~$a$
provide the degrees in~$g$ of its entries.
\end{proof}

\begin{ex}[continuing from Example~\ref{ex:adamczewski-faverjon-case-i}]
\label{ex:adamczewski-faverjon-case-ii}
For $\sigma = 124$, computing the generating matrix~$W\trsigma$
leads to a rank $\rho\trsigma = 6 = 2\zdim-2$.
Augmenting~$W\trsigma$ by stacking the two-row matrix~\eqref{eq:two-row-matrix}
and taking the coefficients with respect to~$x$ of the single minor to be considered
yields a system~$\Sigma\trsigma$ of $30$~polynomials
of degree~$2$ in $a_1,\dots,a_4$.
The corresponding irredundant prime decomposition is $\fI_1 \cap \fI_2 \cap \fI_3$ for
$\fI_1 = \ideal{a_1 - a_3, a_2 - a_4}$, \
$\fI_2 = \ideal{a_1 - a_2, a_1 + a_3, a_1 + a_4}$, \
$\fI_3 = \ideal{a_1 - a_4, a_1 + a_2, a_1 + a_3}$.
The associated linear subspaces of~$\bL^4$,
predicted to be in direct sum by Lemma~\ref{lem:direct-sum},
have respective dimensions $2$, $1$,~and~$1$.
In the present case, these dimensions add up to~$\zdim = 4$.
The ideals $\fI_2$ and~$\fI_3$ yield isolated candidates~$P/Q$ (projective dimension~$0$),
while the first ideal yields a candidate parametrized by a projective line
(projective dimension~$1$).
Focusing on $\fI_1$, we solve it as
$(a_1, a_2, a_3, a_4) = (g_1, g_2) \begin{pmatrix} 1 & 0 & 1 & 0 \\ 0 & 1 & 0 & 1 \end{pmatrix}$.
The specialized augmented matrix becomes
\[
  W_{a}\trsigma = \begin{pmatrix}
    x & -1 & x & -1 & 0 & 0 & 0 & 0 \\
    p_{2,1}^{(36)} &
    p_{2,2}^{(37)} &
    p_{2,3}^{(0)} &
    p_{2,4}^{(37)} &
    p_{2,5}^{(0)} &
    p_{2,6}^{(1)} &
    p_{2,7}^{(0)} &
    p_{2,8}^{(1)} \\
    0 & 0 & -1 & 0 & x^2 & 0 & 1 & x \\
    0 & 0 & 0 & -1 & 0 & x^2 & x & 1 \\
    -1 & 0 & 0 & 0 & 1 & x & x^2 & 0 \\
    0 & -1 & 0 & 0 & x & 1 & 0 & x^2 \\
    g_1 & g_2 & g_1 & g_2 & 0 & 0 & 0 & 0 \\
    0 & 0 & 0 & 0 & g_1 & g_2 & g_1 & g_2
  \end{pmatrix}
,
\]
where the polynomials~$p_{2,j}^{(d_j)}$ in the second row have degrees~$d_j$
 for $1 \leq j \leq 8$.
This matrix has $\rho\trsigma+2 = 8$~rows and we observe that the first $7$ columns are linearly independent
over $\bL(x,g_1,g_2)$,
leading to an $8\times7$ matrix~$\Omega(g)$
that plays the same role as the matrix with same name in the proof of Lemma~\ref{lem:when-rho<2t-1}.
The determinants of the submatrices of~$\Omega(g)$ obtained by deleting
the $j$th row for $1 \leq j \leq 8$ are all of the form
$\Delta_j = q \delta_j$ with~$q$ a degree~$36$ polynomial and
\begin{multline*}
\delta_1 = -(g_1^2 - g_2^2)x  , \qquad
\delta_2 = 0,\\
\delta_3 = \delta_5 =
- (g_1 - g_2 x + g_1 x^2) (g_1 + g_2 x)  , \qquad
\delta_4 = \delta_6 =
(g_2 - g_1 x + g_2 x^2) (g_1 + g_2 x) ,\\
\delta_7 = -(g_1 + g_2 x^{3})   , \qquad
\delta_8 = -(1 + x^2 + x^4) (g_1 + g_2 x).
\end{multline*}
The left kernel is the $\rat\bL$-line generated by
$K := (\Delta_1, \dots, \Delta_7, -\Delta_8)$.
This provides a rational candidate
\begin{equation}\label{eq:adamczewski-faverjon-PQ-cand}
\frac PQ = \frac{\Delta_7}{\Delta_8} =
\frac{g_1 + g_2 x^3}{g_1 + g_2 x}\frac{1}{1 + x^2 + x^4} .
\end{equation}
\end{ex}

\subsection{Validating candidate solutions}
\label{sec:validation}

Whether obtained when $\rho\trsigma = 2\zdim-1$ (\S\ref{sec:case-i})
or $\rho\trsigma \leq 2\zdim-2$ (\S\ref{sec:case-ii}),
a candidate~$u = P/Q$ has its numerator and denominator
parametrized linearly by some $g = (g_1,\dots,g_v)$ (where $g = a$ and~$v = \zdim$ if in the former case).
So a candidate is really a parametrization by~$g$ in some~$\nonz{\bL^v}$
of a family of rational functions from~$\rat\bL$.
Such a candidate parametrization is obtained algorithmically as a rational function over a transcendental extension of~$\bL$
generated by symbolic generators that represent the parameters~$g$.

Suppose a candidate parametrization could provide genuine solutions for some specializations of~$g$ as well as rational functions that do not solve~\eqref{eq:riccati} for some other specializations of~$g$.
For sure, enlarging~$\sigma$ would ultimately reject the wrong specializations.
However, one could try to avoid these new calculations
and hope to separate the correct parameters from the wrong ones
by using just the data at hand.
For example, one could try to
substitute~$P/Q$ for~$u$ in~\eqref{eq:riccati}
and identify coefficients to~$0$,
but this
would lead to a polynomial system of degree~$r$ in~$g$,
which we would not know how to solve efficiently.

So, verifying that a sufficiently large~$\sigma$ has been used
amounts to verify for each parametrized candidate
that $P/Q$~is truly a solution of~\eqref{eq:riccati}
for \emph{all} choices of~$g$.
As a first criterion,
any fixed parametrized candidate in normal form~$P/Q$ can thus be rejected
if either $\deg P > \Bnum$ or~$\deg Q > \Bden$,
for the bounds defined by \eqref{eq:bound-P} and~\eqref{eq:bound-Q}.
Otherwise,
after substituting~$P/Q$ for~$u$ in the left-hand side of~\eqref{eq:riccati},
either we obtain~$0$ and $P/Q$~is a valid solution,
or $\sigma$ has to be increased.

\begin{ex}[continuing from Example~\ref{ex:adamczewski-faverjon-case-ii}]
\label{ex:adamczewski-faverjon-case-conclusion}
Gathering the results from the previous examples, we see that
the rank~$\rho\trsigma$
of the submodule~$\pol\bL\cT\trsigma$ defined by~\eqref{eq:approx-syz-mod} of the approximate syzygy module~$\cS\trsigma$
is nonincreasing when $\sigma$~increases.
Meanwhile, the cone~$\cA\trsigma$ of parameters
of candidate solutions decreases as well,
until it reaches the limit~$\cA\trinfty$ that
corresponds to true solutions only.
Let us describe this in more detail:
\begin{list}{}{\setlength{\leftmargin}{10mm}\setlength{\labelwidth}{5mm}}  
\item[$\bullet$ \ $\rho\trsigma = 2\zdim = 8$ for $\sigma \leq 120$.]
  We already explained that these~$\sigma$ are too small.
\item[$\bullet$ \ $\rho\trsigma = 2\zdim-1 = 7$ for $121 \leq \sigma \leq 123$.]
  As the degrees in~$x$ of each candidate are
  $81 > 344/9 = \Bnum \simeq 38.2$ for the numerator
  and $83 > \Bden = 86/3 \simeq 28.7$ for the denominator,
  we reject those candidates.
\item[$\bullet$ \ $\rho\trsigma = 6 \leq 2\zdim-2$ for $124 \leq \sigma \leq 126$.]
  For $\sigma = 124$ and the
  parametrized candidate~\eqref{eq:adamczewski-faverjon-PQ-cand},
  substituting into the Riccati equation proves
  that \eqref{eq:adamczewski-faverjon-PQ-cand} defines a family of true rational solutions parametrized by $(g_1:g_2) \in \bP^1(\bQ)$.
  Similarly, the ideals $\fI_2$ and~$\fI_3$ of Example~\ref{ex:adamczewski-faverjon-case-ii} each yields one solution.
  As all candidates are solutions, $\cA\trsigma = \cA\trinfty$ for this~$\sigma$,
  and therefore for all~$\sigma \geq 124$.
\item[$\bullet$ \ $\rho\trsigma = 5 \leq 2\zdim-2$ for $127 \leq \sigma$.]
  As we will see in~\S\ref{sec:rank-rels},
  each similarity class of hypergeometric solutions of dimension~$s$ (over $\bL = \bQ$)
  contributes $2s-1$ linearly independent syzygies
  (of degree at most $B_\infty = \Bnum \simeq 38.2$),
  and in presence of several classes,
  the~$2s-1$ add up to a lower bound
  on the rank~$\rho\trinfty$ of the submodule~$\pol\bL \cT\trinfty$
  of the module of syzygies on hypergeometric solutions.
  In the present example, this lower bound on~$\rho\trinfty$ is
  $(2\times2-1) + (2\times1-1) + (2\times1-1) = 5$.
  A computation for~$\sigma = 127$ finds $\rho\trsigma = 5$,
  proving that $\rho\trsigma = \rho\trinfty = 5$,
  as well as $\cT\trsigma = \cT\trinfty$,
  for all~$\sigma \geq 127$.
\end{list}
The reader might be surprised to note
that the equality~$\cA\trsigma = \cA\trinfty$ is satisfied
before $\cT\trsigma = \cT\trinfty$~occurs:
indeed,
the definition of~$\rho\trsigma$ as the rank of~$\pol\bL \cT\trsigma$
and the inequality $\rho^{[124]} = 6 > \rho^{[127]} = 5 = \rho\trinfty$
imply that $\cT^{[124]}$~is strictly larger than~$\cT\trinfty$,
whereas $\cA^{[124]} = \cA\trinfty$.
\end{ex}

\subsection{Ramified rational solutions}
\label{sec:red-fps}

As a consequence of Theorem~\ref{thm:structure-puiseux} (and Definition~\ref{def:Lambda'}),
\emph{all} the ramified rational solutions of~\eqref{eq:riccati} with a given leading coefficient~$\lambda$
are images by~${y\mapsto My/y}$ of the $\bL$-vector space of solutions~$y$
of~\eqref{eq:linear} in $\mylog{\lambda} \rfls\bL{q_\lambda}$
for the ramification bound~$q_\lambda$ given by~\eqref{eq:q-lambda},
and this set is nonempty if and only if~$\lambda \in \bL\cap\Lambda'$.
When nonempty,
this is~$\ricsol_{\bL,\lambda}$.

Let us fix  a $\lambda \in \bL\cap\Lambda'$.
Then, we first change the
operator $L = L(x, M)$ into $L(x,\lambda M)$ to reduce
the search in $\mylog{\lambda}\rfls\bL{q_\lambda}$ to a search in $\rfls\bL{q_\lambda}$,
and we next apply the transformation explained
in \parencite[\S2.7, particularly Lemma~2.21]{ChyzakDreyfusDumasMezzarobba-2018-CSL}
to reduce the search in~$\rfls\bL{q_\lambda}$ to a search in~$\fps\bL$.
To make the operator resulting from those combined transformations more explicit:
consider the rightmost $\lambda$-admissible edge
of the lower Newton polygon of~$L$ whose denominator is coprime with~$b$
(see Definition~\ref{def:lower-char-poly});
express its slope as~$-p_\lambda/q_\lambda$,
which is possible by the definition of~$q_\lambda$;
and denote by~$c_\lambda$ its intercept on the ordinate axis.
By defining
\begin{equation}\label{eq:operator-change}
L_\lambda(x, M) = x^{-q_\lambda c_\lambda} L(x^{q_\lambda},\lambda M) x^{p_\lambda} ,
\end{equation}
where the factor $x^{-q_\lambda c_\lambda}$ simply ensures
that the coefficients of~$L_\lambda$ are coprime polynomials,
we obtain a linear Mahler operator
whose solutions~$z \in \fps\bL$ parametrize
the solutions~$y$ of~$L$ in~$\mylog{\lambda}\rfls\bL{q_\lambda}$
by $y(x) = \mylog{\lambda} x^{p_\lambda/q_\lambda} z(x^{1/q_\lambda})$.
This parametrization is bijective,
owing to \parencite[Prop.~2.19]{ChyzakDreyfusDumasMezzarobba-2018-CSL}.
To obtain all the ramified rational solutions~$u$ of~\eqref{eq:riccati},
it is therefore sufficient
to apply the method of the previous subsections to each~$L_\lambda$
and to obtain~$u = My/y$ for~$y$ defined from~$z$ as above.

\begin{rem}\label{rem:sols-puisseux-dans-fps}
An important fact,
to be used when applying Theorem~\ref{thm:radical-as-linear-ideals}
in the proof of Theorem~\ref{thm:HP-correct},
is that the solution space of~$L_\lambda$ in~$\puiseux\bL$
is in fact included in~$\fps\bL$.
\end{rem}

\begin{rem}\label{rem:decide-Lambda'}
The set~$\Lambda$ is clearly computable (see~\eqref{eqn:lower-char-poly}).
The previous reasoning shows that
its subset~$\Lambda'$ (Definition~\ref{def:Lambda'}) is also computable:
determining the series solutions of~$L_\lambda$ for all~$\lambda \in \Lambda$
will decide which of the~$\lambda$ is in~$\Lambda'$.
However, the calculation of~$\Lambda'$ will remain implicit in Algorithm~\ref{algo:HP}:
it will loop over~$\lambda \in \Lambda$ and just continue with the next~$\lambda$
if no series solution of~$L_\lambda$ is found.
\end{rem}

\begin{ex}\label{ex:operator-change}
To show an example of calculation of~$L_\lambda$,
let us consider the radix $b \geq 2$ and the linear Mahler operator
$L = \lambda_0 x^\omega - M + M^r$, with
$\lambda_0$ an algebraic number,
$\omega$ a positive integer coprime to~$b-1$,
and $r$ an integer larger than~$1$.
The characteristic polynomial of the leftmost edge of the lower Newton polygon (Definition~\ref{def:lower-char-poly})
has a single root,~$\lambda_0$,
and this edge is the only $\lambda_0$-admissible edge.
The change described by~\eqref{eq:operator-change}
results in the linear operator
$L_{\lambda_0} = \lambda_0(1 - M + \lambda_0^{r - 1} x^{(b^r - b)\omega} M^r)$.
The power series solutions of~$L_{\lambda_0}$ are multiples of
\begin{equation}\label{eq:z-distinct-from-1}
  z = 1 - \lambda_0^{r - 1} x^{(b^r - b)\omega}
        - \lambda_0^{r - 1} x^{(b^r - b)b\omega} + \dotsb,
\end{equation}
from which we deduce a line of solutions of~$L$ in~$\omdr$ generated by
\[
  y = \mylog{\lambda_0} x^{\omega/(b - 1)}\left(1 - \lambda_0^{r - 1} x^{(b^r - b)\omega/(b - 1)} + \dotsb
  \right).
\]
\end{ex}

\subsection{A supplementary remark on the rank of syzygy module}
\label{sec:rank-rels}

Strikingly, the module of syzygies contains but is not limited to
the direct sum over similarity classes
of the module of syzygies contributed by each class.
This section describes the potential interactions
between the latter.
It is not used in the rest of the article.

The following lemma quantifies the syzygies that appear in each similarity class.

\begin{lem}\label{lem:s->2*s-1}
For an $\rat\bL$-similarity class~$\nonz{\fH}$ of\/ $\rat\bL$-hypergeometric elements,
let $(y_1,\dots,y_s)$ denote a basis of the $\bL$-vector space~$\fH$
and~$S$
the column vector $(y_1,\dots,y_s,My_1,\dots,My_s)^T$.
The module of the row vectors~$R \in \bL[x]^{2s}$ such that $RS = 0$ has rank~$2s - 1$.
\end{lem}

\begin{proof}
Making $F = \rat\bL$ in Lemma~\ref{thm:F-sim-F-hyp-is-L-vec} implies that $\fH$~is an $\bL$-vector space.
Since the~$y_i$, $1 \leq i \leq s$, lie in the same $\rat\bL$-similarity class,
we get equalities $y_i = q_i y_1$ for some $q_i$ in~$\rat\bL$ and all~$i$.
As~$y_1$ is $\rat\bL$-hypergeometric, there exists $u_1$ in~$\rat\bL$ such that $My_1 = u_1 y_1$,
from which it results $My_i = (Mq_i) u_1 y_1$.
Therefore, $y_1$, \dots, $y_s$, $My_1$, \dots, $My_s$ are $2s$ elements of the line $\rat\bL y_1$.
The module of their syzygies thus has rank~$2s - 1$.
\end{proof}

However, other syzygies may be produced by the interaction between several similarity classes.
We cannot quantify the phenomenon, and merely give an example.

\begin{ex}\label{ex:puzzling-number-of-relations}
Let us consider, with~$b = 2$, the lclm (least common left multiple)
in~$\bQ[x]\langle M \rangle$
of the operators~$L_1 = (1 - 2x^2) M - (1 - 2x)$ and
$L_2 = (1 - 3x^2) M - (1 - 3x)$, which respectively annihilate the rational functions $y_1 = 1/(1 - 2x)$ and $y_2 = 1/(1 - 3x)$.
This is the operator
\begin{multline*}
    P_1 := \left( 6x^8-5x^4+1 \right) M^2- \left( 6x^6+6x^5+x^4+-x^3-4x^2+x+1 \right) M
\\ \mbox{} +(6x^4+x^3-4x^2+x).
\end{multline*}
We then slightly modify it by truncating its coefficients,
to obtain the operator
\[
  P_2 := \left( -5x^4+1 \right) M^2- \left( x^4-5x^3-4x^2+x+1 \right) M+ \left( x^3-4\,x^2+x \right),
\]
in such a way that $P_2$~has
a $2$-dimensional space of power series solutions, like the operator~$P_1$ we started with.
One of them is the infinite
 product
\[
  y_3 = \prod_{k\geq 0} M^k \frac{1 - 5x^2}{1 - 4x + x^2}
\]
and another linearly independent  one is
\[
  y_4 =
  x+5x^2+19x^3+71x^4+265x^5+983x^6+3667\,
  x^7+13661x^8+O \left( x^9 \right).
\]
At this point, we can observe that the lclm~$L$ of $P_1$ and~$P_2$ admits three right factors of order $1$, namely $M - u_i$, $1 \leq i \leq 3$, with
\[
  u_1 = \frac{1 - 2x}{1 - 2x^2},\qquad
  u_2 = \frac{1 - 3x}{1 - 3x^2},\qquad
  u_3 = \frac{1 - 4x + x^2}{1 - 5x^2}.
\]
As~$y_1$ and~$y_2$ are in the same $\bL(x)$-similarity class, but not~$y_3$,
we could (wrongly) expect a dimension $(2\cdot2-1) + (2\cdot1-1) = 4$
for the rank of the syzygy module
of the column vector $(y_1,\dots,y_4,My_1,\dots,My_4)^T$.
However the operator~$P_2$ factors~as
\[
  P_2 = (1 - 5 x^4) (M - v)(M - u_3),\qquad
  \text{with $ v = x \frac{1 - 5x^2}{1 - 5x^4}$}
\]
and the operator $M - v$ admits the rational solution $z = x/(1 - 5x^2)$.
As a consequence the series~$y_4$, which satisfies $P_2 y_4 = 0$,
also satisfies the equation $(M - u_3) y_4 = c z$ for some constant~$c$.
But~$z$, as a rational function, lies in the same similarity class as~$y_1$ and~$y_2$.
So we have an additional syzygy, between $y_1$, $y_4$, and $My_4$,
and the rank of the syzygy module proves to be~$5$.
\end{ex}

\begin{algo}
\inputs{
  A Riccati Mahler equation~\eqref{eq:riccati} with coefficients in~$\pol\bK$.
}
\outputs{
  The set of ramified rational functions $u \in \ramrat\bKbar$ that solve~\eqref{eq:riccati}.
}
\caption{\label{algo:HP}%
Ramified rational solutions to a Riccati Mahler equation by Hermite--Padé approximants
and prime decompositions.}
\begin{deepenum}[label=(\Alph*)]
\item
      Compute the lower Newton polygon~$\cN$ of~$L := \sum_{k=0}^r \ell_k(x) M^k$,
      then the set $\Lambda$ by~\eqref{eq:def-Lambda}.
\item
      Determine the leftmost edge of~$\cN$ and compute its slope~$-\nu$ and intercept~$\mu$.
\item\label{it:loop-over-lambda}
      For each~$\lambda$ in $\Lambda$,
      \begin{deepenum}[label=(\alph*)]
      \item
            compute the ramification bound~$q_\lambda$ as defined by~\eqref{eq:q-lambda},
      \item
            determine the rightmost $\lambda$-admissible edge of~$\cN$ and compute its slope $-p_\lambda/q_\lambda$ and intercept~$c_\lambda$,
      \item \label{it:HP-sigma-0}
            compute the operator~$L_\lambda$ by~\eqref{eq:operator-change}, the rationals $\nu_\lambda$ and~$\mu_\lambda$ given by~\eqref{eq:nu-lambda-def} in terms of $\mu$ and~$\nu$, and the integer $\sigma_0 := \lfloor\nu_\lambda\rfloor+1$,
      \item \label{it:HP-associated-basis-1}
            compute a basis $(z_1,\ldots, z_\zdim)$ of solutions in~$\fps{\bK[\lambda]}$ to the equation $L_\lambda z = 0$, truncated to order~$O(x^{\sigma_0})$,
      \item\label{it:t=0}
            if $\zdim = 0$, then continue to the next~$\lambda$,
      \item
            compute bounds $\Bnum, \Bden, B_\infty$ by \eqref{eq:bound-P}, \eqref{eq:bound-Q}, and~\eqref{eq:B-infty} applied to~$L_\lambda$,
      \item \label{it:HP-sigma}
            initialize the solution set by $R_\lambda := \varnothing$,
      \item \label{it:HP-iteration}
            for $\sigma := \bigl(\frac{1+\sqrt5}2\bigr)^k \sigma_0$ given by successive~$k=0,1,2,\dots$,
            execute Algorithm~\ref{algo:subtask-of-HP} in the current context,
      \item \label{it:HP-update-result}
            for each~$u$ in $R_\lambda$, change~$u(x)$ into $x^{(b-1)p_\lambda/q_\lambda} u(x^{1/q_\lambda})$.
      \end{deepenum} 
\item
      Return the union of the sets $R_\lambda$ over $\lambda\in \Lambda$.
\end{deepenum} 
\end{algo}

\begin{algo}
\emph{In the context of execution of
step~\ref{it:loop-over-lambda}\ref{it:HP-iteration}
in Algorithm~\ref{algo:HP}.}\\
  \rule[.5\baselineskip]{\textwidth}{.05em}%
\caption{\label{algo:subtask-of-HP}%
Body of step~\ref{it:loop-over-lambda}\ref{it:HP-iteration}
in Algorithm~\ref{algo:HP}.}
\begin{deepenum}[label=(\arabic*)]
\item\label{it:prolong-basis}
      extend the basis elements $z_1,\ldots, z_\zdim$ to solutions truncated to order~$O(x^{\sigma})$ by unrolling recurrences,
\item\label{it:minimal-basis}
      compute a minimal basis of the module of approximate syzygies of $f = (z_1,\ldots, z_\zdim, M z_1,\ldots, M z_\zdim)^T$ to order~$O(x^\sigma)$,
\item\label{it:W}
      build a matrix~$W$ of dimension $\rho \times 2\zdim$ by extracting from the minimal basis the (independent) rows of degree at most~$B_\infty$,
\item\label{it:rho=0}
      if $\rho = 0$, continue to the next~$\lambda$,
\item\label{it:rho=2t-1}
      if $\rho = 2\zdim -1$:
      \begin{deepenum}[label=(\greek*)]
      \item
            for $i=1,\dots,2\zdim$, compute the minor~$\Delta_i$ obtained after removing the $i$th column from~$W$,
      \item\label{it:make-C=u}
            define $\cC := \{u\}$ for the candidate~$u$ provided
            by~\eqref{eq:cand-from-Delta},
            with parameters $a_i$ replaced with~$g_i$,
      \end{deepenum} 
\item\label{it:rho<2t-1}
      if $1 \leq \rho \leq 2\zdim-2$:
      \begin{deepenum}[label=(\greek*)]
      \item
            let $W_{a}$ be the $(\rho+2)\times 2\zdim$ matrix obtained by appending the two-row matrix $(\begin{smallmatrix} a_1 & \dots & a_\zdim & 0 & \dots & 0 \\ 0 & \dots & 0 & a_1 & \dots & a_\zdim \end{smallmatrix})$ below~$W$,
      \item\label{it:absolute-constant-field}
            let~$\ideal{\Sigma}$ be the ideal of $\bKbar[a_1,\dots,a_\zdim]$ generated by the set~$\Sigma \subseteq \bK[\lambda][a_1,\dots,a_\zdim]$ of the coefficients with respect to~$x$ of the $\binom{2\zdim}{\rho+2}$ minors of order $\rho+2$ of~$W_{a}$,
      \item\label{it:prime-decomp}
            compute an irredundant prime decomposition $\bigcap_{j=1}^s\fp_j$ of~$\sqrt{\ideal{\Sigma}}$,
            given for each~$j$ by a Gröbner basis $(p_{j,1},\dots,p_{j,m(j)})$ of~$\fp_j$,
      \item
            initialize the candidate set by $\cC := \varnothing$,
      \item\label{it:loop-over-primes}
            for $j$ from~$1$ to~$s$ do
            \begin{deepenum}[label={[}\roman*{]}]
            \item\label{it:nonlinear-criterion}
                  if for some~$k$, the polynomial~$p_{j,k}$ is nonlinear, continue to the next~$\sigma$,
            \item
              solve the linear system $\{ p_{j,k} = 0 \}$ for the unknowns $a_1,\dots,a_\zdim$, so as to get a full rank matrix~$S$ of dimension~$v\times\zdim$, with ${0 \leq v \leq \zdim}$, and a parametrization $(a_1,\dots,a_\zdim) = (g_1,\dots,g_v)S$,
            \item
                  if~$v=0$, continue to the next~$j$,
            \item\label{it:left-kernel}
                   compute the left kernel of the  matrix~$W_{gS}$,
            \item\label{it:cV>cA}
                  if the kernel has dimension~$2$, or if it has dimension~$1$ and is generated by a row~$K$ with $K_{\rho+1} K_{\rho+2} = 0$, continue to the next~$\sigma$,
            \item
                  add the normalized form of~$u := -K_{\rho+1}/K_{\rho+2}$ to~$\cC$,
            \end{deepenum} 
      \end{deepenum} 
\item\label{it:validation}
      for each~$u = P/Q$ in $\cC \setminus R_\lambda$ do
      \begin{deepenum}[label=(\greek*)]
      \item\label{it:over-bound}
        if $\deg P > \Bnum$ or if $\deg Q > \Bden$, then continue to the next~$\sigma$,
      \item\label{it:failed-div}
        if $P/Q$~cancels the left-hand side of the Riccati equation~\eqref{eq:riccati}, then add~$u$ to~$R_\lambda$, else continue to the next~$\sigma$,
      \end{deepenum} 
\item\label{it:continue-with-next-sigma}
      continue to step~\ref{it:loop-over-lambda}\ref{it:HP-update-result} in Algorithm~\ref{algo:HP},
      thus quitting the loop over~$\sigma$.
\end{deepenum} 
\end{algo}

\section{Algorithms by Hermite--Padé approximation}
\label{sec:hermite-pade}

\subsection{Algorithm by sieving candidates for solving the Riccati equation}
\label{sec:algo-sieve}

By composing the results developed in~\S\ref{sec:syzygies},
we derive Algorithm~\ref{algo:HP}, in which
the body of the loop over the approximation order~$\sigma$
has been isolated as Algorithm~\ref{algo:subtask-of-HP}.
The termination and correctness of the algorithm will be proved
in Theorem~\ref{thm:HP-correct},
after having proved a structural property of the ideals it manipulates
in Theorem~\ref{thm:radical-as-linear-ideals}.
Before this,
we comment on choices and algorithms used at a few specific steps:

\begin{itemize}

\item For each~$\lambda$,
  the algorithm computes a basis of power series solutions
  truncated to an order~$\sigma_0$
  that is sufficient to distinguish the basis elements.
  This order is determined at step~\ref{it:loop-over-lambda}\ref{it:HP-sigma-0}.
  To explain the calculation,
  recall \parencite[Prop.~2.6 and Lemma~2.21]{ChyzakDreyfusDumasMezzarobba-2018-CSL}
  that the computation of a basis of solutions for the equation $L_\lambda z = 0$
  depends on the rational numbers
  \begin{equation}\label{eq:nu-lambda-def}
    \nu_\lambda = q_\lambda  \nu - p_\lambda , \qquad
    \mu_\lambda = q_\lambda (\mu  - c_\lambda) ,
  \end{equation}
  where $-\nu$ and~$\mu$ are respectively
  the slope and intercept of the \emph{leftmost} edge of the lower Newton polygon of~$L$,
  and, as was explained before~\eqref{eq:operator-change},
  $-p_\lambda/q_\lambda$ and~$c_\lambda$ are respectively
  the slope and intercept of the \emph{rightmost} $1$-admissible edge
  of the Newton polygon of~$L_\lambda$.
  The order~$\sigma_0$ is then defined to be~$\lfloor\nu_\lambda\rfloor+1$.

\item Solving for truncated solutions~$z_i$ at step~\ref{it:loop-over-lambda}\ref{it:HP-associated-basis-1}
  is done by solving a linear system of size $(\lfloor\nu_\lambda\rfloor+1) \times (\lfloor\mu_\lambda\rfloor+1)$
  as per \parencite[\S2.6, Algorithm~5]{ChyzakDreyfusDumasMezzarobba-2018-CSL}.

\item The initial value of~$\sigma$
  chosen at step~\ref{it:loop-over-lambda}\ref{it:HP-iteration}
  uses all data already available at this point.
  Next, choosing a geometric sequence to grow~$\sigma$,
  as opposed to, for example, an arithmetic sequence,
  is justified by the superlinear complexity with respect to~$\sigma$
  of the calculations in the loop body.
  The golden ratio used to increase~$\sigma$ is rather arbitrary,
  but reduces a majority of timings for the set of equations tested.

\item Prolonging the truncated solutions~$z_i$ at step~\ref{it:loop-over-lambda}\ref{it:HP-iteration}\ref{it:prolong-basis}
  is done by \parencite[\S2.3, Algorithm~3]{ChyzakDreyfusDumasMezzarobba-2018-CSL}.

\item Computing the minimal basis at step~\ref{it:loop-over-lambda}\ref{it:HP-iteration}\ref{it:minimal-basis}
  can be done by either
  the algorithm in \parencite{Derksen-1994-ACG}
  or the algorithm in \parencite{BeckermannLabahn-1994-UAF}.
  In both cases,
  the theoretical description admits algebraic coefficients in~$\bKbar$.

\item At step~\ref{it:loop-over-lambda}\ref{it:HP-iteration}\ref{it:rho<2t-1}\ref{it:prime-decomp},
  the algorithm uses
  the classical notion of an irredundant primary decomposition of an ideal,
  which applies in particular to radical ideals,
  in which case the primary ideals are even prime
  \parencite[p.~209, especially Theorem~5]{ZariskiSamuel-1958-CA1};
  we will speak of an irredundant prime decomposition in such a situation.
  The computation of such a prime decomposition for radical ideals in~$\bKbar[a_1,\dots,a_\zdim]$
  can be done by a variant of the algorithm in \parencite{GianniTragerZacharias-1988-GBP};
  see also
  \parencite[Algorithms RADICAL, p.~394, and PRIMDEC, p.~396]{BeckerWeispfenning-1993-GB}.

\end{itemize}

The proof of Theorem~\ref{thm:radical-as-linear-ideals} is based on the Nullstellensatz.
Consequently, for the rest of~\S\ref{sec:algo-sieve}
we specialize $\bL$ to~$\bKbar$
in the theory developed in~\S\ref{sec:syzygies}.

\begin{thm}\label{thm:radical-as-linear-ideals}
Consider an operator~$L \in \pol{\bKbar}\langle M \rangle$ and
a basis $z := (z_1,\dots,z_\zdim)$
of its space
of formal power series solutions in~$\fps\bKbar$.
Assume that $L$~has no solutions in~$\puiseux\bKbar$ that are not in~$\fps\bKbar$.
Consider as well
the family of modules\/~$\pol\bKbar\cT\trsigma$ of ranks~$\rho\trsigma$
and the family of characteristic cones $\cA\trsigma$
defined in Definitions \ref{def:truncated-module} and~\ref{def:cone-A-sigma} (see also Proposition~\ref{prop:cone-A-sigma-alt}).
Assume that there exists $\sigma$ satisfying
the hypothesis~$1\leq \rho\trsigma \leq 2\zdim-2$ of~\S\ref{sec:case-ii}
and large enough to have the relation~$\cT\trsigma = \cT\trinfty$.
The system~$\Sigma := \Sigma\trsigma \subseteq \polm{\bKbar} := \bKbar[a_1,\dots,a_\zdim]$
of Definition~\ref{def:Sigma-pol-sys}
generates an ideal~$\ideal{\Sigma}$
whose radical~$\sqrt{\ideal{\Sigma}}$
is a finite intersection of prime ideals,
each given by linear polynomials with coefficients in~$\bKbar$.
\end{thm}

\begin{proof}
Lemma~\ref{lem:direct-sum} proves
the existence of a finite family of ideals~$\fq_i \subseteq \polm{\bKbar}$ generated by $\bKbar$-linear polynomials,
and satisfying
\begin{equation}\label{eq:bar-cA-infty-as-a-union}
\cA\trinfty = \bigcup_{i\in I} V(\fq_i) ,
\end{equation}
where $V({\cdot})$ denotes the variety in~$\bKbar^\zdim$ of an ideal.
Every $\bKbar$-subspace $\fH$ of $\puiseux\bKbar$
associated with a $\ramrat\bKbar$-similarity class~$\nonz{\fH}$
of hypergeometric Puiseux series solutions of~$L$
is, by the assumption on the solutions of~$L$,
in fact included in~$\fps\bKbar$, and therefore
admits a finite basis $h = (h_1,\dots,h_s)$ with elements in $\puiseux{\bKbar}$
that are $\bKbar$-linear combinations of the series~$z_i$.
Therefore,
there is a matrix $H \in \bKbar^{s\times \zdim}$ such that $h^T = Hz^T$.

Each~$V(\fq_i)$ is in bijection with such an~$\fH$
by the $\bKbar$-linear map $a \mapsto az^T$,
and more specifically $V(\fq_i) = \bKbar^s H$.
Consequently, $\fq_i$~is generated by a system of
$\zdim-s$ linearly independent linear polynomials with coefficients in~$\bKbar$.

We now have the successive equalities
\begin{equation*}
V \left(\sqrt{\ideal{\Sigma}}\right) =
V \left(\ideal{\Sigma}\right) =
\cV\trsigma =
\cA\trinfty =
\bigcup_{i\in I} V(\fq_i) =
V \biggl( \bigcap_{i\in I} \fq_i \biggr) ,
\end{equation*}
where the first equality if by \parencite[Theorem~7(ii), p.~183]{CoxLittleOShea-2015-IVA},
the second equality is by definition,
the third equality results by Lemma~\ref{lem:P-Q-nonzero}
from the assumption $\cT\trsigma = \cT\trinfty$,
the fourth equality is~\eqref{eq:bar-cA-infty-as-a-union},
and the fifth equality is by \parencite[Theorem~15, p.~196]{CoxLittleOShea-2015-IVA}.
Retaining the equality between the first and last terms,
then passing to ideals,
Hilbert's Nullstellensatz \parencite[Theorem~2, p.~179]{CoxLittleOShea-2015-IVA} and the definition of a radical
provide the first equality in
\begin{equation*}
\sqrt{\ideal{\Sigma}} =
\sqrt{\bigcap_{i\in I} \fq_i} =
\bigcap_{i\in I} \fq_i ,
\end{equation*}
where, additionally, the second is
because the ideals~$\fq_i$ are prime, therefore radical,
and because an intersection of radical ideals is radical
(\parencite[Theorem~9, p.~147]{ZariskiSamuel-1958-CA1} or
\parencite[Proposition~16, p.~197]{CoxLittleOShea-2015-IVA}).
Finally, any irredundant primary decomposition of~$\sqrt{\ideal{\Sigma}}$
is obtained by retaining a subfamily of the~$\fq_i$,
because the latter are prime.
This ends the proof as the~$\fq_i$ are defined by linear polynomials.
\end{proof}

\begin{thm}\label{thm:HP-correct}
Algorithm~\ref{algo:HP} terminates and correctly computes all solutions of~\eqref{eq:riccati}
in~$\ramrat\bKbar$.
\end{thm}

\begin{proof}
The general structure of the algorithm is
a loop over~$\lambda$ at step~\ref{it:loop-over-lambda},
with independent calculations for different~$\lambda$,
so it is sufficient to prove that for each~$\lambda \in \Lambda$
the calculation terminates
and computes all solutions~$u$ with leading coefficient~$\lambda$.

For each~$\lambda$, the algorithm introduces
the operator~$L_\lambda \in \pol{\bK[\lambda]}\langle M \rangle$
at step~\ref{it:loop-over-lambda}\ref{it:HP-sigma-0}
and a basis $(z_1,\ldots, z_\zdim)$ of its solutions in~$\fps{\bK[\lambda]}$
at step~\ref{it:loop-over-lambda}\ref{it:HP-associated-basis-1}.
Note that by Lemma~\ref{lem:puiseux-denominator-bound},
the dimension of formal power series solutions of~$L_\lambda$ in~$\fps\bKbar$
is also equal to~$\zdim$,
and that the $\bK[\lambda]$-basis~$z$ of solutions in~$\fps{\bK[\lambda]}$
is also a $\bKbar$-basis of the solutions in~$\fps\bKbar$.
It is therefore sufficient to prove that
step~\ref{it:loop-over-lambda}\ref{it:HP-iteration}
computes all rational solutions that can be written~$M(az^T)/(az^T)$ for $a \in \nonz{(\bKbar^\zdim)}$
to get that the algorithm determines
all ramified rational solutions~$u \in \ramrat{\bKbar}$ of~\eqref{eq:riccati}
with leading coefficient~$\lambda$
at step~\ref{it:loop-over-lambda}\ref{it:HP-update-result}.

With respect to termination,
step~\ref{it:loop-over-lambda} may quit early
at step~\ref{it:loop-over-lambda}\ref{it:t=0},
but if the run goes beyond step~\ref{it:loop-over-lambda}\ref{it:t=0},
then the basis~$z$ contains nonzero entries.
By the comment for the case~$\rho\trinfty=2\zdim$ after Lemma~\ref{lem:direct-sum},
this forces $\rho\trinfty < 2\zdim$
and the algorithm continues
with an unbounded inner loop over~$\sigma$
at step~\ref{it:loop-over-lambda}\ref{it:HP-iteration}.
This inner loop can quit early
at step~\ref{it:loop-over-lambda}\ref{it:HP-iteration}\ref{it:rho=0},
ending the calculation for the current~$\lambda$.
The inner loop can only be relaunched
from step~\ref{it:loop-over-lambda}\ref{it:HP-iteration}\ref{it:rho<2t-1}\ref{it:loop-over-primes}\ref{it:nonlinear-criterion} if a nonlinear polynomial if detected,
from step~\ref{it:loop-over-lambda}\ref{it:HP-iteration}\ref{it:rho<2t-1}\ref{it:loop-over-primes}\ref{it:cV>cA} if a nonzero rational function cannot be defined,
and from steps \ref{it:loop-over-lambda}\ref{it:HP-iteration}\ref{it:validation}\ref{it:over-bound} and~\ref{it:failed-div}
if a false solution is detected.

At the construction of~$z$
at step~\ref{it:loop-over-lambda}\ref{it:HP-associated-basis-1},
no~$z_i$ is in~$O(x^{\sigma_0})$,
and therefore no~$z_i$ is ever in~$O(x^\sigma)$
after the extension step~\ref{it:loop-over-lambda}\ref{it:HP-iteration}\ref{it:prolong-basis}.
Consequently, $\rho$~is never~$2\zdim$ in the inner loop,
which is why only the cases $0 \leq \rho \leq 2\zdim-1$ are considered in the algorithm.

The objects computed inside the loop
at step~\ref{it:loop-over-lambda}\ref{it:HP-iteration},
most of which have coefficients in~$\bK[\lambda]$,
relate to the theoretical objects
defined in~\S\ref{sec:syzygies}
for~$\bL = \bKbar$.
We already proved that $z$~is a basis for the solutions of~$L_\lambda$ in~$\fps{\bKbar}$.
At step~\ref{it:loop-over-lambda}\ref{it:HP-iteration}\ref{it:minimal-basis},
the computation of the minimal basis,
by Derksen's algorithm~\parencite*{Derksen-1994-ACG},
Beckermann and Labahn's algorithm~\parencite{BeckermannLabahn-1994-UAF}
or other known algorithms,
is independent of any extension of~$\bK[\lambda]$ in which
the coefficients of the input would be seen.
Consequently,
the matrix so obtained provides a minimal basis of the approximate syzygy module
both over $\pol{\bK[\lambda]}$ and over~$\pol{\bKbar}$.
By the second interpretation, the matrix~$W$ computed
at step~\ref{it:loop-over-lambda}\ref{it:HP-iteration}\ref{it:W}
is the matrix~$W\trsigma$ of Definition~\ref{def:truncated-module} for~$\bL = \bKbar$,
and its height, $\rho$~in the algorithm description,
is the rank~$\rho\trsigma$ of the theory, also for~$\bL = \bKbar$.
For the rest of the proof,
we consider the objects
$\cT\trsigma$ (vectorial truncation of the approximate syzygy module),
$\cA\trsigma$ (characteristic cone),
$\cV\trsigma$ (relaxed cone),
$\Sigma\trsigma$,
which are all obtained only from~$W$,
and the theory applied to~$\bL = \bKbar$
(see Definitions \ref{def:truncated-module}, \ref{def:cone-A-sigma}, and~\ref{def:Sigma-pol-sys}, and Proposition~\ref{prop:cone-A-sigma-alt}).

Suppose that the algorithm runs without exiting the inner loop,
thus making $\sigma$~grow indefinitely.
In particular, $\rho$~is never~$0$, which would cause early quitting,
so we indefinitely have $1 \leq \rho \leq 2\zdim-1$.
From some point on, we have ${\cT\trsigma = \cT\trinfty}$,
so that Lemma~\ref{lem:P-Q-nonzero} applies,
and $\cV\trsigma = \cA\trsigma = \cA\trinfty$ holds.
If the calculation enters step~\ref{it:loop-over-lambda}\ref{it:HP-iteration}\ref{it:rho=2t-1},
it continues to
step~\ref{it:loop-over-lambda}\ref{it:HP-iteration}\ref{it:validation}.
Otherwise, the calculation enters step~\ref{it:loop-over-lambda}\ref{it:HP-iteration}\ref{it:rho<2t-1}.
Then $\rho \leq 2\zdim-2$ and
the ideal $\sqrt{\ideal{\Sigma\trsigma}}$
computed at step~\ref{it:loop-over-lambda}\ref{it:HP-iteration}\ref{it:rho<2t-1}\ref{it:prime-decomp}
is the ideal of the variety~$\cV\trsigma$.
At this point, it is legitimate to apply Theorem~\ref{thm:radical-as-linear-ideals}
to the operator~$L_\lambda$,
which obviously has coefficients in~$\bKbar$,
all other hypotheses being fulfilled in terms of objects over~$\bL = \bKbar$
(see in particular Remark~\ref{rem:sols-puisseux-dans-fps}).
We conclude that the prime ideals~$\fp_j$ have linear generators,
so that the algorithm passes beyond
step~\ref{it:loop-over-lambda}\ref{it:HP-iteration}\ref{it:rho<2t-1}\ref{it:loop-over-primes}\ref{it:nonlinear-criterion}.
By Lemma~\ref{lem:spec-rk-is-rho+1}, the kernel computed
at step~\ref{it:loop-over-lambda}\ref{it:HP-iteration}\ref{it:rho<2t-1}\ref{it:loop-over-primes}\ref{it:left-kernel}
has dimension~$1$,
and by Lemma~\ref{lem:P-Q-nonzero}(\ref{it:P-Q-nonzero}), all~$K$ computed
satisfy $K_{\rho+1} K_{\rho+2} \neq 0$,
so that the algorithm also passes beyond
step~\ref{it:loop-over-lambda}\ref{it:HP-iteration}\ref{it:rho<2t-1}\ref{it:loop-over-primes}\ref{it:cV>cA}.
In this case again,
the calculation continues to
step~\ref{it:loop-over-lambda}\ref{it:HP-iteration}\ref{it:validation}.
At this point,
because $\cV\trsigma = \cA\trinfty$ can only describe true solutions,
owing to Lemma~\ref{lem:direct-sum},
the validation step~\ref{it:loop-over-lambda}\ref{it:HP-iteration}\ref{it:validation}
finds no false solution and the loop ends, a contradiction.

We have just seen that the loop over~$\sigma$ terminates,
and that it does so after proving that all candidates are true solutions.
Consider the final value of~$\sigma$
and exclude the trivial case of early termination
at step~\ref{it:loop-over-lambda}\ref{it:HP-iteration}\ref{it:rho=0},
that is, assume~$\rho\trsigma > 0$.
In the inner loop body, the algorithm runs
either step~\ref{it:loop-over-lambda}\ref{it:HP-iteration}\ref{it:rho=2t-1}
or~\ref{it:rho<2t-1}.
In the first case, $\cV\trsigma = \bKbar^\zdim$,
which, in view of Lemma~\ref{lem:formula-when-rho=2t-1},
leads to a rational candidate~$u$
parametrized by the free parameter~$(g_1,\dots,g_\zdim)$
at step~\ref{it:loop-over-lambda}\ref{it:HP-iteration}\ref{it:rho=2t-1}\ref{it:make-C=u}.
In the second case,
the prime decomposition computed
at step~\ref{it:loop-over-lambda}\ref{it:HP-iteration}\ref{it:rho<2t-1}\ref{it:prime-decomp}
represents~$\sqrt{\ideal{\Sigma\trsigma}}$,
so the union of the varieties of the~$\fp_j$ in~$\bKbar^\zdim$ is the relaxed cone~$\cV\trsigma$,
and the rational candidates are the fractions~$P/Q$ implied by Lemma~\ref{lem:when-rho<2t-1}.
By our hypothesis of a terminating~$\sigma$,
all the candidate rational solutions obtained previously
are verified to be true solutions
by step~\ref{it:loop-over-lambda}\ref{it:HP-iteration}\ref{it:validation},
that is $\cV\trsigma \subseteq \cA\trinfty$.

Because of the inclusions $\cA\trinfty \subseteq \cA\trsigma \subseteq \cV\trsigma$
that are valid independently of~$\sigma$,
we have thus proved $\cV\trsigma = \cA\trinfty$,
that is, step~\ref{it:loop-over-lambda}\ref{it:HP-iteration}\ref{it:continue-with-next-sigma}
returns exactly the set of all solutions~$u$ with leading coefficient~$\lambda$.
\end{proof}

\begin{rem}
Although solving by computing a prime decomposition
may exhibit a bad worst-case behavior in theory,
we have not investigated how to improve the theoretical complexity of the step
after observing experimentally that this is never the bottleneck of execution:
the computation of Hermite--Padé approximants always takes more time.
In addition, the special structure of a union of spaces
makes it plausible that one could develop a better method than
the general primary decomposition algorithm.
\end{rem}

\begin{ex}\label{ex:iterations-number-hermite-pade} 

Let us return to Example~\ref{ex:operator-change}
to show that the number of steps before obtaining the final value of~$\sigma$
can be made arbitrarily large
in Algorithm~\ref{algo:HP}.
The algorithm first determines the parameters
$\nu = \omega/(b-1)$ and~$\mu = \omega b/(b-1)$.
For~$\lambda = \lambda_0$,
it then gets the ramification bound~$q_{\lambda_0} = b - 1$,
and the rightmost $\lambda_0$-admissible edge (see \eqref{eq:q-lambda} and~Definition~\ref{def:lower-char-poly}) happens to be the leftmost,
with slope~$-\nu$, thus providing~$p_{\lambda_0} = \omega$
and $c_{\lambda_0} = \mu$.
This implies $\nu_{\lambda_0} = \mu_{\lambda_0} = 0$.
The order~$\sigma = \sigma_0$ defined at step~\ref{it:HP-sigma-0}
for the series expansion at the first iteration of the loop~\ref{it:HP-iteration}
is~$\omega + 1$,
and solving~$L_{\lambda_0}$ yields $\zdim = 1$
at step~\ref{it:HP-associated-basis-1},
with a basis~$z$ of the form~$z = (z_1)$
for~$z_1 = 1 + O(x^{\sigma_0})$.
The algorithm verifies that the candidate obtained from~$z_1 = 1$
is a false solution.
In view of~\eqref{eq:z-distinct-from-1},
to distinguish~$z_1$ from~$1$, the algorithm will need to increase the order
to approximately~$(b^r-b)\omega$.
For the sake of exposition,
consider a variant setting of the algorithm
that multiplies~$\sigma$ by~$b$ at each iteration
(instead of the golden ratio),
so that $\sigma = b^k \sigma_0$ (starting with~$k=0$).
The algorithm then needs to reach $k = r$ before $z$~can be distinguished from~$1$.
For~$k = r$, if $\sigma - e > \Bnum$
for the exponent $e = (b^r - b)\omega$ appearing in~\eqref{eq:z-distinct-from-1},
then the rank~$\rho\trsigma$ is zero and the algorithm immediately stops;
if $\sigma - e \geq \Bnum$, the algorithm continues to increase~$k$
until it decides to stop.
In all cases, the final value of~$\sigma$ is at least~$b^r (\omega + 1)$.
For a numerical example, set $b = 4$, $L = x^{10} - M + M^4$.
It takes five iterations and computation with accuracy~$O(x^{2816})$
to conclude that the Riccati equation has no solution:
this is a case where the algorithm stops at~$k = r$,
with $2816 = 4^4(10+1) > \Bnum = 157$
(for the bound~$\Bnum$ relative to~$L_{\lambda_0}$).
\end{ex}

\subsection{Rational solving of the linear Mahler equation}
\label{sec:back-to-rat}

In~\parencite{ChyzakDreyfusDumasMezzarobba-2018-CSL},
we developed an algorithm to compute the rational solutions
of the linear Mahler equation~\eqref{eq:linear}.
There, we derived degree bounds for the numerators and denominators
of rational solutions
(see Propositions 3.16, 3.17, and~3.21 in that reference,
or, for alternative bounds, \parencite{BellCoons-2017-TTM}).
Here%
\footnote{We are indebted to Alin Bostan for asking us a question that led to the present section.}
we sketch how
the approach by Hermite--Padé approximants of the present~\S\ref{sec:syzygies} adapts,
based on these bounds,
to give an alternative algorithm for computing rational solutions
in the linear case.
This adaptation is not necessary for the rest of the article.

To study the linear equation, we modify
the definition of the characteristic cone~$\cA\trsigma$
and that of the augmented matrix~$W_{a}\trsigma$
(see Definition~\ref{def:W-a} and Proposition~\ref{prop:cone-A-sigma-alt}).
We start with the vector of series~$(z_1,\dots,z_\zdim,1)^T = (z,1)^T$
and, for any given~$\sigma$, we compute a minimal basis
from which we extract the rows compatible with the degree bounds.
This yields a $\rho\trsigma \times (\zdim+1)$-matrix~$W\trsigma$ satisfying
\[ W\trsigma (z,1)^T = O(x^\sigma) . \]
We define $W_{a}\trsigma$~as the matrix~$W\trsigma$ augmented at its bottom
with the two-row matrix
\begin{equation}\label{eq:two-row-matrix-ratmod}
\begin{pmatrix}
a & 0 \\
0 & 1
\end{pmatrix} =
\begin{pmatrix}
a_1 & \dots & a_\zdim & 0 \\
0 & \dots & 0 & 1
\end{pmatrix} ,
\end{equation}
and we let
\begin{multline}\label{eq:def-cone-A-sigma-ratmod}
\cA\trsigma = \{ 0 \} \cup \{ a \in \nonz{(\bL^\zdim)} \mid
  \exists P \in \nonz{\pol\bL}, \
  \exists Q \in \nonz{\pol\bL}, \
  \exists \leftkern \in \pol\bL^{\rho\trsigma}, \\
  (\leftkern_1,\dots,\leftkern_{\rho\trsigma},P,-Q)W_{a}\trsigma = 0
\} .
\end{multline}
The motivation is that for~$a \in \cA\trsigma$ and corresponding~$(\leftkern,P,Q)$,
if $W\trsigma (z,1)^T$ is zero, we get
$(\leftkern_1,\dots,\leftkern_{\rho\trsigma},P,-Q)W_{a}\trsigma (z,1)^T = 0$,
implying $P a z^T - Q = 0$, that is, $Q/P = a z^T$.

In this context, the rank of the generating matrix~$W\trsigma$ satisfies~$\rho\trsigma \leq \zdim$.
Again, we distinguish two cases:
\begin{itemize}
\item If~$\rho\trsigma = \zdim$, we expect, provided $\sigma$ is large enough, that
\eqref{eq:linear}~admits a full basis of rational solutions.
We proceed as in~\S\ref{sec:case-i} to prove~${\cA\trsigma = \bL^\zdim}$
and to identify a candidate rational vector solution~$K^T \in \rat{\bL}^{\zdim+1}$.
If~$W\trsigma (z,1)^T = 0$, then $K$
must be proportional to~$(z,1)^T$ in~$\fls{\bL}^{\zdim+1}$,
hence each series~$z_i$ is equal to the rational series~$K_i/K_{\zdim+1}$.
We therefore consider a parametrized candidate in the form
$Q/P = \sum_{i=1}^\zdim g_i (K_i/K_{\zdim+1})$ for the parameter~$a = g$
in~$\bL^\zdim$.
\item Else, $\rho\trsigma \leq \zdim-1$, and we proceed as in~\S\ref{sec:case-ii}.
The characteristic cone~$\cA\trsigma$ is given by the vanishing
of all minors of rank~$\rho\trsigma+2$ in~$W_{a}\trsigma$.
(These are the minors of rank~$\rho\trsigma+1$ in the submatrix obtained
by removing the last row and last column.)
Extracting coefficients with regard to~$x$ yields a linear system in~$a$,
so the cone~$\cA\trsigma$ is a vector space
that can be parametrized in the form~$a = gS$.
Solving for a left kernel after this specialization delivers a basis
of rows of the form $(\leftkern_1,\dots,\leftkern_{\rho\trsigma},P,-Q)$,
each parametrized by~$g \in \bL^v$.
If $W\trsigma (z,1)^T = 0$, then for each row $P a z^T = Q$, so $P$ and~$Q$
can be zero only simultaneously:
but in this case, the rank of~$W\trsigma$ makes $\leftkern$~be zero,
so the kernel has dimension~$1$.
(In an actual calculation,
a failure of the computed kernel to have dimension~$1$
proves that $\sigma$~must be increased.)
From this, we obtain a parametrized candidate solution~$Q/P$.
Cramer's rules applied to
a selection of $\rho\trsigma+1$ columns of~$W_{a}\trsigma$ including the last one
show that $P$~is independent of~$a$ and $Q$~is linear in~$a$.
\end{itemize}
We continue as in the Riccati case:
all candidates have to be checked against degree bounds
and exact evaluation to~$0$ of the linear Mahler equation,
and false solutions require to increase~$\sigma$.

Experimentally, the algorithm above seems to behave better
than our Algorithm~9 in~\parencite{ChyzakDreyfusDumasMezzarobba-2018-CSL}.
It would thus be of interest to analyze its complexity.

\part{Implementation and application}
\label{part:implem}

\section{Benchmark}
\label{sec:impl-and-benchmark}

\subsection{Implementation}

To test the examples listed in~\S\ref{sec:examples},
we used Dumas's package~\verb+dcfun+\footnote{%
available from \url{https://mathexp.eu/dumas/dcfun/}}.
This contains an implementation of
Algorithms \ref{algo:BP}, \ref{algo:IP}, and \ref{algo:HP},
and the needed parts of~\parencite{ChyzakDreyfusDumasMezzarobba-2018-CSL}
in the computer-algebra system Maple.
It calls Singular's routine for primary decomposition (see below).

The implementation of Algorithms \ref{algo:BP} and~\ref{algo:IP}
corresponds to their specification,
but the implementation of Algorithm~\ref{algo:HP} has limitations,
which have however no impact on the validity
of the treatment of examples in~\S\ref{sec:discussion},
as we now explain.

Although algorithms for computing Hermite--Padé approximants
allow algebraic coefficients,
no implementation was available in Maple
beyond rational number coefficients.
For step~\ref{it:loop-over-lambda}\ref{it:HP-iteration}\ref{it:minimal-basis}
in Algorithm~\ref{algo:HP},
we have therefore used
the Maple command
\verb+MahlerSystem+ in the package \verb+MatrixPolynomialAlgebra+
that is restricted to coefficients in~$\bQ$,
thus forcing~$\bK = \bQ$,
and limiting the search for rational solutions of~\eqref{eq:riccati}
to series in~$\ramrat{\bar\bQ}$ with leading coefficient~$\lambda$ in~$\bQ$.
In principle, this allows the implementation  to find solutions
corresponding to~$\lambda \in \Lambda \cap \bQ$ (recall~\eqref{eq:def-Lambda}),
but in practice, all of our examples have $\Lambda \subseteq \bQ$.

Although algorithms for computing primary decompositions allow
the search for decompositions over the algebraic closure of~$\bK$,
implicitly making algebraic extensions as needed along their process,
Maple's command \verb+PrimeDecomposition+ in the package \verb+PolynomialIdeals+
limits the algebraic numbers used to an algebraic field
specified as part of the input.
To the best of our knowledge, the only general implementation
is available in the computer-algebra system Singular,
as the command \verb+absPrimdecGTZ+.
The Maple implementation of Algorithm~\ref{algo:HP} transparently
calls Singular at its step~\ref{it:loop-over-lambda}\ref{it:HP-iteration}\ref{it:rho<2t-1}\ref{it:prime-decomp}.

\subsection{Examples}
\label{sec:examples}

To validate and exemplify our theory,
we propose two kinds of examples.
The first consists of generating functions of automatic sequences or more generally of sequences satisfying a linear recurrence with constant coefficients of the divide-and-conquer type.
The second kind is provided by  lclms (least common left multiples) of operators vanishing on simple expressions,
like rational or power functions.
In the subsequent tables and discussions,
all examples are referred to by self-explanatory pseudonyms, like \verb+Baum_Sweet+ for the equation appearing in Example~\ref{ex:classical}(\ref{item:baum-sweet}) and \texttt{Rudin\_Shapiro} for the first one below.

\subsubsection{Generating functions of automatic sequences and like sequences}
\label{sec:gen-fun}

\,\!\!

\begin{ex}[]
\label{ex:rudin-shapiro}
The Rudin--Shapiro sequence $(a_n)_{n\in \bN}$ is the automatic sequence defined by
$a_n = (-1)^{e_n}$ where $e_n$ is the number of (possibly overlapping) blocks~$11$ in the binary
representation of $n$ \parencite[A020985]{oeis}.
It is characterized by the recurrence relations
\[
a_0 = 1,\quad a_{2n} = a_{n} ,\quad a_{2n+1} = (-1)^{n}a_n ,
\]
hence its generating function satisfies
the Mahler equation
\begin{equation}\tag{\text{\texttt{Rudin\_Shapiro}}}
2xM^2 y - (x - 1)M y - y = 0 \qquad (b=2).
\end{equation}
\end{ex}

\begin{ex}[]
\label{ex:stern-brocot}
Let us consider again the Stern--Brocot sequence that was defined in Example~\ref{ex:classical}(\ref{item:stern-brocot}).
We will re-obtain the well-known fact that
its generating function
is the Mahler hypergeometric function
\[
  y(x) = x \prod_{k\geq 0} (1 + x^{2^k} + x^{2^{k+1}}) = \sum_{n\geq1}a_nx^n
  \in \bZ[[x]] ,
\]
which is obviously a solution of $L_2 y(x) = x y(x) - (1 + x + x^2) y(x^2) = 0$
with $b = 2$.

To this end, we follow
a method implicit in \parencite{ChristolKamaeMendesFranceRauzy-1980-SAA}
that was detailed in \parencite{Allouche-1987-AFT}.
Write $y_1(x) = y(x)$ and introduce $y_2(x) = \sum_{n\in\bN}a_{2n+1}x^n$,
to obtain
\begin{equation*}
\vecty(x) = A(x) \vecty(x^2)
\qquad\text{for}\qquad
\vecty(x) =
\begin{pmatrix} y_1(x) \\ y_2(x) \end{pmatrix} ,
\quad
A(x) =
\begin{pmatrix}
1 & x \\
1-x & 1+2x
\end{pmatrix} .
\end{equation*}
After introducing the Mahler operator for the radix $b = 2$, we get
$\vecty = A (M \vecty)$, next $Y = A (M A) (M^2 \vecty)$,
hence, after setting $R = (1, 0)$,
\begin{equation*}
y_1 = R A (MA) (M^2 \vecty) , \quad
M y_1 = R (MA) (M^2 \vecty) , \quad
M^2 y_1 = R (M^2 \vecty) .
\end{equation*}
Because the row vectors $R A (M A)$, $R (M A)$, $R$ are linearly dependent over~$\rat\bQbar$,
there is an operator~$L'_2$ of order~$2$ canceling~$y_1$.
Performing the procedure using Maple, we readily obtain it in explicit form,
so that the generating series~$y(x)$ is a solution of
\begin{multline}\tag{\text{\texttt{Stern\_Brocot\_b2}}}
{} \\
L'_2 y(x) = x y(x) - (1 + x + 2x^2) y(x^2) + (1 + x^2 + x^4) y(x^4) = 0 .
\end{multline}
Both Algorithm~\ref{algo:IP} and Algorithm~\ref{algo:HP} compute
the hypergeometric solutions of~$L'_2$,
to prove that the only ones are multiples of~$y(x)$.
Correspondingly, one can also verify the relation $L'_2 = (1 - M) L_2$.

Because $y(x)$~is as solution of a linear Mahler equation for radix~$b = 2$,
for any~$k \geq 2$,
it is also a solution of a linear Mahler equation for radix~$b^k$
that one can make explicit.
This gives rise to related operators:
an annihilating operator of the generating function~$y(x)$, of order~$2$ for radix~$b = 4$, is
\begin{multline}\tag{\text{\texttt{Stern\_Brocot\_b4}}}
  L'_4 =
  (1+x +2 x^{2})
  (1+x +x^{2})^2
  (1-x +x^{2})^2
  (1-x^{2}+x^{4})^2
  (1-x^{4}+x^{8})
  M^2\\
  - c_1(x) M + x^3 (1 + x^4 + 2 x^8) \qquad (b=4)
\end{multline}
with
\begin{multline*}
  c_1(x) =  \sum_{n=0}^{14} a_{n+1} x^n  =
   1+x +2 x^{2}+x^{3}+3 x^{4}+2 x^{5}+3 x^{6}
   \\
   \mbox{}+x^{7}+4 x^{8}+3 x^{9}+5 x^{10}+2 x^{11}+5 x^{12}+3 x^{13}+4 x^{14}.
\end{multline*}
As well, $y(x)$~is Mahler hypergeometric with respect to radix~$b = 4$,
as reflected by the factorization $L'_4 = ((2x^2+x+1) M - (2x^8+x^4+1)) L_4$ for
\begin{equation*}
L_4 = (x^2+x+1)(x^4+x^2+1) M - x^3 \qquad (b = 4).
\end{equation*}

\end{ex}

\begin{ex}[Missing digit in ternary expansion]
The sequence
0, 1, 3, 4,~\dots
of nonnegative integers
whose ternary expansion does not contain the digit~2
\parencite[A005836]{oeis} has a generating function
$y(x) = x + 3x^2 + 4x^3 + 9x^4 + 10x^5 + \dotsb$
annihilated by the operator
\begin{equation}\tag{\text{\texttt{no\_2s\_in\_3\_exp}}}
  L = x - (1 + 3x + 4x^2) M + 3(1 +x^2)^2M^2\qquad(b = 2).
\end{equation}
Using either of our algorithms, we find that
$L$~admits the unique right-hand factor $M -  1/(3(1+x))$,
corresponding to hypergeometric series solutions of~$L$ that are
scalar multiples of $\mylogg{\log_2(1/3)} / (1-x)$.
This proves that the generating series~$y(x)$ is not hypergeometric.
\end{ex}

\begin{ex}[]
\label{ex:dilcher-stolarsky}
Inspired by \parencite[Prop.~5.1]{DilcherStolarsky-2007-PAS},
we consider the formal power series $F(x) \in \bZ[[x]]$ that is
a solution of
\begin{equation}\tag{\text{\texttt{Dilcher\_Stolarsky}}}
x^4 M^2y(x) - (1 + x + x^2) My(x) + y(x) = 0 \qquad (b=4), \qquad y(0) = 1 .
\end{equation}
\end{ex}

\begin{ex}[]
\textcite[\S3.1]{KatzLinden-2022-PSL} define a sequence $(A_t(x))_{t\in\bN}$ whose generating function
$y(w, x) = \sum_{t\geq 0} A_t(x) w^t$
satisfies an order~$4$ and degree~$14$ equation
$Ly = 0$ with respect to~$x$ with $b = 2$ and
\begin{equation}\tag{\text{\texttt{Katz\_Linden}}}
\begin{aligned}
L &= {} - x (x +1) (8 x^4 w^2+4 x^2 w -x^2+2 w -1) \\
&\phantom{{}={}} {} - x (8 w^3 x^7-8 w^3 x^6+8 w^3 x^5+8 w^3 x^4-4 w^2 x^5-4 x^4 w^2-w x^5 \\
& \qquad\qquad {} - 2 w^2 x^3-3 w x^4+2 w^2 x^2-4 w x^3+x^4+2 w^2 x -4 x^2 w \\
& \qquad\qquad {} + x^3+2 w^2-3 w x +x^2-w +x ) M \\
&\phantom{{}={}} {} + x^2 w (16 w^3 x^8+8 w^2 x^8+32 w^3 x^5+4 w x^7+16 w^3 x^4-x^7+12 x^4 w^2 \\
& \qquad\qquad {} - 4 w x^5-x^6+16 w^2 x^3-4 w x^4+x^5+16 w^2 x^2+x^4+8 w^2 x -x^3 \\
& \qquad\qquad {} + 4 w^2-x^2-x -1) M^2 \\
&\phantom{{}={}} {} + 2 x^4 w^2 (-8 w^2 x^{10}-4 w x^9+32 w^3 x^6-2 w x^8+x^9+x^8-4 w x^4-8 w x^3 \\
&\phantom{{}={}} {} - 8 x^2 w +2 x^3-4 w x +2 x^2-2 w +x +1) M^3 \\
&\phantom{{}={}} {} - 8 x^{12} w^3 (8 w^2 x^2+4 w x +2 w -x -1) M^4 .
\end{aligned}
\end{equation}

The theory we developed in the previous sections
made the hypothesis that
$\bK$~is a computable subfield of~$\bC$
mostly in order to reuse results from our previous article.
However,
the theory easily adapts to equations that depend rationally on auxiliary parameters,
and our implementation is able to deal with this example as well.

\end{ex}

\begin{ex}[Parities and ternary expansion]
\label{ex:adamczewski-faverjon}

The operator of order~$4$ and degree~$258$ of Example~\ref{ex:adamczewski-faverjon-intro}
has been our running example in Examples
\ref{ex:adamczewski-faverjon-case-no},
\ref{ex:adamczewski-faverjon-case-i},
\ref{ex:adamczewski-faverjon-case-ii},
and~\ref{ex:adamczewski-faverjon-case-conclusion}.
We call it \texttt{Adamczewski\_Faverjon} in our tables.
In roughly one second,  Algorithm~\ref{algo:HP}
finds that the only rational solutions to the Riccati Mahler equation
are
\begin{equation*}
\frac1{1-x-x^2} , \qquad
\frac1{1+x-x^2} , \qquad
\frac{g_1+g_2x^3}{g_1+g_2x}\frac1{1+x^2+x^4} .
\end{equation*}
(Algorithm~\ref{algo:IP} has an equivalent output in over $70$~minutes.)
This corresponds to the hypergeometric solutions of~$L$ that were given as~\eqref{eq:adamczewski-faverjon-hyp-sols}
to deduce that none of the~$y_i$ of Example~\ref{ex:classical}(\ref{item:adamczewski-faverjon}) is hypergeometric.
\end{ex}

\subsubsection{Operators related to criteria of differential transcendence}

The criteria to be discussed in~\S\ref{sec:transcendence}
will lead us to solve the auxiliary Riccati equation~\eqref{eq:r2}
beside the equations~\eqref{eq:riccati}, a.k.a.~\eqref{eq:r1},
for operators of order~${r = 2}$.
We thus introduced problems named \verb+dft_+$\langle p\rangle$
where $p$~ranges over relevant problems already listed in~\S\ref{sec:gen-fun}.

\subsubsection{Cooked-up examples}

\begin{ex}[Annihilator related to rational functions]
\label{ex:rational-function-annihilator}

We produced two examples as lclms of three first-order operators in~$\pol\bQ\langle M\rangle$,
thus forcing the Riccati equation to have known rational solutions.
In both cases, the Riccati equation has an isolated rational solution
plus a family parametrized by a projective line.
However,
by adjusting a coefficient, we forced the number~$m$ of distinct~$\lambda$
in the logarithmic parts of solutions of~\eqref{eq:linear}
to be either $1$ or~$2$.
These examples are named \verb+lclm_+\texttt{3}\verb+rat_+$\langle m\rangle$\verb+log+.

For more involved examples,
we started with a first operator~$L_1$
defined as the lclm of $q$ first-order operators,
which we then tweaked into an operator~$L_2$
by discarding monomials above some line of a given slope~$s$
in the lower Newton polygon of~$L_1$ (Definition~\ref{def:lower-newton}).
In this way, we predict the same dimensions of series solutions,
but we expect to lose their hypergeometric nature.
After taking the lclm of $L_1$ and~$L_2$,
we obtain the operators for our examples
\verb+lclm_+$\langle q\rangle$\verb+rat_trunc_sl+$\langle s\rangle$.
In particular,
the operator in Example~\ref{ex:puzzling-number-of-relations}
is the one for example \verb+lclm_2rat_trunc_sl1+ in the tables.
The corresponding~$L_1$ is of order~$2$ with all its solutions rational.
\end{ex}

\begin{ex}[Annihilator of power functions]
\label{ex:roots-annihilator}
For various pairs~$(b,q)$ given by
a radix~$b \leq 5$ and some positive integer~$q \leq 5$,
we considered some operator~$L_{b,q} \in \pol\bQ\langle M\rangle$ annihilating
the $q$~powers $x^{1/1}, x^{1/2}, \dots, x^{1/q}$.
The corresponding example is called \verb+lclm_+$\langle q\rangle$\verb+pow_b+$\langle b\rangle$
in the tables.
To obtain the operator~$L_{b,q}$,
we considered the lclm of binomial annihilators for the relevant~$x^{1/i}$,
which, to ensure no ramification in these operators,
 need not be of order~$1$.
As an example, for $(b, q) = (3, 4)$, we used the annihilators~$B_i$ of $x^{1/i}$
given as
\[
  B_1 = M - x^2,\quad
  B_2 = M - x,\quad
  B_3 = M^2 - x^2 M,\quad
  B_4 = M^2 - x^2 .
\]
Their lclm~$L_{3,4}$ has order $r = 6$ and degree $d = 727$.
\end{ex}

\begin{ex}[Random equations]\label{ex:random-eqns}
To test the robustness of the Hermite--Padé approach,
we considered random operators constructed as follows.
Given a radix $b \in \{2,3\}$ and some degree parameter~$\delta$,
we first draw random $A$, $B$, and~$C$,
each of the form $VM - U$
for dense polynomials $U$ and~$V$ of degree~$\delta$ in~$x$
having integer coefficients in the range~$[-1000,1000]$.
Let $\delta'$ be the degree in~$x$ of $C' := \operatorname{lclm}(A,B)$,
normalized to have no denominator.
Then, the operator $\tilde C := C' + x^{\delta'}(M^2+M+1)$ has the same dimension of series solutions
and series solutions with the same possible valuations
as~$C'$,
but no more hypergeometric solutions.
Then, our random operator is chosen to be $L := \operatorname{lclm}(C, \tilde C)$.
It has a hypergeometric solution,
and the Riccati Mahler equation admits
the rational function~$U/V$ corresponding to~$C$
as a rational solution.
These are the examples \verb+rmo_+$b$\verb+_+$\delta$ in the tables,
for $b = 2,3$ and~$\delta=1,\dots,5$.
\end{ex}

\subsection{Discussion of the timings}
\label{sec:discussion}

We have executed our algorithms on the operators of~\S\ref{sec:examples}
on a Dell Precision Mobile 7550 with i9~processor and 64~GB of RAM,
running under an up-to-date Archlinux system.
For reproducibility of the timings, we have put the processor in a state
where thermal status and number of concurrent jobs has no influence,%
\footnote{Using the “performance cpufreq governor” and
forbidding any adaptive cpu overclocking, a.k.a.~cpu “turbo mode”
avoids variations of timings up to a factor of~$2$.}
effectively fixing the cpu frequency to 2.4~GHz,
and run all examples one after another.
We killed any example above either 12~hours of calculation
or 60~GB of used memory.

\begin{table}
\begin{small}
\centerline{%
\begin{tabular}{|@{\,}r@{\,}|@{\,}r@{\,}r@{\,}r@{\,}|@{\,}r@{\,}r@{\,}r@{\,}r@{\,}r@{\,}r@{\,}r@{\,}r@{\,}|}
\hline
example & $b$ & $r$ & $d$ & var & tpl & anc & cfs & pol & $\myhash$ & $\check\myhash$ & tot \\
\hline
\verb+Baum_Sweet+ & $2$ & $2$ & $1$ & BP & $1$ & $0.04$ & $0.00$ & $0.01$ & $0$ & $0$ & $0.06$ \\
 & & & & IP & $1$ & $0.04$ & $0.00$ & $0.02$ & $0$ & $0$ & $0.07$ \\
\verb+Rudin_Shapiro+ & $2$ & $2$ & $1$ & BP & $2$ & $0.04$ & $0.00$ & $0.01$ & $0$ & $0$ & $0.07$ \\
 & & & & IP & $1$ & $0.00$ & $0.00$ & $0.02$ & $0$ & $0$ & $0.08$ \\
\verb+no_2s_in_3_exp+ & $2$ & $2$ & $4$ & BP & $6$ & $0.06$ & $0.01$ & $0.03$ & $2$ & $1$ & $0.11$ \\
 & & & & IP & $4$ & $0.01$ & $0.00$ & $0.04$ & $2$ & $1$ & $0.12$ \\
\verb+Stern_Brocot_b2+ & $2$ & $2$ & $4$ & BP & $4$ & $0.08$ & $0.01$ & $0.02$ & $2$ & $1$ & $0.13$ \\
 & & & & IP & $2$ & $0.01$ & $0.00$ & $0.03$ & $1$ & $1$ & $0.12$ \\
\verb+Stern_Brocot_b4+ & $4$ & $2$ & $26$ & BP & $57$ & $19$ & $2.2$ & $1.3$ & $2$ & $1$ & $23$ \\
 & & & & IP & $30$ & $4.7$ & $0.00$ & $0.53$ & $1$ & $1$ & $5.4$ \\
\verb+Dilcher_Stolarsky+ & $4$ & $2$ & $4$ & BP & $3$ & $0.06$ & $0.01$ & $0.02$ & $0$ & $0$ & $0.11$ \\
 & & & & IP & $1$ & $0.01$ & $0.00$ & $0.02$ & $0$ & $0$ & $0.09$ \\
\verb+Katz_Linden+ & $2$ & $4$ & $14$ & BP & $54$ & $4.2$ & $1.1$ & $12$ & $0$ & $0$ & $17$ \\
 & & & & IP & $8$ & $0.45$ & $0.00$ & $1.5$ & $0$ & $0$ & $2.1$ \\
\verb+Adamczewski_Faverjon+ & $3$ & $4$ & $258$ & BP & $168$ & $1307$ & $275$ & $20792$ & $12$ & $3$ & $22412$ \\
 & & & & IP & $92$ & $67$ & $0.00$ & $475$ & $11$ & $3$ & $543$ \\
\hline
\verb+lclm_3rat_1log+ & $3$ & $3$ & $121$ & BP & $1628$ & $1693$ & $222$ & $7754$ & $21$ & $2$ & $9702$ \\
 & & & & IP & $116$ & $62$ & $0.00$ & $140$ & $5$ & $2$ & $203$ \\
\verb+lclm_3rat_2log+ & $3$ & $3$ & $122$ & BP & $1598$ & $1438$ & $184$ & $7762$ & $21$ & $2$ & $9411$ \\
 & & & & IP & $116$ & $67$ & $0.00$ & $147$ & $5$ & $2$ & $215$ \\
\verb+lclm_2rat_trunc_sl0+ & $2$ & $4$ & $56$ & BP & $1580$ & $3339$ & $391$ & $1432$ & $19$ & $2$ & $5263$ \\
 & & & & IP & $653$ & $270$ & $0.00$ & $219$ & $14$ & $2$ & $490$ \\
\verb+lclm_2rat_trunc_sl1+ & $2$ & $4$ & $61$ & BP & $2069$ & $7463$ & $1052$ & $2595$ & $9$ & $2$ & $11346$ \\
 & & & & IP & $915$ & $468$ & $0.00$ & $358$ & $8$ & $2$ & $828$ \\
\hline
\verb+dft_Baum_Sweet+ & $4$ & $2$ & $6$ & BP & $5$ & $0.08$ & $0.02$ & $0.06$ & $0$ & $0$ & $0.18$ \\
 & & & & IP & $1$ & $0.01$ & $0.00$ & $0.03$ & $0$ & $0$ & $0.10$ \\
\verb+dft_Rudin_Shapiro+ & $4$ & $2$ & $7$ & BP & $941$ & $83$ & $13$ & $18$ & $0$ & $0$ & $152$ \\
 & & & & IP & $81$ & $3.0$ & $0.00$ & $2.3$ & $0$ & $0$ & $5.8$ \\
\verb+dft_Stern_Brocot_b2+ & $4$ & $2$ & $24$ & BP & $50$ & $12$ & $1.0$ & $0.87$ & $2$ & $1$ & $15$ \\
 & & & & IP & $21$ & $2.2$ & $0.00$ & $0.57$ & $1$ & $1$ & $3.0$ \\
\verb+dft_no_2s_in_3_exp+ & $4$ & $2$ & $20$ & BP & $98$ & $23$ & $2.6$ & $1.7$ & $2$ & $1$ & $29$ \\
 & & & & IP & $80$ & $8.1$ & $0.00$ & $1.3$ & $2$ & $1$ & $9.6$ \\
\verb+dft_Dilcher_Stolarsky+ & $16$ & $2$ & $50$ & BP & & & & & & & \timelimitexceeded
 \\
 & & & & IP & $2025$ & $1064$ & $0.00$ & $2158$ & $0$ & $0$ & $3382$ \\
\verb+dft_Stern_Brocot_b4+ & $16$ & $2$ & $348$ & BP & & & & & & & \timelimitexceeded
 \\
 & & & & IP & $1528$ & $26096$ & $0.00$ & $3481$ & $1$ & $1$ & $29670$ \\

\hline
\end{tabular}}

\bigskip

\caption{\label{tab:Basic-vs-Improved-GP}Comparison of both variants of our Mahlerian analogue of Petkovšek's method given by Algorithms~\ref{algo:BP} and~\ref{algo:IP} for the search of rational solutions, with~${\bL = \bQ}$. See Table~\ref{tab:Basic-vs-Improved-GP:legend} for the meaning of columns. All times measured in seconds.}
\end{small}
\end{table}

\begin{table}
\begin{small}
\rule{\textwidth}{.04em}
\begin{itemize}
\item `var'~stands for the used variant, `BP' (basic) or `IP' (improved).
\item `tpl'~counts the triples $(B,A,\zeta)$ considered (and dealt with) by the loops.
\item `anc'~is the sum over~$(B,A)$ of the times to compute the ancillary operators~$\tilde L$.
\item `cfs'~is, in the basic variant, the sum over~$(B,A)$ of the times to compute the sets~$Z(\tilde L)$ of coefficients~$\zeta$ in solutions~$u$; it reduces to the single calculation of~$Z(L)$ in the improved variant.
\item `pol'~is the sum over~$(B,A,\zeta)$ of the times to solve for the polynomials~$C$.
\item `$\myhash$'~is the cumulative count over~$(B,A,\zeta)$ of obtained parametrized solutions~$u$, with possible redundancy.
\item `$\check\myhash$'~counts the number of parametrizations of solutions~$u$ retained after removing repeated and embedded parametrizations.
\item `tot'~is the total time of the ramified rational solving.
\end{itemize}
\rule{\textwidth}{.04em}
\caption{\label{tab:Basic-vs-Improved-GP:legend}Meaning of the columns in Table~\ref{tab:Basic-vs-Improved-GP}.}
\end{small}
\end{table}

Table~\ref{tab:Basic-vs-Improved-GP} compares both variants of our Mahlerian analogue of Petkovšek's method.
It shows the speedup of the improved Algorithm~\ref{algo:IP},
which uses the various prunings and optimizations
discussed in~\S\ref{sec:efficiency-improvements},
over the basic Algorithm~\ref{algo:BP}.
All those calculations were obtained for the simpler case of the field~$\bL = \bQ$.
The reduction of the number of triples $(B,A,\zeta)$ considered by the algorithms
(typically by a factor in the range $2$--$10$ in our list of examples)
induces part of the speedup.
In addition, we observe that computing the ancillary operators~$\tilde L$ (see~\eqref{eq:mahler-necessary-2})
and the corresponding sets~$Z(\tilde L)$ (see Definition~\ref{def:upper-newton})
takes a significant part of the time in some of the runs of the basic variant
(more than half of the total time in cases like \verb+Stern_Brocot_b4+ and \verb+dft_Rudin_Shapiro+),
and this time is saved by the improved variant.
For computations of nonnegligible times,
the overall speedup is typically of a few units or dozens of units.
Because we do not expect any direct impact of the nature of the field~$\bL \subseteq \bKbar$
on the applicability of the pruning rules to avoid redundant pairs (see~\S\ref{sec:avoid-redundant-pairs}),
we speculate that similar significant improvements would also occur
for calculations with a more general number field~$\bL$.

An example will suggest that Algorithms \ref{algo:BP} and~\ref{algo:IP} are slower
if we consider~${\bL = \bQbar}$ instead of~$\bL = \bQ$,
because the introduction of absolute factorizations of the polynomials $\ell_0$ and~$\ell_r$
induces more factors and exponentially more divisors~$(A,B)$ to test.
To this end, we consider~\texttt{Stern\_Brocot\_b4}.
When we change
$\bL = \bQ$ to~$\bL = \bQbar$,
the ramification bound~\eqref{eq:q-L}, which generally could increase, is unchanged: $q_{\bQbar} = q_\bQ = 1$.
However,
the number of factors of $\ell_0 = (2x^8 + x^4 + 1) x^3$ increases from~$4$ to~$11$
(when counted with multiplicities),
and that of $\ell_2 = (2x^2 + x + 1) (x^8 - x^4 + 1) (x^2 + x + 1)^2 (x^2 - x + 1)^2 (x^4 - x^2 + 1)^2$ from~$8$ to~$26$,
so that the number of pairs~$(A,B)$ to be potentially tested is changed
from~$2\cdot4\cdot2\cdot2\cdot3\cdot3\cdot3 = 864$
to~$2^8\cdot4\cdot2^2\cdot2^8\cdot3^2\cdot3^2\cdot3^4 = 6879707136$:
even if the filtering coprimality condition makes
the count of pairs effectively leading to a calculation
be smaller (that is, the value~$57$ of~`tpl' is smaller than~$864$),
we expect that many more pairs are used over~$\bQbar$.

\begin{table}
\begin{small}
\centerline{%
\begin{tabular}{|@{\,}r@{\,}|@{\,}r@{\,}r@{\,}r@{\,}|@{\,}r@{\,}|@{\,}r@{\,}r@{\,}r@{\,}r@{\,}r@{\,}r@{\,}r@{\,}|}
\hline
\multicolumn{1}{|@{\,}c@{\,}|@{\,}}{} & \multicolumn{3}{@{\,}c@{\,}|@{\,}}{} & \multicolumn{1}{@{\,}c@{\,}|@{\,}}{IP} & \multicolumn{7}{@{\,}c@{\,}|}{HP} \\
example & $b$ & $r$ & $d$ & tot & fst & dim & $\sigma$ & rfn & syz & sng & tot \\
\hline
\verb+Baum_Sweet+ & $2$ & $2$ & $1$ & $0.07$ & $0.07$ & $(1, 1)$ & $(6, 6)$ & $0.03$ & $0.03$ & $\text{-}$ & $0.13$ \\
\verb+Rudin_Shapiro+ & $2$ & $2$ & $1$ & $0.08$ & $0.07$ & $(1, 0)$ & $(6, \text{-})$ & $0.02$ & $0.01$ & $\text{-}$ & $0.10$ \\
\verb+no_2s_in_3_exp+ & $2$ & $2$ & $4$ & $0.12$ & $0.08$ & $(1, 1)$ & $(33, 9)$ & $0.03$ & $0.08$ & $\text{-}$ & $0.21$ \\
\verb+Stern_Brocot_b2+ & $2$ & $2$ & $4$ & $0.12$ & $0.07$ & $(1)$ & $(21)$ & $0.01$ & $0.02$ & $\text{-}$ & $0.12$ \\
\verb+Stern_Brocot_b4+ & $4$ & $2$ & $26$ & $5.4$ & $0.08$ & $(1)$ & $(63)$ & $0.02$ & $0.11$ & $\text{-}$ & $0.22$ \\
\verb+Dilcher_Stolarsky+ & $4$ & $2$ & $4$ & $0.09$ & $0.07$ & $(2)$ & $(27)$ & $0.04$ & $0.08$ & $0.02$ & $0.23$ \\
\verb+Katz_Linden+ & $2$ & $4$ & $14$ & $2.1$ & $0.12$ & $(0, 1, 0, 0)$ & $(\text{-}, 69, \text{-}, \text{-})$ & $0.06$ & $0.39$ & $\text{-}$ & $0.57$ \\
\verb+Adamczewski_Faverjon+ & $3$ & $4$ & $258$ & $543$ & $0.16$ & $(4)$ & $(163)$ & $0.32$ & $1.8$ & $0.05$ & $2.4$ \\
\hline
\verb+lclm_3rat_1log+ & $3$ & $3$ & $121$ & $203$ & $0.08$ & $(3)$ & $(140)$ & $0.16$ & $2.5$ & $0.03$ & $2.9$ \\
\verb+lclm_3rat_2log+ & $3$ & $3$ & $122$ & $215$ & $0.09$ & $(2, 1)$ & $(88, 52)$ & $0.07$ & $0.51$ & $\text{-}$ & $0.71$ \\
\verb+lclm_2rat_trunc_sl0+ & $2$ & $4$ & $56$ & $490$ & $0.11$ & $(4)$ & $(294)$ & $2.6$ & $12$ & $0.05$ & $14$ \\
\verb+lclm_2rat_trunc_sl1+ & $2$ & $4$ & $61$ & $828$ & $0.12$ & $(4)$ & $(519)$ & $13$ & $104$ & $0.05$ & $117$ \\
\verb+lclm_3rat_trunc_sl1+ & $3$ & $5$ & $1260$ & \timelimitexceeded
 & $0.49$ & $(3, 2)$ & $(574, 268)$ & $11$ & $51$ & $0.07$ & $63$ \\
\verb+lclm_4pow_b2+ & $2$ & $7$ & $107$ & $25351$ & $0.20$ & $(1, 4)$ & $(429, 739)$ & $0.16$ & $2.4$ & $\text{-}$ & $2.8$ \\
\verb+lclm_4pow_b3+ & $3$ & $6$ & $727$ & \timelimitexceeded
 & $0.56$ & $(1, 4)$ & $(108, 174)$ & $0.47$ & $0.64$ & $\text{-}$ & $1.7$ \\
\verb+lclm_4pow_b4+ & $4$ & $5$ & $989$ & \timelimitexceeded
 & $0.23$ & $(4)$ & $(223)$ & $0.40$ & $0.59$ & $\text{-}$ & $1.4$ \\
\verb+lclm_4pow_b5+ & $5$ & $5$ & $3103$ & \timelimitexceeded
 & $2.0$ & $(1, 4)$ & $(44, 289)$ & $2.8$ & $0.94$ & $\text{-}$ & $5.9$ \\
\verb+lclm_5pow_b4+ & $4$ & $7$ & $17270$ & \memorylimitexceeded
 & $39$ & $(1, 5)$ & $(274, 1326)$ & $64$ & $6.5$ & $\text{-}$ & $115$ \\
\hline
\verb+dft_Baum_Sweet+ & $4$ & $2$ & $6$ & $0.10$ & $0.08$ & $(2)$ & $(77)$ & $0.06$ & $0.18$ & $0.02$ & $0.37$ \\
\verb+dft_Rudin_Shapiro+ & $4$ & $2$ & $7$ & $5.8$ & $0.06$ & $(1, 0)$ & $(88, \text{-})$ & $0.03$ & $0.15$ & $\text{-}$ & $0.25$ \\
\verb+dft_Stern_Brocot_b2+ & $4$ & $2$ & $24$ & $3.0$ & $0.09$ & $(1)$ & $(59)$ & $0.03$ & $0.10$ & $\text{-}$ & $0.22$ \\
\verb+dft_no_2s_in_3_exp+ & $4$ & $2$ & $20$ & $9.6$ & $0.09$ & $(1, 1)$ & $(85, 33)$ & $0.07$ & $0.84$ & $\text{-}$ & $1.0$ \\
\verb+dft_Dilcher_Stolarsky+ & $16$ & $2$ & $50$ & $3382$ & $0.10$ & $(2)$ & $(666)$ & $0.25$ & $3.7$ & $\text{-}$ & $4.1$ \\
\verb+dft_Stern_Brocot_b4+ & $16$ & $2$ & $348$ & $29670$ & $0.13$ & $(1)$ & $(239)$ & $0.14$ & $2.0$ & $\text{-}$ & $2.4$ \\
\hline
\verb+rmo_2_1+ & $2$ & $3$ & $19$ & $5.3$ & $0.07$ & $(3)$ & $(263)$ & $1.1$ & $23853$ & $0.03$ & $23854$ \\
\verb+rmo_3_1+ & $3$ & $3$ & $37$ & $14$ & $0.07$ & $(3)$ & $(133)$ & $0.22$ & $1166$ & $0.03$ & $1167$ \\
\verb+rmo_2_2+ & $2$ & $3$ & $44$ & $15$ & & & & & & & \timelimitexceeded
 \\
\verb+rmo_3_2+ & $3$ & $3$ & $82$ & $39$ & $0.08$ & $(3)$ & $(247)$ & $2.6$ & $11031$ & $0.03$ & $11034$ \\
\verb+rmo_2_3+ & $2$ & $3$ & $69$ & $26$ & & & & & & & \timelimitexceeded
 \\
\verb+rmo_3_3+ & $3$ & $3$ & $127$ & $70$ & & & & & & & \timelimitexceeded
 \\
\verb+rmo_2_4+ & $2$ & $3$ & $94$ & $41$ & & & & & & & \timelimitexceeded
 \\
\verb+rmo_3_4+ & $3$ & $3$ & $172$ & $109$ & & & & & & & \timelimitexceeded
 \\
\verb+rmo_2_5+ & $2$ & $3$ & $119$ & $58$ & & & & & & & \timelimitexceeded
 \\
\verb+rmo_3_5+ & $3$ & $3$ & $217$ & $166$ & & & & & & & \timelimitexceeded
 \\

\hline
\end{tabular}}

\bigskip

\caption{\label{tab:Improved-GP-vs-Iterated-HP}Comparison of the improved Mahler analogue of Petkovšek's method (Algorithm~\ref{algo:IP}, with~$\bL = \bQ$) and the Hermite--Padé approach (Algorithm~\ref{algo:HP}, over~$\bQbar$). See Table~\ref{tab:Improved-GP-vs-Iterated-HP:legend} for the meaning of columns. All times measured in seconds.}
\end{small}
\end{table}

\begin{table}
\begin{small}
\rule{\textwidth}{.04em}
\begin{itemize}
\item `tot'~is the total time for ramified rational solving using the improved Mahler analogue of Petkovšek's approach (IP) or the Hermite--Padé approach (HP).
\item `fst'~is the time for a first series computation, sufficient to determine the dimensions of series-solutions spaces behind the various logarithmic parts in solutions, provided in the column `dim'.
\item `dim'~is a list, indexed by the $\lambda \in \Lambda$, of the dimension of series appearing in front of~$\mylog\lambda$ in solutions.
\item `$\sigma$'~is a list with same indexing of the last value of~$\sigma$ used to find the hypergeometric series solutions of~$L_\lambda$ (or~`-' when the dimension for~$\lambda$ is~$0$).
\item `rfn'~is the cumulative time over~$\lambda$ for all refined series computations up to the corresponding final approximation orders in~`$\sigma$'.
\item `syz'~is the total time for computing minimal bases.
\item `sng'~is the cumulative time over~$\lambda$ for all prime decompositions computed by calling Singular, or `-'~if no prime decomposition was needed for the operator~$L$.
\end{itemize}
\rule{\textwidth}{.04em}
\caption{\label{tab:Improved-GP-vs-Iterated-HP:legend}Meaning of the columns in Table~\ref{tab:Improved-GP-vs-Iterated-HP}.}
\end{small}
\end{table}

Table~\ref{tab:Improved-GP-vs-Iterated-HP} compares
the improved Mahler analogue of Petkovšek's method by Algorithm~\ref{algo:IP} over~$\bL = \bQ$
with the Hermite--Padé approach by Algorithm~\ref{algo:HP}
(which by design is necessarily over~$\bQbar$).
Even though Algorithm~\ref{algo:HP} computes a more complete solution set,
it is by far the faster algorithm,
at least if we exclude the special examples \verb+rmo_+$b$\verb+_+$\delta$.
Recognizing and extracting the subset of solutions in~$\bQ$ from the complete solution set in~$\bQbar$
could be done easily,
so solving the Riccati equation over~$\bL = \bQ$
reduces to solving it over~$\bL = \bQbar$:
this makes Algorithm~\ref{algo:HP} be the better algorithm.
For longer calculations,
speedups of Algorithm~\ref{algo:HP} over Algorithm~\ref{algo:IP} can be very high
(e.g., \verb+Adamczewski_Faverjon+, \verb+dft_Dilcher_Stolarsky+, \verb+dft_Stern_Brocot_b4+).

Examples~\verb+rmo_+$b$\verb+_+$\delta$ were constructed in~Example~\ref{ex:random-eqns}
to have solutions in~$\bQ(x)$
and show a situation where Algorithm~\ref{algo:HP} fails:
for example, for~\verb+rmo_3_3+, the polynomials $\ell_0$ and~$\ell_3$
have few factors but factors of large degrees;
specifically, they have the following factorization patterns:
\begin{equation*}
\ell_0 = x^{10} (x^3+\dotsb) (x^{21}+\dotsb) (x^{93}+\dotsb) ,
\quad
\ell_3 = (x^{27}+\dotsb) (x^{31}+\dotsb) (x^{69}+\dotsb) .
\end{equation*}
This makes Algorithms \ref{algo:BP} and~\ref{algo:IP} iterate over few pairs~$(A,B)$.
In contrast,
even if the degree bound (equations \eqref{eq:bound-P} and~\eqref{eq:bound-Q})
and the induced~$\sigma$ needed in Algorithm~\ref{algo:HP}
are not too large, of the order of a few hundreds,
the sequences of coefficients of the series solutions
are not automatic (see~\S\ref{sec:gen-fun})
and involve very large numbers,
which dramatically slows down the calculation.

\section{Differential transcendence of Mahler functions}
\label{sec:transcendence}

This section shows an application of algorithms
for solving the Riccati Mahler equation for its solutions that are rational functions.
It uses such algorithms as a black box
and is otherwise completely independent from the rest of the text.

Mahler equations originate in number theory,
where they were introduced by Mahler
as he developed his eponymous method
to construct new transcendental numbers;
see \parencite{Adamczewski-2017-MM} for a recent survey.
Consider again the linear Mahler equation~\eqref{eq:linear},
this time with~$\bK = \bQbar$,
as well as some series solution $f\in \bQbar[[x]]$.
It is classical \parencite{Nishioka-1996-MFT}
that $f$~has a positive radius of convergence, that it can be extended
to a meromorphic function on the open unit disk,
and that except if it is rational,
it has a natural boundary on the unit circle
and is therefore transcendental.

Concerning values of~$f$,
\textcite{philippon2015groupes} proved that for all $\alpha\in\bQbar$
satisfying  $|\alpha|<1$
and such that $\alpha^{b^{\bN}}$ does not intersect the zero set of~$\ell_0 \ell_r$,
the algebraic relations
between $f(\alpha),\dots, f(\alpha^{b^{r-1}})$ over~$\bQbar$
are specialization to $x=\alpha$ of algebraic relations
between the functions $f,\dots, M^{r-1}f$ over $\bQbar (x)$.
Therefore, if the latter functions have no algebraic relations,
their values at most algebraic points are algebraically independent.
As similar statement is known as the Hermite--Lindermann theorem,
which states that for any nonzero algebraic number~$\alpha$,
the value~$\exp(\alpha)$ of the exponential function is transcendental,
and of its generalization to $E$-functions; see \parencite{Beukers-2006-RVS} for a modern treatment.

Let~$\partial$ denote the derivation with respect to~$x$.
The algebraic relations between $f(\alpha),\dots, \partial^{n}f (\alpha)$
can be studied by the same approach,
at least for convenient $\alpha \in \bQbar$
\parencite[Theorem~1.5]{AdamczewskiDreyfusHardouin-2021-HLD}.
The result is
that such relations come from specializations of algebraic relations between $f$ and its derivatives.

The question of algebraic relations between values therefore motivates
the question of determining whether $f$~is \emph{differentially transcendental} (Definition~\ref{def:diff-trans})
and more generally whether $f$ and its iterates under~$M$ are \emph{differentially algebraically independent}.
Several results based on difference Galois theory have been developed,
leading to effective criteria.
We summarize this now.

\begin{defi}\label{def:diff-trans}
Let $f$, $f_1$, \dots, $f_m$ be series in~$\fps\bQbar$.

We say that $f$~is \emph{differentially algebraic}
when there exist $n\in \bN$ and a nonzero $P\in \bQbar[x][X_{0},\dots,X_{n}]$
such that $P(f,\partial f, \dots, \partial^nf)=0$.
We say that $f$~is \emph{differentially transcendental} otherwise.

We say that $(f_1,\dots,f_m)$~is \emph{differentially algebraically dependent}
when there exist $n\in \bN$ and a nonzero $P\in \bQbar[x][(X_{i,j})_{1\leq i\leq m, \ 0\leq j\leq n}]$
such that
\[ P(f_1,\partial f_1, \dots, \partial^nf_1, \dots, f_m,\partial f_m, \dots, \partial^nf_m)=0 . \]
We say that $(f_1,\dots,f_m)$~is \emph{differentially algebraically independent} otherwise.
\end{defi}

Let $G$ denote the difference Galois group of~\eqref{eq:linear}.
This is an algebraic group.
In a way that reinforces the dichotomy between rational and transcendental solutions of~\eqref{eq:linear},
\textcite{DreyfusHardouinRoques-2018-HSM} proved that
if $G$~contains $\mathrm{SL}_r(\bQbar)$,
which implies that \eqref{eq:linear}~has no nonzero rational solutions,
then the nonzero series solutions are differentially transcendental%
\footnote{In \textcite{DreyfusHardouinRoques-2018-HSM}, the series have coefficients in~$\bC$
but everything remains correct
if we replace $\bC$ by the algebraically closed field~$\bQbar$.}
 (not just transcendental).
More recently,
\textcite{AdamczewskiDreyfusHardouin-2021-HLD} proved that
a solution~$f$ is differentially transcendental
unless it is a rational function $f\in \bQbar(x)$.
Thus,
to prove the differential transcendence of~$f$,
it is sufficient to check that it is not rational.
To this end, our Algorithm~9
in~\parencite{ChyzakDreyfusDumasMezzarobba-2018-CSL},
or for that matter our new development in~\S\ref{sec:back-to-rat},
can be used.

Concerning the differentially algebraic independence of the~$M^if$ for a Mahler function~$f$ that solves~\eqref{eq:linear},
the best we can expect is that $(f,\dots,M^{r-1}f)$~is differentially algebraically independent,
since by~\eqref{eq:linear},
the functions $f,\dots, M^rf$ are linearly dependent over~$\rat\bQbar$.
\textcite{DreyfusHardouinRoques-2018-HSM} also proved that
if \emph{(i)}~the difference Galois group of \eqref{eq:linear} contains $\mathrm{SL}_r(\bQbar)$and \emph{(ii)}~$\ell_0/\ell_r$, which up to sign is the determinant of the companion matrix of~$L$, is a monomial,
then
$f,\dots, M^{r-1} f$ are differentially algebraically independent.
\textcite[paragraph just after the proof of Theorem~5.2]{ArrecheSinger-2017-GGI}%
\footnote{We are indebted to an anonymous reviewer for pointing us to \parencite{ArrecheSinger-2017-GGI}.
This allowed us to prove the algebraic independence of $f$ and~$Mf$
for all six natural examples considered in the article,
instead of just three as in our earlier draft.}
later explained how the assumption on~$\ell_0/\ell_r$ can be avoided.

For general~$r$, the previous criteria based on testing the inclusion of~$\mathrm{SL}_r(\bQbar)$ into~$G$ are not practical:
although an algorithm exists for computing the difference Galois group
\parencite{Feng-2018-CGG},
it is too theoretical to work in practice.
When $r=2$, an efficient, specialized criterion for the inclusion can be formulated
in terms  of solutions of Riccati equations.
This is provided by the following theorem, whose proof is implicit in \parencite{Roques-2018-ARB}.

\begin{thm}[{\textcite[\S6]{Roques-2018-ARB}}]
For~$r=2$,
assume the existence of a nonzero solution\/ $f\in \bQbar[[x]]$ of~\eqref{eq:linear}.
Assume further the condition $\ell_1 \neq 0$.
Then, the difference Galois group of~\eqref{eq:linear} contains~$\mathrm{SL}_2(\bQbar)$
if and only if neither of the equations
\begin{align}
\label{eq:r1} \ell_2 u M u + \ell_1 u + \ell_0 &= 0 , \\
\label{eq:r2} u M^{2} u + \left(M^{2}\left(\frac{\ell_0}{\ell_1} \right)-M\left( \frac{\ell_1}{\ell_2}\right) + \frac{\ell_2}{\ell_1}M\left( \frac{\ell_0}{\ell_2}\right)\right)u
+ \frac{\ell_2 \ell_0 M \ell_0}{\ell_1^2 M \ell_1} &= 0
\end{align}
has any solution in~$\ramrat{\bQbar}$.

\end{thm}

\begin{proof}
Again, let $G$ denote the difference Galois group of~\eqref{eq:linear},
and let $G^0$ be the connected component of the identity in~$G$.
Lemma~40 in \parencite{Roques-2018-ARB} shows that \eqref{eq:r1}~has no solution if and only if $G$~is irreducible.
Theorem~42 in the same reference, where our~$\ell_1$ is denoted~$a$, then shows that,
if \eqref{eq:r1}~has no solution and~$\ell_1\neq0$,
then \eqref{eq:r2}~has no solution
if and only if $G$~is not imprimitive.
By \parencite[Theorem~4]{Roques-2018-ARB}, which is a direct transposition of \parencite[Prop.~1.20]{PutSinger-1997-GTD},
$G/G^0$~is finite and cyclic\footnote{In both references, “cyclic” means “finite and cyclic”.}.
The classical classification of algebraic groups then finishes the proof.
\end{proof}

\begin{rem}
The condition on~$\ell_1$ is no major restriction:
if~$\ell_1 = 0$, \eqref{eq:linear}~is first-order with respect to the Mahler operator
with respect to the radix~$b^2$,
and the rationality of~$f$ can be studied directly.
It is also worth mentioning that
\eqref{eq:r2}~can be viewed as a Riccati equation in radix~$b^2$.
\end{rem}

The Galoisian criterion in \parencite{ArrecheSinger-2017-GGI}
and the previous theorem
straightforwardly combine into the following corollary.
\begin{cor}\label{cor:criterion}
For~$r=2$,
assume the existence of a nonzero solution\/ $f\in \bQbar[[x]]$ of~\eqref{eq:linear}.
Assume further the condition $\ell_1 \neq 0$ and that \eqref{eq:r1} and~\eqref{eq:r2} have no solutions in~$\ramrat{\bQbar}$.
Then,
$f$ and~$Mf$ are differentially algebraically independent, and in particular $f$~is differentially transcendental.
\end{cor}

We tested our implementation on the six equations of order~$r = 2$ appearing in Example~\ref{ex:classical} or in~\S\ref{sec:examples}:
\texttt{Baum\_Sweet}, \texttt{Rudin\_Shapiro},
\texttt{Stern\_Brocot\_b2}, \texttt{Stern\_Brocot\_b4},
\texttt{no\_2s\_in\_3\_exp}, and \texttt{Dilcher\_Stolarsky}.
Algorithm~\ref{algo:HP} established
that the hypothesis of Corollary~\ref{cor:criterion} is satisfied for all six in a total of 8.4~seconds:
a total of 1.0~seconds for~\eqref{eq:r1};
a total of 8.3~seconds for~\eqref{eq:r2}.
In the worst case, solving of~\eqref{eq:r2} takes 18~times as long as the solving of~\eqref{eq:r1}.
We also executed Algorithm~\ref{algo:IP},
which could also handle all six equations,
in much more time (9.4~hours in total).
The application of Corollary~\ref{cor:criterion} therefore
proves the differential algebraic independence of~$\{f,Mf\}$
in the six cases.

For the examples \texttt{Baum\_Sweet} and \texttt{Rudin\_Shapiro},
we recover the results of differential algebraic independence
obtained by \textcite{DreyfusHardouinRoques-2018-HSM}
when they combined their criterion
with the determination of the difference Galois groups of those examples
by \textcite{Roques-2018-ARB}.
To the best of our knowledge, the results for the other examples are new
and were not easily accessible by hand calculations:
for instance, the degree~$d = 50$ and the radix~$b = 16$
of the equation~\eqref{eq:r2} for the example \texttt{Dilcher\_Stolarsky}
would lead to arduous calculations.

\section*{Acknowledgements}

We would like to express our sincere gratitude to the reviewers
for their constructive feedback and valuable suggestions,
which have greatly contributed to improving the presentation and clarity of this article.

\printbibliography
\end{document}